\documentclass[a4paper,11pt,twoside]{article}

\usepackage{graphicx,color}

\usepackage[latin1]{inputenc}
\usepackage{amsmath}
\usepackage{amsfonts}
\usepackage{amssymb}

\usepackage{a4wide}

\usepackage[thmmarks,amsmath,amsthm]{ntheorem}

\usepackage[mathscr]{eucal}

\usepackage{index}

\usepackage{fancyhdr}

\usepackage{times}

\usepackage{tikz}
\usetikzlibrary{arrows}

\numberwithin{equation}{section}

\theoremstyle{nonumberplain}

\theoremstyle{plain}
\newtheorem{thm}{Theorem}[section]

\newtheorem{lemma}[thm]{Lemma}
\newtheorem{cor}[thm]{Corollary}

\newtheorem{observation}[thm]{Observation}

\theoremstyle{definition}

\newtheorem{example}[thm]{Example}

\theoremstyle{remark}
\newtheorem{claim}{Claim}

\theoremsymbol{\ensuremath{_\square}}
\renewtheorem*{proof}{Proof\setcounter{claim}{0}}
\theoremsymbol{\ensuremath{_\dashv}}
\newtheorem*{clproof}{Proof}

\renewcommand{\qed}{\mbox{}\hfill$\scriptstyle\square$\par\medskip}

\newcommand{\nc}[1]{\expandafter\newcommand\csname#1\endcsname}

\nc{gX}[1]{\nc{g#1}{\mathfrak #1}} 

\nc{aX}[1]{\nc{a#1}{\mathcal{#1}}} 

\gX G
\gX H
\gX C
\gX S
\gX A

\aX H
\aX G
\aX D
\aX {QG}
\aX S

\newcommand{\m}[1]{\mathcal{#1}}
\newcommand{\mset}[1]{\left\lbrace#1\right\rbrace} 

\newlength{\ppshort}
\newcommand{\idfunc}{\mathbf{id}}

\newcommand{\Nat}{ \mathbb{N}}
\newcommand{\Int}{ \mathbb{Z}}
\newcommand{\Real}{ \mathbb{R}}

\newcommand{\FP}{\textup{FP}}
\newcommand{\NP}{\textup{NP}}
\newcommand{\sharpP}{\#\textup{P}}
\newcommand{\scalp}[1]{\ensuremath{\langle #1\rangle}}

\newcommand{\Tle}{ \ensuremath{\le} }  
\newcommand{\Tequiv}{ \ensuremath{\equiv} }  
\newcommand{\rank}[1]{\ensuremath{\mathrm{rank}\, #1}} 

\newcommand{\sdef}[1]{\emph{#1}} 
\newcommand{\sdefi}[2]{\emph{#1}} 

\newcommand{\sdefisub}[3]{\emph{#1}} 

\newcommand{\sdefis}[3]{\emph{#1}} 
\newcommand{\idxsymb}[1]{#1} 
\newcommand{\idxsymbm}[1]{#1} 

\newcommand{\comment}[1]{}

\newcommand{\marc}[1]{\comment{{\bf Marc: }#1}}
\newcommand{\martin}[1]{\comment{{\bf Martin: }#1}}
\newcommand{\draftOK}[1]{}

\newcommand{\PP}{\textup{P}}

\newcommand{\cnt}{\textup{COUNT}}

\renewcommand{\vec}[1]{\mathbf{#1}}

\newcommand{\eval}{\textup{EVAL}}

\newcommand{\evalk}{\textup{EVAL}^{\textsf{pin}}}

\newcommand{\cntk}{\textup{COUNT}^{\textsf{pin}}}

\newcommand{\row}[1]{\ensuremath{_{#1,*}}}

\newcommand{\C}{\mathbb{C}}
\newcommand{\Q}{\mathbb{Q}}
\renewcommand{\Int}{\mathbb{Z}}

\usepackage{mathrsfs}

\usepackage{stmaryrd}

\newcommand{\twres}[1]{[#1]}

\newcommand{\vcfg}{\sigma} 

\newenvironment{condition}[1]
{\begin{list}{} {\setlength{\leftmargin}{40pt} 
                \setlength{\itemindent}{-6pt}}
\item[\textbf{#1}]} 
{\end{list}}

\newcommand{\cond}[1]{\textup{\textbf{(#1)}}}

\newcommand{\df}{\textup{dom}}  

\newcommand{\vpin}{\phi} 

\newcommand{\Qu}{\mathbb{Q}}

\newcommand{\wset}{\mathcal{W}} 

\newcommand{\vdeg}{d} 

\newcommand{\Ralg}{\mathbb{R}_{\mathbb{A}}}
\newcommand{\Calg}{\mathbb{C}_{\mathbb{A}}}
\newcommand{\Ring}{\mathbb S}

\begin{document}

\pagestyle{headings}




\title{{Counting Homomorphisms and Partition Functions}}
\author{Martin Grohe \\
Humboldt Universit\"at zu Berlin \\
Berlin, Germany
\and 
Marc Thurley\footnote{supported in part by Marie Curie Intra-European Fellowship 271959 at the Centre de Recerca Matem\`{a}tica, Bellaterra, Spain}\\
Centre de Recerca Matem\`{a}tica\\
Bellaterra, Spain
}

\maketitle 





\section{Introduction}
\emph{Homomorphisms} between relational
structures are not only fundamental mathematical objects, but are also
of great importance in an applied computational context. Indeed,
\emph{constraint satisfaction problems}, a wide class of algorithmic
problems that occur in many different areas of computer science such
as artificial intelligence or database theory, may be viewed as asking
for homomorphisms between two relational structures
\cite{fedvar98}. 
In a logical setting, homomorphisms may be viewed as witnesses for
\emph{positive primitive formulas} in a relational language.
As we shall see below, homomorphisms, or more precisely the numbers of
homomorphisms between two structures, are also related to a fundamental
computational problem of statistical physics. Homomorphisms of graphs
are generalizations of colorings, and for that reason a homomorphism from a
graph $G$ to a graph $H$ is also called an \emph{$H$-coloring} of
$G$. Note that if $H$ is the complete graph on $k$-vertices,
then an $H$-coloring of a graph $G$ may be viewed as a proper
$k$-coloring of $G$ in the usual graph theoretic sense that adjacent
vertices are not allowed to get the same color.

It is thus no surprise that the computational complexity of
various algorithmic problems related to homomorphisms, in particular
the \emph{decision problem} of whether a homomorphism between
two given structures exists and the \emph{counting problem} of determining the
number of such homomorphisms, have been intensely studied. (For the
decision problem, see, for
example, \cite{barkozniv09,bul06,buljeakro05,gro07,helnes90}. References
for the counting problem will be given in Section~\ref{sec:complexity}. Other related problems, such as optimization or
enumeration problems, have been studied, for example, in
\cite{aus07,bdgm09,djkk08,rag08,schnoschno07}.) 

In this article, we are concerned with the complexity of counting
homomorphisms from a given structure $A$ to a fixed structure
$B$. Actually, we are mainly interested in a
generalization of this problem to be introduced in the next section. We almost
exclusively focus on graphs. The first part of the article, consisting
of the following two sections, is a short survey of what is known
about the problem. In the second part, consisting of the remaining
Sections~\ref{sec:it}-\ref{sec:11c}, we give a proof of a theorem due to Bulatov
and the first author of this paper \cite{bulgro05}, which classifies
the complexity of partition functions described by matrices with
non-negative entries. The proof we give here is essentially the same as
the original one, with a few shortcuts due to \cite{thu09}, but it is
phrased in a different, more graph theoretical language that may make it
more accessible to most readers.

\section{From Homomorphisms to Partition Functions}
For a fixed graph $H$ we let $Z_H$ be the ``homomorphism-counting
function'' that maps each graph $G$ to the number of homomorphisms
from $G$ to $H$. Several well-known combinatorial graph invariants can
be expressed as homomorphism counting functions, as the following
examples illustrate:

\begin{figure}
  \centering
  \begin{tikzpicture}
  [
  line width=0.4mm,
  vertex/.style={draw,circle,inner sep=1pt,minimum
    size=2.5mm},
  every edge/.style={draw}
  ]
  \small
  \node[vertex] (i1) at (0,0) {$1$};
  \node[vertex] (i2) at (1.5,0) {$2$} edge (i1);
  \draw (i1) .. controls (-1,-1) and (-1,1) .. (i1); 
  \draw (0.5,-0.8) node {$I$};

  \node[vertex] (j1) at (4,0) {$1$};
  \node[vertex] (j2) at (5.5,0) {$2$};
  \draw (j1) .. controls (4.5,0.5) and (5,0.5) .. (j2);
  \draw (j1) .. controls (4.5,-0.5) and (5,-0.5) .. (j2);
  \draw (j1) .. controls (3,-1) and (3,1) .. (j1); 
  \draw (4.5,-0.8) node {$J$};

  \node[vertex] (k21) at (8,0) {$1$};
  \node[vertex] (k22) at (9.5,0) {$2$} edge (k21);
  \draw (8.75,-0.8) node {$K_2$};

  \node[vertex] (k31) at (12,0) {$1$};
  \node[vertex] (k32) at (13.5,0) {$2$} edge (k31);
  \node[vertex] (k33) at (12.75,1.2) {$3$} edge (k32) edge (k31);
  \draw (12.75,-0.8) node {$K_3$};
  
\end{tikzpicture}

  \caption{The graphs $I$, $J$, $K_2$, $K_3$}
  \label{fig:is-col}
\end{figure}
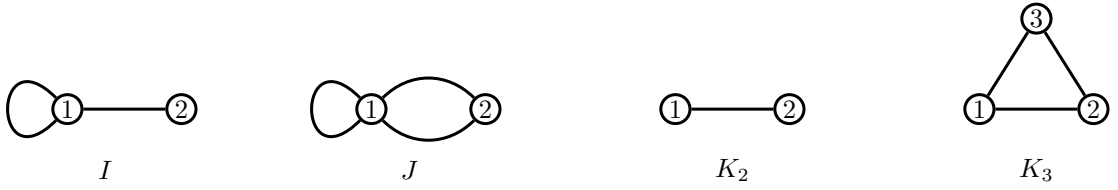

\begin{example}\label{exa:is}
  Let $I$ be the first graph displayed in Figure~\ref{fig:is-col}, and
  let $G$ be an arbitrary graph. Remember that an \emph{independent set} (or
  \emph{stable set}) of $G$ is set of pairwise nonadjacent vertices of
  $G$. For every set $S\subseteq
  V(G)$, we define a mapping $h_S:V(G)\to V(I)$ by $h_S(v)=2$ if $v\in S$ and $h_S(v)=1$ otherwise. Then $h_S$ is a
  homomorphism from $G$ to $I$ if and only if $S$ is an independent
  set. Thus the number $Z_I(G)$ of homomorphisms from $G$ to $I$ is
  precisely the number of independent sets of $G$.
\end{example}

\begin{example}\label{exa:col}
  For every positive integer $k$, let $K_k$ be the complete graph with
  vertex set $[k]:=\{1,\ldots,k\}$ (see Figure~\ref{fig:is-col}). Let
  $G$ be a graph. Recall that a \emph{(proper) $k$-coloring} of $G$
  is a mapping $h:V(G)\to[k]$ such that for all $vw\in E(G)$ it holds
  that $h(v)\neq h(w)$. Observe that a mapping $h:V(G)\to[k]$ is a
  $k$-coloring of $G$ if and only if it is a homomorphism from $G$ to
  $K_k$. Hence $Z_{K_k}(G)$ is the number of $k$-colorings of $G$.
\end{example}

\begin{figure}
  \centering
  $A(I)=
  \begin{pmatrix}
    1&1\\1&0
  \end{pmatrix}$
  \hspace{8mm}
  $A(J)=
  \begin{pmatrix}
    1&2\\2&0
  \end{pmatrix}$
  \hspace{8mm}
  $A(K_2)=
  \begin{pmatrix}
    0&1\\1&0
  \end{pmatrix}$
  \hspace{8mm}
  $A(K_3)=
  \begin{pmatrix}
    0&1&1\\1&0&1\\1&1&0
  \end{pmatrix}$
  \caption{The adjacency matrices of the graphs $I,K_2,k_3$}
  \label{fig:is-col-adj}
\end{figure}

\noindent
Unless mentioned otherwise, graphs in this article are undirected, and
they may have loops and parallel edges. 
Graphs without loops and 
parallel edges are called \emph{simple}. 
We always
assume the edge set and the vertex set of a graph to be disjoint. The class of all graphs is denoted by
$\aG$. A \emph{graph invariant} is a function defined on
$\aG$ that is invariant under isomorphism.  The \emph{adjacency
  matrix} of a graph $H$ is the square matrix $A:=A(H)$ with rows and
columns indexed by vertices of $H$, where the entry $A_{v,w}$ at row
$v$ and column $w$ is the number of edges from $v$ to
$w$. Figure~\ref{fig:is-col-adj} shows the adjacency matrices of the
graphs in Figure~\ref{fig:is-col}. For
all graphs $G,H$, we define a \emph{homomorphism} from $G\to H$ to be a
mapping $h:V(G)\cup E(G)\to V(H)\cup E(H)$ such that for all $v\in
V(G)$ it holds that $h(v)\in V(H)$ and for all edges $e\in E(G)$ with
endvertices $v,w$ it holds that $h(e)\in E(H)$ is an edge with
endvertices $h(v),h(w)$.\footnote{Usually, homomorphisms from $G$ to
  $H$ are defined to be mappings $g:V(G)\to V(H)$ that preserve
  adjacency. A mapping $g:V(G)\to V(H)$ is a homomorphism in this sense if and only
  if it has an extension $h:V(G)\cup E(G)\to V(H)\cup E(H)$ that is a
  homomorphism as defined above. Thus the two notions are closely
  related. However, if $H$ has parallel edges, then there a
  different numbers of homomorphisms for the two notions.}
The
following observation expresses a homomorphism counting function
$Z_H$ in
terms of the adjacency matrix of $H$:

\begin{observation}\label{obs:hom-par}
  Let $H$ be a graph and $A:=A(H)$. Then for every graph $G$,
  \begin{equation}\label{eq:hom-par}
  Z_H(G)=\sum_{\sigma:V(G)\to V(H)}\prod_{\substack{e\in E(G)\textup{ with}\\\textup{endvertices }v,w}}A_{\sigma(v),\sigma(w)}.
  \end{equation}
\end{observation}

\noindent
To simplify the notation, we write
$\displaystyle\prod_{\substack{vw\in E(G)}}$ instead of
$\displaystyle\prod_{\substack{e\in E(G)\text{
      with}\\\text{endvertices }v,w}}$ in similar expressions in the
following. By convention, the empty sum
evaluates to $0$ and the empty product evaluates to $1$. Thus
for the empty graph $\emptyset$ we have $Z_H(\emptyset)=1$ for all $H$ and 
$Z_\emptyset(G)=0$ for all $G\neq\emptyset$.

\begin{example}\label{exa:is}\martin{neu}
  Consider the second graph $J$ displayed in Figure~\ref{fig:is-col}, and
  let $G$ be an arbitrary graph. $Z_J(G)$ is a weighted sum over all
  independent sets of $G$: For all sets $S,T\subseteq
  V(G)$ we let $e(S,T)$ be the number of edges between $S$ and
  $T$. Then
  \[
  Z_J(G)=\sum_{\substack{S\subseteq V(G)\\\text{independent
        set}}}2^{e(S,V(G)\setminus S)}.
  \]
\end{example}

Equation~\eqref{eq:hom-par} immediately suggests the following
generalization of the homomorphism counting functions $Z_H$: For every
symmetric $n\times n$ matrix $A$ with entries from some ring $\Ring$
we
let $Z_A:\aG\to \Ring$ be the function that associates the following
element of $\Ring$ with each graph $G=(V,E)$:
\begin{equation}
\label{eq:part}
Z_A(G):=\sum_{\;\sigma:V\to [n]\;}\prod_{vw\in E}A_{\sigma(v),\sigma(w)}.
\end{equation}
We call functions $Z_A$, where $A$ is a symmetric matrix over $\Ring$,
\emph{partition functions} over $\Ring$. (All rings in this paper are commutative
with a unit. $\Ring$ always denotes a ring.) 

With each $n\times n$ matrix $A$ we associate a simple graph $H(A)$
with vertex set $[n]$ and edge set $\{ij\mid A_{i,j}\neq
0\}$. We may view $A$ as assigning nonzero weights to the edges of
$G(A)$, and we may view $Z_A(G)$ as a weighted sum of homomorphisms
from $G$ to $H(A)$, where the \emph{weight} of a mapping
$\sigma:V\to [n]$ is 
\[
\omega_A(G,\sigma):=\prod_{vw\in
  E}A_{\sigma(v),\sigma(w)}.
\]
Homomorphisms are precisely the mappings with nonzero weight. Inspired
by applications in statistical physics (see
Section~\ref{sec:stat-phys}), we often call the elements of the index
set of a matrix, usually $[n]$, \emph{spins}, and we call mappings
$\sigma:V\to [n]$ assigning a spin to each vertex of a graph
\emph{configurations}.

\begin{example}\label{exa:euler}
  Recall that a graph $G$ is Eulerian if there is a closed walk in $G$
  that traverses each edge exactly once. It is a well-known theorem,
  which goes back to Euler, that a graph is Eulerian if and only if it
  is connected and every vertex has even degree. Consider the matrix
  \[
  U=
  \begin{pmatrix}
    1&-1\\-1&1
  \end{pmatrix}.
  \]
  It is not hard to show that for every $N$-vertex graph $G$ we have $Z_{U}(G)=2^N$ if every vertex of $G$ has even
  degree and $Z_U(G)=0$ otherwise. Hence on connected graphs, $(1/2^{N})\cdot Z_{U}$ is the
  characteristic function of Eulerianicity.
\end{example}

\begin{example}\label{exa:evensubgraphs}
   Consider the matrix
  \[
  B=
  \begin{pmatrix}
    1&1\\1&-1
  \end{pmatrix}.
  \]
  Let $G$ be a graph. Then for every $\sigma:V(G)\to[2]$ it holds that 
  \[
  \omega_B(G,\sigma)
  =
  \begin{cases}
    1&\text{if the induced subgraph $G\big[\sigma^{-1}(2)\big]$ has an even
      number of edges},\\
    -1&\text{otherwise}.
  \end{cases}
  \]
  It follows that for every $N$-vertex graph $G$,
  \[
  \frac{1}{2}Z_{B}(G)+2^{N-1}
  \]
  is the number of
  induced subgraphs of $G$ with an even number of edges.
\end{example}

\begin{example}
  Recall that a \emph{cut} of a graph is a partition of its vertex set
  into two parts, and the \emph{weight} of a cut is the number of
  edges from one part to the other. A \emph{maximum cut} is a cut of
  maximum weight.
  Consider the matrix
  \[
  C:=
  \begin{pmatrix}
    1&X\\
    X&1
  \end{pmatrix}
  \]
  over the polynomial ring $\mathbb Z[X]$. It is not hard to see that
  for every graph $G$,
  the degree of the polynomial $Z_{C}(G)$ is the weight of a maximum
  cut of $G$ and the leading coefficient the number of maximum cuts. 
\end{example}

\noindent
  Graph polynomials present another important way to uniformly describe
  families of graph invariants. Examples of graph polynomials
  are the \emph{chromatic polynomial} and the \emph{flow
    polynomial}. Both of these are subsumed by the bivariate 
  \emph{Tutte polynomial}, arguably the most important graph
  polynomial. The following example exhibits a relation between the
  Tutte polynomial and partition functions. 

\begin{example}
Let $G=(V,E)$ be a graph with $N$ vertices, $M$
edges, and $Q$ connected components. For a subset $F\subseteq E$, by $q(F)$ we
denote the number of connected components of the graph $(V,F)$. 
The \emph{Tutte polynomial} of $G$ is the bivariate polynomial
$T(G;X,Y)$ defined by 
\[
T(G;X,Y)=\sum_{F\subseteq E}(X-1)^{q(F)-Q}\cdot(Y-1)^{|F|-N+q(F)}.
\]
It is characterized by the following
\emph{contraction-deletion equalities}. For an edge $e\in E$, we let $G\setminus
e$ be the graph obtained from $G$ by \emph{deleting} $e$, and we let $G/e$ be
the graph obtained from $G$ by \emph{contracting} $e$. A \emph{bridge} of $G$
is an edge $e\in E$ such that $G\setminus e$ has more connected components
than $G$. A \emph{loop} is an edge that is only incident to one vertex.
\[
T(G;X,Y)=
\begin{cases}
  1&\text{if }E(G)=\emptyset,\\
  X\cdot T(G\setminus e;X,Y)&\text{if }e\in E(G)\text{ is a bridge},\\
  Y\cdot T(G/e;X,Y)&\text{if }e\in E(G)\text{ is a loop},\\
  T(G\setminus e;X,Y)+T(G/e;X,Y)&\text{if }e\in E(G)\text{is neither a
    loop nor a bridge}. 
\end{cases}
\]
Let $r,s\in\mathbb C$. It can be shown that the partition function
of the $n\times
n$ matrix $A(n,r,s)$ with diagonal entries $r$ and
off-diagonal entries $s$ satisfies similar contraction-deletion
equalities, and this implies that it can be expressed in terms of
the Tutte polynomial as follows:
\begin{equation}
  \label{eq:tutte-part}
  Z_{A(n,r,s)}(G)= s^{M-N+Q}\cdot(r-s)^{N-Q}\cdot n^Q\cdot T\left(G;\frac{r+s\cdot(n-1)}{(r-s)},\frac{r}{s}\right).
  \end{equation}
This implies that for all $x,y\in\mathbb C$ such that
$n:=(x-1)\cdot(y-1)$ is a positive integer, it holds that
\[
  T\left(G;x,y\right)=(y-1)^{Q-N}\cdot n^{-Q}\cdot Z_{A(n,y,1)}(G).
  \]
\end{example}

\begin{example}\label{exa:flow}
  For simplicity, in this example we assume that $G=(V,E)$ is a
  \emph{simple} graph, that is, a graph without loops and without multiple edges. For every positive integer $k$, a \emph{$k$-flow} in
  $G$ is a
  mapping $f:V\to\mathbb Z_k$ (the group of integers modulo $k$) such
  that the following three conditions are
  satisfied:
  \begin{enumerate}
    \item[(i)] $f(v,w)=0$ for all $v,w\in V$ with $vw\not\in E$;
    \item[(ii)] $f(v,w)=-f(w,v)$ for all $v,w\in V$;
    \item[(iii)] $\displaystyle\sum_{\substack{w\in V\text{ with}\\vw\in E}}f(v,w)=0$ for
        all $v\in V$.
  \end{enumerate}
  The $k$-flow $f$ is \emph{nowhere zero} if $f(v,w)\neq 0$ for all
  $vw\in E$. Let $F(G,k)$ be the number of nowhere zero $k$-flows of
  $G$. It can be shown that
  \[
  F(G,k)=(-1)^{M-N+Q}\cdot T(G;0,1-k)=k^{-N}\cdot Z_{A(k,k-1,-1)}(G),
  \]
  where $T$ denotes the Tutte polynomial and $A(k,k-1,1)$ the $k\times k$ matrix with diagonal entries $(k-1)$ and off-diagonal
  entries $-1$.
\end{example}

\subsection{Partition functions in statistical physics}\label{sec:stat-phys}

The term ``partition function'' indicates the fact that the functions
we consider here do also have an origin in statistical physics. A
major aim of this branch of physics is the prediction of
phase transitions in dynamical systems from knowing only the
interactions of their microscopic components. In this context,
partition functions are the central quantities allowing for such a
prediction. As a matter of fact many of these partition functions can
be described in the framework we defined above.

Let us see an example for this connection --- the partition function of the \emph{Ising model}.
\newcommand{\zisi}{Z}
Originally introduced by Ising in 1925 \cite{isi25} this model was developed to describe the phase transitions in ferromagnets. 
For some given graph $G$, the model associates with each vertex $v$ a spin $\vcfg_v$ which may be either $+1$ or $-1$. Then the energy of a state $\vcfg$ is given by the \emph{Hamiltonian} defined by
\begin{equation}\label{eq:hamilt_ising}
 H(\vcfg) = -J\sum_{uv \in E} \vcfg_u\vcfg_v  
\end{equation}
where $-J\vcfg_u\vcfg_v$ is the contribution of the energy of each pair of nearest neighbor particles. Let $T$ denote the temperature of the system and $k$ be Boltzmann's constant, define $\beta = (kT)^{-1}$. Then, for a graph $G=(V,E)$ with $N$ vertices, $M$
edges, and $Q$ connected components, we have
\begin{eqnarray*}
\zisi(G,T) &=& \sum_{\vcfg: V \to \{+1,-1\}} e^{- \beta H(\vcfg)} 
\end{eqnarray*}
We straightforwardly get $\zisi(G,T) = e^{\beta J M} Z_{A}(G)$ for the matrix
\begin{equation}\label{eq:ising_pf_pfad}
A = A(T) = \left(\begin{array}{c c}
           e^{2\beta J} & 1 \\
            1 & e^{2\beta J}
          \end{array}\right).
\end{equation}
An extension of this model to systems with more than two spins is the
\emph{$n$-state Potts model}, whose partition function satisfies
\begin{eqnarray*}
Z_{\textup{Potts}}(G;n,v) &=& \sum_{\vcfg : V \to [n]} \prod_{uv \in E} (1 +
v\cdot \delta_{\vcfg(u),\vcfg(v)}).
\end{eqnarray*}
In fact, for $n=2$ and $v = e^{2\beta J} - 1$ we see
that $\zisi(G,T) = e^{\beta J M} Z_{\textup{Potts}}(G;n,v)$. Moreover, this
model can be seen as a specialization of the Tutte Polynomial. Expanding the
above sum over connected components of $G$, we obtain
$$
Z_{\textup{Potts}}(G;n,v) = \sum_{A \subseteq E} n^{q(A)} v^{|A|}
$$
whence it is not difficult to see that, if $(X-1)(Y-1) \in \Nat$,
$$
T(G;X,Y) =  (X-1)^{-q(E)}(Y-1)^{-N} Z_{\textup{Potts}}(G;(X-1)(Y-1),Y-1).
$$
Since $Z_{\textup{Potts}}(G;n,v) =  Z_{A(n,v+1,1)}(G)$ this relation is actually
a special case of equation \eqref{eq:tutte-part}.
\subsection{Which functions are partition functions?}
\label{sec:whichpart}

In this section, we shall state (and partially prove) precise
algebraic characterizations of partition functions over the real and
complex numbers. 

A graph invariant $f:\aG\to \Ring$ is
\emph{multiplicative} if $f(\emptyset)=1$ and $f(G\cdot H)=f(G)\cdot
f(H)$. Here $\emptyset$ denotes the empty graph and $G\cdot H$ denotes
the disjoint union of the graphs $G$ and $H$. An easy calculation
shows:

\begin{observation}\label{obs:multiplicative}
  All partition functions are multiplicative.
\end{observation}

\noindent
To characterize the class of partition functions over the reals, let
$f:\aG\to\mathbb R$ be a graph invariant.  Consider
the(infinite) real matrix $M=(M_{G,H})_{G,H\in\aG}$ with
entries $M_{G,H}:=f(G\cdot H)$. It follows from the multiplicativity
that if $f$ is a partition function then $M$ has rank $1$ and is
positive semidefinite. (Here an infinite matrix is \emph{positive
  semidefinite} if each finite principal submatrix is positive
semidefinite.) The criterion for a graph invariant being a
partition function is a generalization of this simple criterion. Let
$k\ge 0$. A \emph{$k$-labeled graph} is a graph with $k$ distinguished
vertices. Formally, a $k$-labeled graph is a pair
$(G,\phi)$, where $G\in\aG$ and $\phi:[k]\to V(G)$ is injective.  The class of all
$k$-labeled graphs is denoted by $\aG_k$.  For two
$k$-labeled graphs $(G,\phi),(H,\psi)\in\aG_k$, we let
$(G,\phi)\cdot(H,\psi)$ be the $k$-labeled graph obtained from the
disjoint union of $G$ and $H$ by identifying the labeled vertices
$\phi(i)$ and $\psi(i)$ for all $i\in [k]$ and keeping the labels
where they are. We extend $f$ to $\aG_k$ by letting $f(G,\phi):=f(G)$
and define a matrix 
\[
M(f,k)=\big(M(f,k)_{(G,\phi),(H,\psi)}\big)_{(G,\phi),(H,\psi)\in\mathcal
  G_k}
\]
by letting
$M(f,k)_{(G,\phi),(H,\psi)}=f\big((G,\phi)\cdot(H,\psi)\big)$. Note
that if we identify $0$-labeled graphs with plain graphs, then
$M(f,0)$ is just the matrix $M$ defined above. The matrices $M(f,k)$
for $k\ge 0$ are called the \emph{connection matrices} of
$f$. Connection matrices were first introduced by Freedman, Lov{\'a}sz,
and Schrijver~\cite{frelovschri07} to prove a theorem similar to the
following one (Theorem~\ref{thm:fls} below).

\begin{thm}[Schrijver~\cite{schri09}]\label{thm:char-real}
  Let $f:\aG\to\mathbb R$ be a graph invariant. Then $f$ is a
  partition function if and only if it is multiplicative and all its
  connection matrices are positive semidefinite.
\end{thm}

\noindent
We first give the simple proof of the forward direction. A proof
sketch for the backward direction will be given at the end of this subsection. Suppose that $f=Z_A$ for a symmetric matrix $A\in\mathbb
R^{n\times n}$. For all mappings $\chi:=[k]\to[n]$ and $k$-labeled
graphs $(G,\phi)\in\aG_k$ we let
\[
Z_{A,\chi}(G,\phi):=\sum_{\substack{\sigma:V(G)\to[n]\\\sigma(\phi(i))=\chi(i)\text{
    for all }i\in[k]}}\prod_{vw\in
    E(G)}A_{\sigma(v),\sigma(w)}.
\]
Note that for $(H,\psi)\in\aG_k$ we have
\[
Z_{A,\chi}\big((G,\phi)\cdot(H,\psi)\big)=Z_{A,\chi}(G,\phi)\cdot Z_{A,\chi}(H,\psi).
\]
Hence the matrix $M(Z_{A,\chi},k)$ with entries 
$M(Z_{A,\chi},k)_{(G,\phi),(H,\psi)}=Z_{A,\chi}\big((G,\phi)\cdot(H,\psi)\big)$
is positive semidefinite. Furthermore, 
\[
f(G,\phi)=Z_A(G)=\sum_{\chi:[k]\to
    [n]}Z_{A,\chi}(G,\phi),
\]
and thus $M_{f,k}$ is the sum of $n^k$ positive semidefinite matrices,
which implies that it is positive semidefinite. Incidentally, the same
argument shows that the row rank of the $k$th connection matrix
$M(Z_A,k)$ of a partition function of a matrix $A\in\mathbb R^{n\times
  n}$ is at most $n^k$.

Schrijver also obtained a characterization of the class of partition
functions over the complex numbers, which looks surprisingly different
from the one for the real numbers. 
In the following, let $\Ring=\mathbb C$
or $\mathbb C[\vec X]$ for some tuple $\vec X$ of
variables. For a symmetric $n\times n$ matrix $A\in
\Ring^{n\times n}$ we define the ``injective''
partition function $Y_A:\aG\to \Ring$ by
\[
Y_A(G):=\sum_{\substack{\;\tau:V(G)\to[n]\;\\\text{injective}}}\prod_{vw\in
    E(G)}A_{\tau(v),\tau(w)}.
\]
Note that $Y_A(G)=0$ for all $G$ with $|V(G)|>n$. For a graph $G$ and
a partition $P$ of $V(G)$, we let $G/P$ be the graph whose vertex set
consists of the classes of $P$ and whose edge set contains an edge
between the class of $v$ and the class of $w$ for every edge $vw\in
E(G)$. Note that in general $G/P$ will have many loops and multiple
edges. For example, if $P$ has just one class, then $G/P$ will consist
of a single vertex with $|E(G)|$ many loops. We denote the set of
all partitions of a set $V$ by $\Pi(V)$, and for a graph $G$ we let
$\Pi(G):=\Pi\big(V(G)\big)$. For $P,Q\in\Pi(V)$, we write $P\le Q$ if
$P$ refines $Q$. Then we have
\begin{equation}\label{eq:zeta}
Z_A(G)=\sum_{P\in\Pi(G)} Y_A(G/P).
\end{equation}
We can also express $Y_A$ in terms of $Z_A$. For every finite set $V$
there is a unique function $\mu:\Pi(V)\to\mathbb Z$ satisfying the
following equation for all $P\in\Pi(V)$:
\begin{equation}\label{eq:moebius}
\sum_{\substack{Q\in\Pi(V)\\Q\le P}}\mu(Q)=
\begin{cases}
1&\text{if }P=T_V,\\
0&\text{otherwise}.
\end{cases}
\end{equation}
Here $T_V$ denotes the trivial
partition $\{\{v\}\mid v\in V\}$. The function $\mu$ is a restricted
version of the \emph{M\"obius function} of the partially ordered set
  $\Pi(V)$ (for background, see for example
\cite{aig07}). Now it is
easy to see that
\begin{equation}\label{eq:inversion}
Y_A(G)=\sum_{P\in \Pi(G)}\mu(P)\cdot Z_A(G/P).
\end{equation}
We close our short digression on M\"obius inversion
by noting that for all sets $V$ with $|V|=:k$ we have the following
polynomial identity:
\begin{equation}\label{eq:moebius-ff}
\sum_{P\in\Pi(V)}\mu(P)\cdot X^{|P|}=X(X-1)\cdots(X-k+1).
\end{equation}\martin{Beweis angefügt}
To see this, note that it suffices to prove it for all $X\in \Nat$. So
let $X\in\Nat$, and let $A$ be the $(X\times X)$-identity matrix. Furthermore,
let $G$ be the graph with vertex set $V$ an no edges. Then
$Y_A=X(X-1)\cdots(X-k+1)$ and $Z_A(G/P)=X^{|P|}$ for every partition
$P\in\Pi(V)$, and \eqref{eq:moebius-ff} follows from \eqref{eq:inversion}.

Now we are ready to state Schrijver's characterization of the
partition functions over the complex numbers. 

\begin{thm}[Schrijver~\cite{schri09}]\label{thm:char-complex}
  Let $f:\aG\to\mathbb C$ be a graph invariant. Then $f$ is a
  partition function if and only if it is multiplicative and 
  \begin{equation}
    \label{eq:char-complex}
  \sum_{P\in\Pi(G)}\mu(P)\cdot f(G/P)=0
  \end{equation}
  for all $G\in\aG$ with $|V(G)|>|f(K_1)|$.
\end{thm}

\noindent
To understand the condition $|V(G)|>|f(K_1)|$, remember that $K_1$
denotes the graph with one vertex and no edges and note that for every
$n\times n$ matrix $A$ it holds that $Z_A(K_1)=n$. 

\medskip\noindent
\textit{Proof of Theorem~\ref{thm:char-complex} (sketch).}
  The forward direction is almost trivial: Suppose that $f=Z_A$ for a
  symmetric matrix $A\in\mathbb C^{n\times n}$. Then $f$ is
  multiplicative by Observation~\ref{obs:multiplicative}, and we have
  \[
  \sum_{P\in\Pi(G)}\mu(P)\cdot f(G/P)=Y_A(G)=0
  \]
  for all $G$ with $|V(G)|>n=f(K_1)$, where the first equality holds by
  \eqref{eq:inversion} and the second because for $G$ with $|V(G)|>n$ there are no
  injective functions $\sigma:V(G)\to[n]$.

  It is quite surprising that these trivial conditions are sufficient
  to guarantee that $f$ is a partition function. To see that they are,
  let $f:\aG\to\mathbb C$ be a multiplicative graph invariant such
  that \eqref{eq:char-complex} holds for all $G\in\aG$ with
  $|V(G)|>f(K_1)$. 

  We first show that $n:=f(K_1)$ is a non-negative integer. Let
  $k:=\lceil |f(K_1)|\rceil$,
  and let $I_k$ be the graph
  with vertex set $[k]$ and no edges. (Hence $I_k=\emptyset$ if
  $k=0$.) Suppose that $n$ is not a non-negative integer. Then
  $k>|n|$, and by \eqref{eq:char-complex}, the multiplicativity of $f$, and \eqref{eq:moebius-ff} we have
  \[
  0=\sum_{P\in\Pi(I_k)}\mu(P)\cdot
  f(I_k/P)=\sum_{P\in\Pi([k])}\mu(P)\cdot n^{|P|}\\
  =
  n\cdot \big(n-1\big)\cdots \big(n-k+1\big)\neq 0.
  \]
  This is a
  contradiction. 

  A \emph{quantum graph} is a formal linear combination of graphs with
  coefficients from $\mathbb C$, that is, an expression
  $\sum_{i=1}^\ell a_i G_i$, where $\ell\ge0$ and $a_i\in\mathbb C$
  and $G_i\in\aG$ for all $i\in[\ell]$. The class of all quantum
  graphs is denoted by $\aQG$. The quantum graphs obviously form a
  vector space over $\mathbb C$, and by extending the product
  ``disjoint union'' linearly from $\aG$ to $\aQG$, we turn this
  vector space into an algebra. We also extend the function $f$
  linearly from $\aG$ to $\aQG$. Observe that $f$ is an algebra
  homomorphism from $\aQG$ to $\mathbb C$, because it is
  multiplicative.

  For all $i,j\in[n]$ we let $X_{\{i,j\}}$ be a variable, and we let $\vec X$ be the tuple of all these
  variables ordered lexicographically. Furthermore, we let $X$ be the
  $n\times n$ matrix with $X_{i,j}:=X_{j,i}:=X_{\{i,j\}}$. We view $X$ as a matrix over the
  polynomial ring $\mathbb C[\vec X]$. We extend the function
  $Z_X:\aG\to\mathbb C[X]$ linearly from $\aG$ to $\aQG$.
  Then $Z_X$ is an algebra homomorphism. It is not too hard to show
  that the image $Z_X(\aQG)$ consists precisely of all polynomials in
  $\mathbb C[\vec X]$ that are invariant under all permutations 
  $
  X_{\{i,j\}}\mapsto X_{\{\pi(i),\pi(j)\}}
  $
  of the variables for permutations $\pi$ of $[n]$.
Using \eqref{eq:inversion}
  and other properties of the M\"obius inversion, it can be shown that
  the kernel of $Z_X$ is contained in the kernel of $f$. This implies
  that there is an algebra homomorphism $\hat f$ from the image
  $Z_X(\aQG)$ to $\mathbb C$ such that $f=\hat f\circ Z_X$.

  Then
  \[
  I:=\big\{ p\in Z_X(\aQG)\mid \hat f(p)=0\big\}
  \]
  is an ideal in the subalgebra $Z_X(\aQG)\subseteq\mathbb C[\vec
  X]$. We claim that the polynomials in $I$ have a common
  zero. Suppose not. Let $I'$ be the ideal generated by $I$ in
  $\mathbb C[\vec X]$. Then by Hilbert's Nullstellensatz it holds that
  $1\in I'$. Using the fact that the the subalgebra $Z_X(\aQG)$
  consists precisely of all polynomials in $\mathbb C[\vec X]$ that
  are invariant under all permutations $ X_{\{i,j\}}\mapsto
  X_{\{\pi(i),\pi(j)\}}$ for $\pi\in S_n$, one
  can show that actually $1\in I$. But then
  \[
  0=\hat f(1)=\hat f(Z_X(K_0))=f(K_0)=1,
  \]
  which is a contradiction. 

  Thus the polynomials in $I$ have a common zero. Let $\vec
  A=(A_{\{i,j\}}\mid i,j\in[n])$ be such a zero, and let $A\in\mathbb C^{n\times
    n}$ be the corresponding symmetric matrix. Observe that for each graph $G$
  it holds that $Z_X(G)-f(G)\in I$, because $\hat f(Z_X(G))-\hat
  f(f(G))=f(G)-f(G)=0$. Hence $Z_A(G)-f(G)=0$ and thus
  $f(G)=Z_A(G)$. This completes our proof sketch.
\qed

\noindent
Even though Theorems~\ref{thm:char-real} and \ref{thm:char-complex} look
quite different, it is not hard to derive Theorem~\ref{thm:char-real}
from Theorem~\ref{thm:char-complex}. In the remainder of this
subsection, we sketch how this is done.

\medskip\noindent
\textit{Proof of Theorem~\ref{thm:char-real} (sketch).}
  We have already proved the forward direction. For the backward
  direction, let $f:\aG\to\mathbb R$ be a multiplicative graph
  invariant such that all connection matrices of $f$ are positive
  semidefinite. We first prove that $f$ satisfies condition
  \eqref{eq:char-complex}: Let $n:=f(K_1)$ (we do not know yet that
  $n$ is an integer, but it will turn out to be), and let $k>n$ be a non-negative integer. Let $I_k$ be the graph with
  vertex set $[k]$ and no edges, and for every partition $P$ of $[k]$,
  define $\phi_P:[k]\to V(I_k/P)$ to be the canonical projection, that
  is, $\phi_P(i)$ is the class
  of $i$ in the partition $P$. Then $(I_k/P,\phi_P)\in\aG_k$. As
  the $k$th connection matrix $M(f,k)$ is positive semidefinite, we
  have
  \[
  \sum_{P,Q\in\Pi([k])}\mu(P)\cdot\mu(Q)\cdot
  f\big((I_k/P,\phi_P)\cdot(I_k/Q,\phi_Q)\big)\ge 0.
  \]
  For partitions $P,Q\in\Pi([k])$, let $P\vee Q$ be the least upper
  bound of $P$ and $Q$ in the partially ordered set $\Pi([k])$, and
  note that $(I_k/P,\phi_P)\cdot(I_k/Q,\phi_Q)=(I_k/P\vee
  Q,\phi_{P\vee Q})$. Hence by the multiplicativity of $f$ we have
  $f\big((I_k/P,\phi_P)\cdot(I_k/Q,\phi_Q)\big)=n^{|P\vee Q|}$.
  A calculation similar to the one that leads to \eqref{eq:moebius-ff}
  shows that $\sum_{P,Q\in\Pi[k]}\mu(P)\cdot\mu(Q)\cdot X^{|P\vee
    Q|}=X\cdot (X-1)\cdots(X-k+1)$. Hence for all non-negative integers
  $k>n$ we have
  \begin{equation}\label{eq:2801102220}
  0\le\sum_{P,Q\in\Pi([k])}\mu(P)\cdot\mu(Q)\cdot
  f\big((I_k/P,\phi_P)\cdot(I_k/Q,\phi_Q)\big)=n\cdot(n-1)\cdots(n-k+1).  
  \end{equation}
  This is only possible if $n$ is a non-negative integer, in
  which case equality holds. 

  \martin{Das folgende etwas ausführlicher}
  Now let $G$ be an arbitrary graph with $k:=|V(G)|>n$. Without loss of
  generality we may assume that $V(G)=[k]$. Let $\psi$ be the identity on
  $[k]$. Consider the $(|\Pi([k])|+1)\times(|\Pi([k])|+1)$-principal submatrix
  $M_0$ of $M(f,k)$ with rows and columns indexed by the $k$-labeled graphs
  $(I_k/P,\phi_P)$ for all $P\in\Pi([k])$ and $(G,\psi)$. For every
  $x\in\mathbb R$, let $\vec v_x$ be the column vector with entries $\mu(P)$
  for all $P\in\Pi([k])$ followed by $x$. By the positive semi-definiteness of
  $M(f,k)$, we have
  \begin{align*}
  0\le\;&\vec v_x^\top M_0\vec v_x\\
=&\sum_{P,Q\in\Pi([k])}\mu(P)\cdot\mu(Q)\cdot
  f\big((I_k/P,\phi_P)\cdot(I_k/Q,\phi_Q)\big)\\
  &+2x\cdot\sum_{P\in\Pi([k])}\mu(P)\cdot
  f\big((I_k/P,\phi_P)\cdot(G,\psi)\big)+x^2\cdot
f\big((G,\psi)\cdot(G,\psi)\big)\\
  =\;&2x\cdot\sum_{P\in\Pi([k])}\mu(P)\cdot
  f\big((I_k/P,\phi_P)\cdot(G,\psi)\big)+x^2\cdot
f\big((G,\psi)\cdot(G,\psi)\big)\\
  \end{align*}
  This is only possible for all $x\in\mathbb R$ if $\sum_{P\in\Pi([k])}\mu(P)\cdot
  f\big((I_k/P,\phi_P)\cdot(G,\psi)\big)=0$. Note that for every $P\in\Pi([k])$
  it holds that
  $(I_k/P,\phi_P)\cdot(G,\psi)=(G/P,\phi_P)$. Thus 
  \[
  \sum_{P\in\Pi(G)}\mu(P)\cdot f(G/P)=0.
  \]    
  If follows from Theorem~\ref{thm:char-complex} that $f=Z_A$ for some
  matrix $A\in\mathbb C^{n\times n}$.
  
  For every $\ell\ge 0$, let $F_\ell$ be the graph with two vertices
  and $\ell$ parallel edges between these vertices. Exploiting the positive 
  semi definiteness of the principal submatrix of
  $M(f,2)$ with rows and columns indexed by the graphs $F_\ell$ for
  $0\le\ell\le n(n+1)/2$, it is not hard to show that the matrix $A$
  is real.
\qed

\subsection{Generalizations}
\label{sec:part-gen}

\subsubsection*{Vertex weights}
We have mentioned that a partition function $Z_A$ can be viewed as
mapping a graph $G$ to a weighted sum of homomorphisms to the
edge-weighted graph represented by the matrix $A$. We may also put
weights on the vertices of a graph. It will be most convenient to
represent vertex weights on a graph $H$ by a diagonal matrix
$D=(D_{v,w})_{v,w\in V(H)}$, where $D_{v,v}$ is the weight of vertex $v$
and $D_{v,w}=0$ if $v\neq w$. Then we define the \emph{weight} of a
mapping $\sigma$ from $G$ to $H$ to be the product of the weights of
the images of the edges and the images of the vertices. More
abstractly, for every symmetric matrix $A\in \Ring^{n\times n}$ and diagonal
matrix $D\in \Ring^{n\times n}$ we define
a function $Z_{A,D}:\aG\to \Ring$ by
\[
Z_{A,D}(G):=\sum_{\;\sigma:V(G)\to [k]\;}\prod_{vw\in
  E(G)}A_{\sigma(v),\sigma(w)}\cdot\prod_{v\in
  V(G)}D_{\sigma(v),\sigma(v)}.
\]
We call $Z_{A,D}$ a \emph{partition functions with vertex weights}
over $\Ring$. ([Freedman, Lov{\'a}sz, and Schrijver~\cite{frelovschri07} use
the term \emph{homomorphism functions}.) The vertex
weights enable us to get rid of the constant factors in some of the
earlier examples and give smoother formulations. For example:

\begin{example}
  Recall that by Example~\ref{exa:flow}, the number $F(G,k)$ of nowhere-zero
  $k$-flows of an $N$-vertex graph $G$ is $k^{-N}\cdot Z_{A}(G)$ for
  the $k\times k$ matrix $A$ with diagonal entries $(k-1)$ and
  off-diagonal entries $-1$. Let $D$ be the $(k\times k)$-diagonal
  matrix with entries $D_{ii}:=1/k$ for all $i\in[k]$. Then $F(G,k)=Z_{A,D}(G)$.
\end{example}

\noindent
Freedman, Lov{\'a}sz, and Schrijver gave an algebraic characterization of
the class of partition functions with non-negative vertex weights over the reals, again
using connection matrices. However, they only consider graph
invariants defined on the class $\aG'$ of graphs without loops (but
with parallel edges). For a graph invariant $f:\aG'\to\mathbb R$, we
define the $k$th connection matrix $M(f,k)$ as for graph invariants
defined on $\aG$, except that we omit all rows and columns indexed by
graphs with loops. We call a matrix $A\in\mathbb R^{n\times n}$
\emph{non-negative} if all its entries are non-negative.

\begin{thm}[Freedman, Lov{\'a}sz, and Schrijver~\cite{frelovschri07}]
  \label{thm:fls}
  Let $f:\aG'\to\mathbb R$ be a graph invariant. Then the following
  two statements are equivalent:
  \begin{enumerate}
  \item There are a symmetric matrix $A\in\mathbb R^{n\times n}$ and a
    non-negative diagonal matrix $D\in\mathbb R^{n\times n}$ such that $f(G)=Z_{A,D}(G)$ for all $G\in\aG'$.
  \item There is a $q\ge 0$ such that for all $k\ge 0$ the matrix $M(f,k)$ is positive semidefinite
    and has row rank at most $q^k$.
  \end{enumerate}
\end{thm}
Despite the obvious similarity of this theorem with
Theorem~\ref{thm:char-real}, the known proofs of the two theorems are
quite different.

\subsubsection*{Asymmetric matrices and directed graphs}
Of course we can also count homomorphism between directed graphs. We
denote the class of all directed graphs by $\aD$; as undirected graphs
we allow directed graphs to have loops and parallel edges. For
every directed graph $H$, we define the homomorphism counting function
$Z_H:\aD\to \mathbb Z$ by 
letting $Z_H(D)$ be the number of homomorphisms from $D$ to $H$. 
For every square matrix $A\in \Ring^{n\times n}$
we define the function $Z_A:\aD\to \Ring$ by
\[
Z_A(D):=\sum_{\;\sigma:V(D)\to [n]\;}\prod_{(v,w)\in E(D)}A_{\sigma(v),\sigma(w)}.
\]
Note that if $A$ is a symmetric matrix, then $Z_A(D)=Z_A(G)$, where
$G$ is the underlying undirected graph of $G$.

\subsubsection*{Hypergraphs and relational structures}
Recall that a \emph{hypergraph} is a pair $H=(V,E)$ where $V$ is a
finite set and $E\subseteq 2^V$. Elements of $V$ are
called \emph{vertices}, elements of $E$ \emph{hyperedges}. A
hypergraph is $r$-uniform, for some $r\ge 1$, if all its hyperedges have
cardinality $r$. Thus a 2-uniform hypergraph is just a simple graph. A
\emph{homomorphism} from a hypergraph $G$ to a hypergraph $H$ is a
mapping $h:V(G)\to V(H)$ such that $h(e)\in E(H)$ for all $e\in
E(H)$. Of course, for $e=\{v_1,\ldots,v_r\}$ we let
$h(e)=\{h(v_1),\ldots,h(v_r)\}$.\footnote{There is no need to define
  hypergraph homomorphisms as mappings from $V(G)\cup E(G)$ to
  $V(H)\cup E(H)$ because we do not allow parallel hyperedges.}
Hypergraph homomorphism counting functions and partition functions
have only been considered for uniform hypergraphs. The class of all
$r$-uniform hypergraphs is denoted by $\aH_r$. For all $H\in\aH_r$, we
define the homomorphism counting function $Z_H:\aH_r\to\mathbb Z$ in
the obvious way. The natural generalization of partition functions to
$r$-uniform hypergraphs is defined by symmetric functions $A:[n]^r\to
\Ring$. We define $Z_A:\aH_r\to \Ring$ by
\[
Z_A(G):=\sum_{\;\sigma:V(G)\to[n]\;}\prod_{e\in E(G)}A(\sigma(e)),
\]
where for $e=\{v_1,\ldots,v_r\}$ we let
$A(\sigma(e)):=A(\sigma(v_1),\ldots,\sigma(v_r))$. This is
well-defined because $A$ is symmetric.

Let us finally consider homomorphisms between
relational structures. A \emph{(relational) vocabulary} $\sigma$ is a
set of relation symbols; each relation symbol comes with a prescribed
finite arity. A \emph{$\sigma$-structure} $B$ consist of a set $V(B)$,
which we call the \emph{universe} or \emph{vertex set} of $B$, and for
each $r$-ary relation symbol $S\in\sigma$ an $r$-ary relation
$S(B)\subseteq V(B)^r$. Here we assume all structures to be finite,
that is, to have a finite universe and vocabulary. For example, a
simple graph may be viewed as an $\{E\}$-structure $G$, where $E$ is a
binary relation symbol and $E(G)$ is irreflexive and symmetric. An
$r$-uniform hypergraph $H$ may be viewed as an $\{E_r\}$-structure,
where $E_r$ is an $r$-ary relation symbol and $E_r(H)$ is symmetric
and contains only tuples of pairwise distinct elements. For each
vocabulary $\sigma$, we denote the class of all $\sigma$-structures by
$\aS_\sigma$.
For each $\sigma$-structure $B$ we define the
homomorphism counting functions $Z_B:\aS_\sigma\to\mathbb Z$ in the
usual way. 

To generalize partition functions, we consider
\emph{weighted structures}. For a ring $\Ring$ and a vocabulary $\sigma$,
an \emph{$\Ring$-weighted $\sigma$-structure} $B$ consists of a finite set
$V(B)$ and for each $r$-ary $R\in\sigma$ a mapping $R(\Ring):V(B)^r\to
\Ring$. Then we define $Z_B:\aS_\sigma\to \Ring$ by 
\[
Z_B(A):=\sum_{\;\sigma:V(G)\to[n]\;}\prod_{R\in\sigma}\prod_{\substack{(a_1,\ldots,a_r)\in
  V(A)^r,\\
\text{where }r\text{ is the arity of
}R}}R(\Ring)(\sigma(a_1),\ldots,\sigma(a_r)).
\]
Note that this does not only generalize plain partition functions on graphs,
but also partition functions with vertex weights, because we may view a
graph with weights on the vertices and edges as a weighted
$\{E,P\}$-structure, where $E$ is a binary and $P$ a unary relation
symbol.

It is a well-known observation due to Feder and Vardi~\cite{fedvar98}
that \emph{constraint satisfaction problems} may be viewed as
homomorphism problems between relational structures and vice
versa. Thus counting homomorphisms between relational structures
correspond to counting solutions to constraint satisfaction problems.

\subsubsection*{Edge Models}

Partition functions and all their generalizations considered so far
are weighted sums over mappings defined on the vertex set of the graph or
structure. In the context of statistical physics, they are sometimes
called \emph{vertex models}. There is also a notion of \emph{edge
  model}. An edge model is given by a function $F:\mathbb N^n\to \Ring$. Let $G=(V,E)$ be a graph. Given a
``configuration'' $\tau:E\to[n]$, for every vertex $v$ we let
$t(\tau,v):=(t_1,\ldots,t_n)$, where $t_i$ is the number of edges $e$
incident with $v$ such that $\tau(e)=i$. We define a function $\tilde
Z_F:\aG\to \Ring$ by
\[
\tilde Z_F(G):=\sum_{\;\tau:E\to[n]\;}\prod_{v\in V}F\big(t(\tau,v)\big).
\]
Szegedy~\cite{sze07} gave a characterization of the graph invariants over the
reals 
expressible by edge models that is similar to the characterizations of
partition functions (vertex models) given in
Theorems~\ref{thm:char-real} and \ref{thm:fls}. 

\begin{example}
  Let $F:\mathbb N^2\to\mathbb Z$ be defined by $F(i,j):=1$ if $j=1$
  and $F(i,j)=0$ otherwise. Let $G=(V,E)$ be a graph. Observe that for
  every $\tau:E(G)\to\{1,2\}$ and every $v\in V$ it holds that
  $F\big(t(\tau,v)\big)=1$ if and only if there is exactly one edge in
  $\tau^{-1}(2)$ that is incident with $v$. Hence $\prod_{v\in
    V}F\big(t(\tau,v)\big)=1$ if $\tau^{-1}(2)$ is a perfect matching
  of $G$ and $\prod_{v\in
    V}F\big(t(\tau,v)\big)=0$ otherwise. It follows that $\tilde Z_F(G)$ is
  the number of perfect matchings of $G$.

  It can be proved that the function $f=\tilde Z_F$ counting perfect matchings is not
  a partition function, because the connection matrix $M(f,1)$ is not
  positive semi-definite~\cite{frelovschri07}.
\end{example}
%
%
%
Related to the edge models is a class of functions 
based on Valiant's holographic algorithms \cite{val08}:
so-called \emph{holant functions} which have been introduced
by Cai, Lu, and Xia~\cite{cailuxia09}.
A holant function is given by a \emph{signature grid} 
$\Omega = (G, \mathcal F, \pi)$ where $G = (V,E)$ is a graph and, for some 
$n \in \Nat$, the set $\mathcal F$ contains functions $f:[n]^{a_f} \to \Ring$,
each of some arity $a_f$. Further 
$\pi$ maps vertices $v \in V$ to functions $f_v \in \mathcal F$ such that
$d(v) = a_{f_v}$.
Let, for some vertex $v$, denote $E(v)$ the set of edges incident to $v$. The holant function over this signature grid is then defined as
$$
\textup{Holant}_{\Omega} = \sum_{\tau: E \to [n]} \prod_{v \in V} f_v(\tau\vert_{E(v)}).
$$
Note that we assume here implicitly that $E$ has some ordering, therefore 
$f_v(\tau\vert_{E(v)})$ is well-defined.
Edge models are a special case of this framework, since $\tilde Z_F(G) = \textup{Holant}_{\Omega}$ for the signature grid which satisfies   $f_v(\tau\vert_{E(v)}) = F(t(\tau,v))$ for all $v \in V$.
More generally it can be shown that partition functions on relational
structures can be captured by holant functions. 

\section{Complexity}
\label{sec:complexity}

Partition functions tend to be hard to compute. More precisely, they tend to
be hard for the complexity class $\sharpP$ introduced by Valiant in \cite{val79a},
which may be viewed as the
``counting analogue'' of \NP. A counting problem $C$ (that is, a function with
values in the non-negative integers) belongs to $\sharpP$ if and only if there is
a nondeterministic polynomial time algorithm $A$ such that for every instance
$x$ of $C$ it holds that $C(x)$ is the number of accepting computation paths
of $A$ on input $x$.  There is a theory of reducibility and
$\sharpP$-completeness much like the theory of \NP-completeness. (The
preferred reductions in the complexity theory of counting problems are
\emph{Turing reductions}. We exclusively work with Turing reductions in this
article .)  Most \NP-complete decision problems have natural
$\sharpP$-complete counting problems associated with them. For example, the
problem of counting the number of independent sets of a graph and the problem
of counting the number of 3-colorings of a graph are both $\sharpP$-complete.
There is a counting analogue of a well-known Theorem due to Ladner~\cite{lad75} stating
that there are counting problems in $\sharpP$ that are neither $\sharpP$-complete nor in
$\FP$ (the class of all counting, or more generally functional problems solvable in polynomial time);
indeed the counting complexity classes between $\FP$ and $\sharpP$ form a
dense partial order. 

Even though the class of partition functions is very rich, it turns
out that partition functions and also their various generalizations
discussed in Section~\ref{sec:part-gen} exhibit a \emph{complexity
  theoretic dichotomy}: Some partition functions can be computed in
polynomial time, most are $\sharpP$-hard, but there are no partition
functions of intermediate complexity. A first dichotomy theorem for
homomorphism counting functions of graphs was obtained by Dyer and
Greenhill:

\begin{thm}[Dyer and Greenhill~\cite{dyegre00}]\label{thm:dg}
  Let $H$ be a graph without parallel edges. Then $Z_H$ is computable
  in polynomial time if each connected component of $H$ is either a
  complete graph with a loop at every vertex or a complete bipartite
  graph.\footnote{We count the empty graph and the graph $K_1$ with one vertex and no
  edges as bipartite graphs.} Otherwise, $Z_H$ is $\sharpP$-complete.
\end{thm}
When thinking about the complexity of partition functions over the
real or complex numbers, we face the problem of which model of
computation to use. There are different models which lead to different
complexity classes. To avoid such issues, we restrict our attention to
algebraic numbers, which can be represented in the standard bit
model.\footnote{In \cite{bulgro05}, the complexity classification of
  partition functions for non-negative real matrices
  (Theorem~\ref{thm:bg} of this article) was stated for arbitrary real
  matrices in one of the standard models of real number
  computation. However, the proof is faulty and can only be made to
  work for real algebraic numbers (cf.\ Section~\ref{subsec:pos_to_X} of this
  article).}  We will discuss the representation of and
computation with algebraic numbers in Section~\ref{subsec:pos_to_X}.
The fields of
real and complex algebraic numbers are denoted by $\Ralg$ and $\Calg$,
respectively. In general, partition functions are no counting
functions (in the sense that their values are no integers) and hence
they do not belong to the complexity class $\sharpP$. For that reason,
in the following results we only state $\sharpP$-hardness and not
completeness. As an easy upper bound, we note that partition functions
over $\Ralg$ and $\Calg$ belong to the complexity class $\FP^ {\sharpP}$
of all functional problems that can be solved by a polynomial time
algorithm with an oracle to a problem in $\sharpP$.

Let us turn to partition functions over the reals. 
We first observe that if $A\in\Ralg^{n\times n}$ is a symmetric matrix of row
rank $1$, then $Z_A$ is easily computable in polynomial time. Indeed,
write $A=\vec a^T\vec a$ for a (row) vector $\vec
a=(a_1,\ldots,a_n)\in\Ralg^n$. Then for every graph $G=(V,E)$,
\begin{equation}\label{eq:comp1}
  Z_A(G)=\sum_{\;\sigma:V\to[n]\;}\prod_{vw\in
    E}a_{\sigma(v)}a_{\sigma(w)}
  =\sum_{\;\sigma:V\to[n]\;}\prod_{v\in V}a_{\sigma(v)}^{\deg(v)}=
  \prod_{v\in V}\sum_{i=1}^na_i^{\deg(v)}.
\end{equation}
Here $\deg(v)$ denotes the degree of a vertex $v$.
The last term in \eqref{eq:comp1}, which only involves a polynomial number of arithmetic
operations, can easily be evaluated in polynomial time. Thus partition
functions of rank-$1$ matrices are easy to compute. It turns out that
all easy partition functions of non-negative real matrices are based on
a ``rank-$1$'' condition. 

Let us call a matrix $A$ \emph{bipartite} if its underlying graph
$G(A)$ is bipartite. After suitable permuting rows and columns, a bipartite matrix has the form
\[
A=\begin{pmatrix}0&B\\B^T&0\end{pmatrix},
\]
where we call $B$ and $B^T$ the \emph{blocks} of $A$.
A similar argument
as the one above shows that if $B$ has row rank $1$ then $Z_A$ is
computable in polynomial time. Note that we only need to compute
$Z_A(G)$ for bipartite $G$, because $Z_A(G)=0$ for non-bipartite
$G$. 

The \emph{connected components} of a matrix $A$ are the principal
submatrices corresponding to the connected
components of the underlying graph $G(A)$. Note that if $A$ is a
matrix with connected components $A_1,\ldots,A_m$, then for every connected graph $G$ it holds that
$Z_A(G)=\sum_{j=1}^mZ_{A_j}(G)$. By the multiplicativity of
partition functions, for a graph $G$ with connected components
$G_1,\ldots,G_\ell$, we thus have
$Z_A(G)=\prod_{i=1}^\ell\sum_{j=1}^mZ_{A_j}(G_i)$. This reduces the
computation of $Z_A$ to the computation of the $Z_{A_j}$.  

\begin{thm}[Bulatov and Grohe~\cite{bulgro05}]\label{thm:bg}
  Let $A\in\Ralg^{n\times n}$ be a symmetric non-negative matrix. Then
  $Z_A$ is computable in polynomial time if all components of $A$ are
  either of row rank $1$ or bipartite and with blocks of row rank $1$.
  Otherwise, $Z_A$ is $\sharpP$-hard.
\end{thm}

\noindent
Note that the theorem is consistent with Theorem~\ref{thm:dg}, the
special case for 0-1-matrices. We have already proved that $Z_A$ is
computable in polynomial time if all components of $A$ are either of
row rank $1$ or bipartite and with blocks of row rank $1$. The much
harder proof that all other cases are $\sharpP$-hard will be given in Sections~\ref{sec:it}--\ref{sec:11c} of
this paper. As a by-product, we also obtain a proof of Theorem~\ref{thm:dg}.

The rest of this section is a survey of further dichotomy results. We
will not prove them. The following two examples show that if we admit
negative entries in our matrices, then the rank-1 condition is no
longer sufficient to explain tractability.

\begin{example}\label{exa:b2}
  Consider the matrix $B=
  \begin{pmatrix}
    1&1\\1&-1
  \end{pmatrix}
  $
  first introduced in Example~\ref{exa:evensubgraphs}. The row rank of
  this matrix is $2$, yet we shall prove that $Z_B$ is computable in
  polynomial time. Let $G=(V,E)$.
  Then
  \begin{align*}
  Z_B(G)=\sum_{\;\sigma:V\to[2]\;}\prod_{vw\in
    E}B_{v,w}&=\sum_{\;\sigma:V\to[2]\;}\prod_{vw\in
    E}(-1)^{(\sigma(v)-1)\cdot(\sigma(w)-1)}\\
  &=\sum_{\;\sigma:V\to[2]\;}(-1)^{\sum_{vw\in E}\sigma(v)\cdot\sigma(w)}.
  \end{align*}
  Hence $Z_B(G)$ is $2^N$ minus twice the number of mappings
  $\sigma:V\to\{0,1\}$ such that $\sum_{vw\in
    E}\sigma(v)\cdot\sigma(w)$ is odd. Here $N$ is the number of
  vertices of $G$. Thus to compute $Z_B$, we need to determine the
  number of solutions of the quadratic equation
  \[
  \sum_{vw\in E}x_vx_w=1
  \]
  over the $2$-element field $\mathbb F_2$. The number of solutions of
  a quadratic equation over $\mathbb F_2$ or any other finite field
  can be computed in polynomial time. This follows easily from
  standard normal forms for quadratic equations (see, for example,
  \cite{lidnie97}, Section~6.2).
\end{example}

\begin{example}
  The \emph{tensor product} of two matrices $A\in \Ring^{m\times n}$ and
  $B\in \Ring^{k\times\ell}$ is the $m\cdot k\times n\cdot\ell$ matrix
  \[
  A\otimes B:=
  \begin{pmatrix}
    A_{1,1}\cdot B&\cdots&A_{1,n}\cdot B\\
    \vdots&&\vdots\\
     A_{m,1}\cdot B&\cdots&A_{m,n}\cdot B
  \end{pmatrix}.
  \]
  It is easy to see that for all square matrices $A,B$ and all
  graphs $G$ it holds that 
  \[
  Z_{A\otimes B}(G)=Z_A(G)\cdot Z_B(G).
  \]
  We can thus use the tensor product to construct new matrices with
  polynomial time computable partition functions. For example, the
  following three $4\times 4$-matrices have polynomial time computable
  partition functions:
  \[
  \left(
  \begin{array}{rrrr}
    1&1&1&1\\
    1&-1&1&-1\\
    1&1&-1&-1\\
    1&-1&-1&1
  \end{array}
  \right)
  \hspace{2cm}
  \left(\begin{array}{rrrr}
    -1&-1&1&1\\
    -1&1&1&-1\\
    1&1&-1&-1\\
    1&-1&-1&1
  \end{array}\right)
  \hspace{2cm}
  \left(\begin{array}{rrrr}
    1&2&1&2\\
    2&4&2&4\\
    1&2&-1&-2\\
    2&4&-2&-4
  \end{array}\right).
\]
\end{example}

Very roughly, all symmetric real matrices with a polynomial time
computable partition function can be formed from matrices of rank $1$
and matrices associated with quadratic equations over $\mathbb
F_2$ (in a way similar to the matrix $B$ of Example~\ref{exa:b2}) by
tensor products and similar constructions. The precise characterization
of such matrices is very complicated, and it makes little sense to state
it explicitly. In the following, we say that a class $\mathcal F$ of
functions \emph{exhibits an $\FP$~--~$\sharpP$-dichotomy} if all functions
in $\mathcal F$ are either in $\FP$ or $\sharpP$-hard. We say that
$\mathcal F$ \emph{exhibits an effective $\FP$~--~$\sharpP$-dichotomy} if
in addition it is decidable if a given function in $\mathcal F$, represented
for example by a matrix, is in $\FP$ or $\sharpP$-hard.

\begin{thm}[Goldberg, Grohe, Jerrum, Thurley~\cite{golgrojerthu09}]
  The class of partition functions of symmetric matrices over the
  reals exhibits an effective $\FP$~--~$\sharpP$-dichotomy.
\end{thm}

\noindent
There are two natural ways to generalize this result to complex
matrices. Symmetric complex matrices where studied by Cai, Chen,
Lu~\cite{caichelu09}, and Hermitian matrices by Thurley~\cite{thu09}.

\begin{thm}[Cai, Chen, Lu~\cite{caichelu09}, Thurley~\cite{thu09}]
  \begin{enumerate}
  \item
    The class of partition functions of symmetric matrices over the
  complex numbers exhibits an effective $\FP$~--~$\sharpP$-dichotomy.
  \item
    The class of partition functions of Hermitian matrices exhibits an effective $\FP$~--~$\sharpP$-dichotomy.
  \end{enumerate}
\end{thm}

\noindent
Beyond Hermitian matrices, partition functions of arbitrary, not necessarily symmetric matrices are much more difficult to
handle. An $\FP$~--~$\sharpP$-dichotomy follows from Bulatov's
Theorem~\ref{thm:bul} below, but it is difficult to understand how this
dichotomy classifies directed graphs. Dyer, Goldberg and Paterson \cite{dyegolpat07}
proved an effective $\FP$~--~$\sharpP$-dichotomy for the class of
homomorphism counting functions $Z_H$ for directed acyclic graphs $H$;
it is based on a complicated ``rank-1'' condition. 
A very recent result by Cai and Chen \cite{che10}
establishes an effective dichotomy for partition functions $Z_A$ 
on non-negative real-valued (i.e. not necessarily symmetric) matrices $A$.
The complexity of
hypergraph partition functions was studied by Dyer, Goldberg, and
Jerrum~\cite{dyegoljer08}, who proved the following theorem:

\begin{thm}[Dyer, Goldberg, and
Jerrum~\cite{dyegoljer08}]
For every $r$, the class of functions $Z_A$ from the class $\aH_r$ of
$r$-uniform hypergraphs defined by non-negative symmetric functions
$A:[n]^r\to\Ralg$ exhibits an effective $\FP$~--~$\sharpP$-dichotomy.
\end{thm}

\noindent
Let us finally turn to homomorphism counting functions for arbitrary
relational structures, or equivalently solution counting functions for
constraint satisfaction problems. Creignou and Hermann~\cite{creher96} proved a
dichotomy for the Boolean case, that is, for homomorphism counting
functions of relational structures with just two elements. Dichotomies
for the weighted Boolean case were proved in \cite{dyegoljer09} for
non-negative real weights, in \cite{buldyegoljr08} for arbitrary real
weights, and in \cite{cailuxia09} for complex weights. Briquel and
Koiran~\cite{brikoi09} study the problem in an algebraic computation
model. The general (unweighted) case was settled by Bulatov:

\begin{thm}[Bulatov~\cite{bul08}]\label{thm:bul}
  The class of homomorphism counting functions for 
relational structures exhibits an $\FP$~--~$\sharpP$-dichotomy.
\end{thm}

\noindent
Bulatov's proof uses deep results from universal algebra, and for some time
it was unclear whether his dichotomy is effective. Dyer and Richerby \cite{dyeric10} gave 
an alternative
proof which avoids much of the universal algebra machinery, and they could show that
the dichotomy is decidable in $\NP$ \cite{dyeric11,dyeric10b}.
Extensions to the weighted case for nonnegative weights in $\Q$ have been given by
Bulatov et al. \cite{buldyegoljjr10}, and for nonnegative algebraic weights by
Cai et al. \cite{caichelu10}.\marc{Absatz geaendert!}

\medskip
\noindent In the context of holant functions complexity dichotomies have been obtained 
by Cai, Lu, and Xia \cite{cailuxia08,cailuxia09}.
Besides the result stated above, they considered
holant functions on signature grids $\Omega = (G,\mathcal F,\pi)$, given that 
$\mathcal F$ is a set of symmetric functions satisfying certain additional 
conditions.
A first a dichotomy \cite{cailuxia08} is given for any of boolean 
symmetric functions $\mathcal F$ on planar bipartite $2,3$-regular graphs.
In \cite{cailuxia09} dichotomies are presented, assuming that the class
$\mathcal F$ contains certain unary functions.

\section{An Itinerary of the Proof of Theorem~\ref{thm:bg}}
\label{sec:it}

From now on, we will work exclusively with
partition functions defined on matrices with entries in one of the
rings $\Ralg, \Qu[X], \Qu, \Int[X], \Int$, and we always let $\Ring$
denote one of these rings. For technical
reasons, we will always assume that numbers in $\Ralg$ are given in
\emph{standard representation} in some algebraic extension field
$\Qu(\theta)$. 
That is, we consider numbers in $\Qu(\theta)$ as vectors in a
$d$-dimensional $\Qu$-vectorspace, where $d$ is the degree of $\Qu(\theta)$ over 
$\Qu$.
It is well-known that for any finite set of numbers from
$\Ralg$ we can compute a $\theta$ which constitutes the corresponding
extension field (cf. \cite{coh93} p. 181). For further details see
also the treatment of this issue in \cite{dyegoljer08,thu09}.

By $\deg(f)$ we denote the \emph{degree} of a polynomial $f$.
For two problems $P$ and $Q$ we use $P \Tle Q$ to denote that $P$ is polynomial time Turing reducible to $Q$. Further, we write $P \Tequiv Q$ to denote that $P \Tle Q$ and $Q \Tle P$ holds.


An $m \times n$ matrix $A$ is \sdefi{decomposable}{decomposable
  matrix}\sdefisub{}{matrix}{decomposable}, if there are non-empty
index sets $I \subseteq [m]$, $J \subseteq [n]$ with $(I,J) \neq [m]
\times [n]$ such that $A_{ij} = 0$ for all $(i,j) \in \bar I \times J$
and all $(i,j) \in I \times \bar J$, where $\bar I:=[m]\setminus I$
and $\bar J:=[n]\setminus J$. A matrix is
\sdefi{indecomposable}{indecomposable
  matrix}\sdefisub{}{matrix}{indecomposable} if it is not
decomposable. A \sdefi{block}{block of a matrix}\sdefisub{}{matrix}{block of a}
of $A$ is a maximal indecomposable submatrix.
Let $A$ be an $m \times m$ matrix and $G:=G(A)$ its underlying
graph. Note that every connected component of $G$ that is not
bipartite corresponds to a block of $A$, and every connected component
that is bipartite corresponds to two blocks $B,B^T$ arranged as
in $\left(\begin{array}{c c}
             0 & B \\
             B^T & 0
          \end{array}\right).
$

The proof of Theorem~\ref{thm:bg} falls into two parts, corresponding
to the following two lemmas.

\begin{lemma}[Polynomial Time Solvable Cases]\label{lem:rank1_FP}
Let $A \in \Ring^{m\times m}$ be symmetric. If each block of $A$ has row rank at most $1$ then $\eval(A)$ is polynomial time computable.
\end{lemma}
We leave the proof of this lemma as an exercise for the reader. The
essential ideas of its proof were explained in 
Section~\ref{sec:complexity} before the statement of
Theorem~\ref{thm:bg}.

\begin{lemma}[\#\PP-hard Cases]\label{lem:bg05_hardness_part}
Let $A \in \Ring^{m\times m}$ be symmetric and non-negative. If $A$ contains a block of row rank at least $2$ then $\eval(A)$ is \#\PP-hard.
\end{lemma}
Theorem~\ref{thm:bg} directly follows from
Lemmas~\ref{lem:rank1_FP} and \ref{lem:bg05_hardness_part}. The
remainder of this paper is devoted to a proof of
Lemma~\ref{lem:bg05_hardness_part}.

For technical reasons, we need to introduce an extended version of
partition functions with vertex weights. Several different restricted
flavors of these will be used in the proof to come.  Let $A\in
\Ralg^{m \times m}$ be a symmetric matrix and $D \in \Ralg^{m \times m}$ a
diagonal matrix. Let $G = (V,E)$ be some given graph. Recall that a
\sdef{configuration} is a mapping $\vcfg: V \rightarrow [m]$ which
assigns a spin to \emph{every} vertex of $G$. By contrast, a
\sdef{pinning} of (vertices of) $G$ with respect to $A$ is a mapping
$\vpin: W \rightarrow [m]$ for some subset $W \subseteq V$.  Whenever
the context is clear, we simply speak of pinnings and configurations
without mentioning the matrices $A$,$D$ and the graph $G$ explicitly.
We define the \sdef{partition function} on $G$ and $\vpin$ by
$$
Z_{A,D}(\vpin,G) = \sum_{\vpin \subseteq \vcfg: V \rightarrow [m]} \prod_{uv \in E} A_{\vcfg(u),\vcfg(v)} \prod_{v \in V \setminus{\df(\vpin)}} D_{\vcfg(v),\vcfg(v)}
$$
where $\df(\vpin)$ denotes the domain of $\vpin$.
Note that the sum is over all configurations $\vcfg: V \rightarrow [m]$ which extend the fixed given pinning $\vpin$. In the presence of a pinning $\vpin$ we denote the \sdefi{weight}{weight of a configuration}\sdefisub{}{configuration}{weight of a} of the configuration $\vcfg$ by the following term 
$$
\prod_{uv \in E} A_{\vcfg(u),\vcfg(v)} \prod_{v \in V \setminus \df(\vpin)} D_{\vcfg(v),\vcfg(v)}.
$$
Note that, for technical reasons, the terms $D_{\vcfg(v),\vcfg(v)}$ for $v \in \df(\vpin)$ are excluded from this weight. Whenever $\vpin$ is empty in the sense that $\df(\vpin) = \emptyset$ then we say that it is \sdefi{trivial}{trivial pinning}. In this case, its appearance in the above expression is vacuous. This is analogously true for $D$ if it is the identity matrix. In either of these cases we omit the terms $D$ ($\vpin$, respectively) in the expression. For example if $\df(\vpin) = \emptyset$ and $D = I_m$, then
$$
Z_{A}(G) = \sum_{\vcfg: V \rightarrow [m]} \prod_{uv \in E} A_{\vcfg(u),\vcfg(v)}.
$$
The definitions of $Z_{A,D}(G)$ and $Z_{A}(\vpin,G)$ are analogous. We define $\evalk(A,D)$ as the computational problem of computing $Z_{A,D}(\vpin, G)$ on input $\vpin, G$. Similarly $\eval(A,D)$ restricts the inputs to empty pinnings and $\evalk(A)$ denotes the problem where $D$ is the identity matrix.

We will now explain the overall structure of the proof of
Lemma~\ref{lem:bg05_hardness_part}.
In a first step we will see how we can augment our capabilities so as to fix some vertices of the input graphs, without changing the complexity of the problems under consideration (cf. Lemma \ref{lem:pinning}). Then in the General Conditioning Lemma~\ref{lem:new_gen_cond} we will show, that we can reduce the abundance of non-negative matrices to certain well-structured cases. From these we will show in two steps (Lemmas~\ref{lem:two-1-cell} and \ref{lem:sing_1_cell}) how to obtain \#\PP-hardness.

\begin{lemma}[Pinning Lemma]\label{lem:pinning}
Let $A\in \Ralg^{m\times m}$ be a symmetric non-negative matrix.  Then
$$
  \evalk(A) \Tequiv \eval(A).
$$
\end{lemma}
A \sdef{$1$-cell} in a matrix $A \in \Ring^{m \times m}$ is a submatrix $A_{IJ}$ such that $A_{ij} = 1$ for all $(i,j) \in I \times J$ and $A_{ij} \neq 1$ for all $(i,j) \in (\bar I \times J) \cup (I \times \bar J)$.
For a number or an indeterminate $X$ an \sdef{$X$-matrix} is a matrix whose entries are powers of $X$.
\paragraph*{General Conditions.} For a matrix $A \in \Ring^{m \times m}$ we define conditions
\begin{condition}{(A)} $A$ is symmetric positive and has $\rank{A} \ge 2.$\end{condition}
\begin{condition}{(B)} $A$ is an $X$-matrix for an indeterminate $X$.\end{condition}
\begin{condition}{(C)} There is a $k \ge 2$, numbers $1=m_0 < \ldots
  <m_k = m+1$ and $I_{i} = [m_{i-1}, m_i -1]$ for all $i\in [k]$ such
  that $A_{I_iI_i}$ is a $1$-cell for every $i \in
  [k-1]$. The matrix $A_{I_kI_k}$ may or may not be a $1$-cell. Furthermore,
  all $1$-entries are contained in one of these
  $1$-cells.
\end{condition}

\begin{lemma}[General Conditioning Lemma]\label{lem:new_gen_cond}
Let $A \in \Ring^{m\times m}$ be a non-negative symmetric matrix which contains a block of rank at least $2$.
Then there is a $\Int[X]$-matrix $A'$ satisfying conditions \cond{A}--\cond{C} such that
$$
  \evalk(A') \Tle \evalk(A).
$$
\end{lemma}
If a matrix $A$ satisfies the General Conditions two different cases can occur which we will treat separately in the following.
The first case is the existence of at least two $1$-cells. 

\begin{lemma}[Two $1$-Cell Lemma]\label{lem:two-1-cell}
Let $A \in \Int[X]^{m \times m}$ be a positive symmetric matrix containing at least two $1$-cells. Then
$\evalk(A)$ is $\#\PP$-hard.
\end{lemma}
The second case is then the existence of only one $1$-cell. The proof of this case is much more involved than the first one.

\begin{lemma}[Single $1$-Cell Lemma]\label{lem:sing_1_cell}
Let $A \in \Int[X]^{m \times m}$ be a matrix which satisfies conditions \cond{A} -- \cond{C} and has exactly one $1$-cell. Then $\evalk(A)$ is $\#\PP$-hard.
\end{lemma}
Once these results have been derived, it will be easy to prove the main result.

\begin{proof}[of Lemma \ref{lem:bg05_hardness_part}]
Let $A \in \Ralg^{m \times m}$ be a non-negative symmetric matrix which contains a block of rank at least $2$. 
By the General Conditioning Lemma~\ref{lem:new_gen_cond} there is a matrix $A'$ satisfying conditions 
\cond{A}--\cond C such that $\evalk(A') \Tle \evalk(A)$. 
If $A'$ contains a single $1$-cell then $\evalk(A')$ is $\#\PP$-hard by Lemma~\ref{lem:sing_1_cell}. 
Otherwise, it is $\#\PP$-hard by Lemma~\ref{lem:two-1-cell}. In both cases
this proves $\#\PP$-hardness of $\eval(A)$ by means of Lemma~\ref{lem:pinning}.
\end{proof}

\section{Evaluation and counting}
\label{sec:ec}
Let $A$ be an $m \times m$ matrix. Then, for every $q \in \Qu$ we define the matrix $A^{(q)}$ by
$$
A^{(q)}_{ij} = \left\{ \begin{array}{l l}
                        (A_{ij})^q &, \text{ if } A_{ij} \neq 0 \\
                          0 &, \text{ otherwise.}
                       \end{array}\right.
$$
The following lemma provides two basic reductions which form basic building blocks of many hardness proofs.
\begin{lemma}\label{lem:basic_reductions}
Let $A \in \Ring^{m\times m}$ the following is true for every $p \in \Nat$ 
\begin{description}
\item[\textbf{($p$-thickening)}]  $\evalk(A^{(p)}) \Tle \evalk(A).$
\item[\textbf{($p$-stretching)}]  $\evalk(A^{p}) \Tle \evalk(A).$
\end{description}
\end{lemma}
\begin{proof}
Let $G = (V,E)$ be a graph and $\vpin$ a pinning. Let the \emph{$p$-thickening} $G^{(p)}$ of $G$ be 
the graph obtained from $G$ by replacing each edge by $p$ many parallel edges. 
The \emph{$p$-stretching} $G_{p}$ of $G$ is 
obtained from $G$ by replacing each edge by a path of length $p$.
The reductions follow from the (easily verifiable) identities
$ Z_{A^{(p)}}(\vpin,G) = Z_{A}(\vpin,G^{(p)})$ and
$ Z_{A^p}(\vpin,G) = Z_{A}(\vpin,G_{p})$.
\end{proof}

\begin{lemma}\label{lem:a_square_prop} 
Let $A \in \Ring^{m \times n}$ be non-negative. If $A$ contains a block of rank $\ge 2$ then so does $A A^T$.
\end{lemma}
\begin{proof}
Let $A_{IJ}$ be a block of rank at least $2$ in $A$ and define $A' = AA^T$. We claim that $A'_{II}$ contains a block of rank at least $2$, which clearly implies the statement of the lemma. 

For $i, i' \in I$ let $A\row i$ and $A\row {i'}$ be two linearly independent rows.
Since $A_{IJ}$ is a block, there are indices $i = i_1,\ldots,i_k = i'$ in $I$ and $j_1,\ldots j_{k-1} \in J$ such that $A_{i_\nu j_{\nu}} \neq 0$ and $A_{i_{\nu+1} j_{\nu}} \neq 0$ for all $\nu \in [k-1]$. Linear independence of $A\row {i}$ and $A\row {i'}$ implies that there is a $\mu \in [k-1]$ such that also $A\row{i_{\mu}}$ and $A\row{i_{\mu+1}}$ are linearly independent.
We claim that the submatrix
$$
\left(\begin{array}{l l}
       A'_{i_{\mu}, i_{\mu}} & A'_{i_{\mu}, i_{\mu+1}} \\
       A'_{i_{\mu+1}, i_{\mu}} & A'_{i_{\mu+1}, i_{\mu+1}}
      \end{array}\right)
=
\left(\begin{array}{l l}
       \scalp{A\row{i_{\mu}}, A\row{i_{\mu}}} & \scalp{A\row{i_{\mu}}, A\row{i_{\mu+1}}} \\
       \scalp{A\row{i_{\mu+1}}, A\row{i_{\mu}}} & \scalp{A\row{i_{\mu+1}}, A\row{i_{\mu+1}}}
      \end{array}\right)
$$
is a witness for the existence of a block of rank at least two in $A'$. To see this note first that, by definition, all entries of this submatrix are positive, it thus remains to show that this submatrix has rank $2$. Assume, for contradiction, that it has zero determinant, that is
$$
\scalp{A\row{i_{\mu}}, A\row{i_{\mu}}}\scalp{A\row{i_{\mu+1}}, A\row{i_{\mu+1}}} = \scalp{A\row{i_{\mu+1}}, A\row{i_{\mu}}}^2.
$$
The Cauchy Schwarz inequality therefore implies linear dependence of $A\row{i_{\mu}}$ and $A\row{i_{\mu+1}}$. Contradiction.
\end{proof}
For $A$ an $m \times m$ matrix and $\pi : [m] \rightarrow [m]$ a permutation, define \idxsymb{$A_{\pi\pi}$} by
$$
(A_{\pi\pi})_{ij} = A_{\pi(i)\pi(j)} \text{ for all } i,j \in [m].
$$
The \emph{Permutability Principle} states that for any evaluation problem on some matrix $A$, we may assume any simultaneous permutation of the rows and columns of this matrix. Its proof is straightforward.

\begin{lemma}[Permutability Principle] \label{lem:principle_permute}
Let $A,D \in \Ring^{m \times m}$ and $\pi : [m] \rightarrow [m]$ a permutation. Then $\evalk(A,D) \Tequiv \evalk(A_{\pi\pi},D_{\pi\pi})$.
\end{lemma}
We will make extensive use of interpolation; the following simple
lemma is one instance.

\begin{lemma} \label{lem:interpolate}
For some fixed $\theta \in \Ralg$ let $x_1, \ldots, x_n \in \Q(\theta)$ be pairwise different and non-negative reals. Let $b_1,\ldots,b_n \in \Q(\theta)$ be arbitrary such that
$$
 b_j = \sum_{i=1}^n c_i x^j_i \text{ for all } j \in [n].
$$
Then the coefficients $c_1,\ldots, c_n$ are uniquely determined and can be computed in polynomial time.
\end{lemma}

\subsection{The Equivalence of $\evalk(A)$ and $\cntk(A)$}
\newcommand{\svcfg}{\textup{CFG}}
Let $A \in \Ring^{m \times m}$ be a matrix, $G=(V,E)$ a graph and $\vpin$ a pinning.
We define a set of \sdefis{potential weights}{potential weights}{$\wset_A(G)$}
\begin{equation}\label{eq:define_wset}
 \wset_{A}(G) := \left\lbrace \prod_{i,j \in [m]} A^{m_{ij}}_{ij} \,\mid \, \sum_{i,j \in [m]} m_{ij} = \vert E \vert,\;  \text{ and } m_{ij} \ge 0, \text{ for all } i,j \in [m]  \right\rbrace.
\end{equation}
For every $w \in \Ring$ define the value
$$
N_{A}(G,\vpin,w) := \left\vert \left\lbrace\vcfg: V \rightarrow [m]\;\left|\; \vert \, \vpin \subseteq \vcfg,\; w = \prod_{uv \in E} A_{\vcfg(u)\vcfg(v)}\right.\right\rbrace \right\vert.
$$ 
By $\cntk(A)$ we denote the problem of computing $N_{A}(G,\vpin,w)$ for given $G, \vpin$ and $w \in \Ring$. In analogy to the evaluation problems we write $\cnt(A)$ for the subproblem restricted to instances with trivial pinnings. It turns out that these problems are computationally equivalent to the evaluation problems of partition functions.

\begin{lemma}\label{lem:cntk_eq_evalk}
For every matrix $A \in \Ring^{m\times m}$ we have 
$$\evalk(A) \Tequiv \cntk(A) \quad \text{ and } \quad \eval(A) \Tequiv \cnt(A).
$$ 
\end{lemma}
\begin{proof}
Let $G = (V,E)$ be a graph and $\vpin$ a pinning. We have
\begin{equation*}
Z_{A}(\vpin, G) = \sum_{\vpin \subseteq \vcfg : V \rightarrow [m]} \prod_{uv \in E} A_{\vcfg(u)\vcfg(v)} = \sum_{w \in \wset_{A}(G)} w \cdot N_{A}(G,\vpin,w). 
\end{equation*}
As the cardinality of $\wset_{A}(G)$ is polynomial in the size of $G$ this proves the reducibilities
$$
\evalk(A) \Tle \cntk(A)  \quad \text{ and } \quad \eval(A) \Tle \cnt(A).
$$
For the backward direction let $G^{(t)}$ denote the graph obtained
from $G$ by replacing each edge by $t$ parallel edges. We have
$$
Z_{A}(\vpin, G^{(t)}) =  \sum_{\vpin \subseteq \vcfg: V \rightarrow [m]}\left(\prod_{uv \in E} A_{\vcfg(u)\vcfg(v)}\right)^t \\
 = \sum_{w \in \wset_{A}(G)} w^t \cdot N_{A}(G,\vpin ,w).
$$
Using an $\evalk(A)$ oracle, we can evaluate this for $t = 1, \ldots, |\wset_{A}(G)|$. 
Therefore, if $\Ring$ is one of $\Ralg,\Q, \Int$, the values $N_{A}(G,\vpin, w)$ can be recovered in polynomial time by Lemma~\ref{lem:interpolate}.

If $\Ring$ is one of $\Int[X]$ or $\Qu[X]$, then let $f_t(X)$ denote $Z_{A}(\vpin, G^{(t)})$. We obtain the following system of equations, for $t = 1, \ldots, |\wset_{A}(G)|$,
\begin{equation}\label{eq:0105092143}
f_t(X) = \sum_{w \in \wset_{A}(G)} w^t(X) \cdot N_{A}(G,\vpin ,w).
\end{equation}
Let $\det(X)$ denote the determinant of this system of equations. W.l.o.g., all the $w(X) \in \wset_A(G)$ are non-zero polynomials, $\det(X)$, as it is a Vandermonde determinant in these $w(X)$, is itself a non-zero polynomial of the form
$$
\det(X) = \prod_{w \in \wset_A(G)} w(X) \cdot \prod_{w \neq w' \in \wset_A(G)} (w(X) - w'(X)).
$$
Let $\delta = 1 + \max\{\deg w(X) \mid w \in \wset_A(G) \}$ and observe that each $w \in \wset_A(G)$ has at most $\delta-1$ roots. Further, each of the terms $w(X) - w'(X)$ of $\det(X)$ has degree at most $\delta - 1$ and thus the degree of $\det(X)$ is strictly smaller than $\delta' := \binom{\vert \wset_A(G) \vert + 1}{2} \cdot \delta$. Hence there is an integer $a \le \delta'$ such that $\det(a)$ is non-zero.
By equation \eqref{eq:0105092143} we obtain an invertible system of equations
\begin{equation}\label{eq:0105092144}
f_t(a) = \sum_{w \in \m W_{A}(G)} w^t(a) \cdot N_{A}(G,\vpin ,w).
\end{equation}
The coefficients $N_{A}(G,\vpin ,w)$ can be now obtained in polynomial time by Lemma~\ref{lem:interpolate}.
This finishes the proof of $\cntk(A) \Tequiv \evalk(A)$.
The proof for $\cnt(A) \Tequiv \eval(A)$ also follows by the argument just presented, as the given pinnings remain unaffected.
\end{proof}

\subsection{Dealing with Vertex Weights}

\begin{lemma}[Theorem 3.2 in \cite{dyegre00}]\label{lem:dg_omit_vertexweights}
Let $A \in \Ralg^{m \times m}$ be a symmetric matrix with non-negative entries such that every pair of rows in $A$ is linearly independent. Let $D \in \Ralg^{m \times m}$ be a diagonal matrix of positive vertex weights. Then $$\evalk(A) \Tle \evalk(A,D) \text{ and } \eval(A) \Tle \eval(A,D).$$ 
\end{lemma} 
After some preparation, we will prove this lemma in Section~\ref{sec:proof_omit_vv}.
\subsubsection{Some Technical Tools}
The following is a lemma from \cite{dyegre00} (Lemma~3.4).
\begin{lemma}\label{lem:DG__eig_interpolate}
Let $A\in \Ralg^{m \times m}$ be symmetric and non-singular, $G =(V,E)$ a graph, $\vpin$ a pinning, and $F \subseteq E$.
If we know the values
\begin{equation}\label{eq:3012091445}
f_r(G) = \sum_{\vpin \subseteq \vcfg: V \rightarrow [m]} c_A(\sigma) \prod_{uv \in F} A^{r}_{\vcfg(u) \vcfg(v)}
\end{equation}
for all $r \in [(|F| + 1)^{m^2}]$, where $c_A$ is a function depending on $A$ but not on $r$. Then we can evaluate
$$
\sum_{\vpin \subseteq \vcfg: V \rightarrow [m]} c_A(\sigma) \prod_{uv \in F} (I_m)_{\vcfg(u) \vcfg(v)}
$$
in polynomial time.
\end{lemma}
\begin{proof}
As $A$ is symmetric and non-singular, there is an orthogonal matrix $P$ such that $P^T A P =: D$ is a diagonal matrix with non-zero diagonal.
From $A = P D P^T$ we have $A^r = P D^{r} P^T$ and thus every entry of $A^r$ satisfies
$A^r_{ij} = (P D^{r} P^T)_{ij} = \sum_{\mu = 1}^m P_{\vcfg(u)\mu} P_{ \vcfg(v) \mu} (D_{\mu\mu})^r$.
Hence equation \eqref{eq:3012091445} can be rewritten as
\begin{eqnarray*}
f_r(G) 
       &=& \sum_{\vpin \subseteq \vcfg: V \rightarrow [m]} c_A(\sigma) \prod_{uv \in F} \sum_{\mu = 1}^m P_{\vcfg(u)\mu} P_{ \vcfg(v) \mu} (D_{\mu\mu})^r 
\end{eqnarray*}
Define the set
$$
\mathcal{W} = \left\{ \prod_{i=1}^m (D_{ii})^{\alpha_i} \mid 0 \le \alpha_i \text{ for all } i \in [m], \; \sum_{i=1}^m \alpha_i = |F|\right\}
$$
which can be constructed in polynomial time. We rewrite
$$
f_r(G) = \sum_{w \in \mathcal{W}} c_w w^r.
$$
for unknown coefficients $c_w$. By interpolation (cf. Lemma~\ref{lem:interpolate}), we can recover these coefficients in polynomial time and can thus calculate $f_0(G) = \sum_{w \in \mathcal{W}} c_w$. We have
\begin{eqnarray*}
f_0(G) &=& \sum_{\vpin \subseteq \vcfg: V \rightarrow [m]} c_A(\sigma) \prod_{uv \in F} \sum_{\mu = 1}^m P_{\vcfg(u)\mu} P_{ \vcfg(v) \mu} (D_{\mu\mu})^0 \\
       &=& \sum_{\vpin \subseteq \vcfg: V \rightarrow [m]} c_A(\sigma) \prod_{uv \in F} I_{\vcfg(u)\vcfg(v)}
\end{eqnarray*}
which finishes the proof.
\end{proof}
The following two lemmas are restatements of those in \cite{dyegre00} (see Lemma~3.6, 3.7 and Theorem~3.1).

\begin{lemma}\label{lem:a_to_ada}
Let $A \in \Real^{m\times m}$ be a symmetric matrix in which every pair of distinct rows is linearly independent. Let $D \in \Real^{m \times m}$ be a diagonal matrix of positive vertex weights. Then every pair of rows in $ADA$ is linearly independent.
Furthermore there is an $0 < \epsilon < 1$ such that for all $i\neq j$
$$
|(ADA)_{ij}| \le \epsilon \sqrt{(ADA)_{ii} (ADA)_{jj}}
$$
\end{lemma}
\begin{proof}
Define $Q = AD^{(1/2)}$. We have $ADA = AD^{(1/2)} D^{(1/2)}A^T = QQ^T$. That is $(ADA)_{ij} = \scalp{Q\row i, Q\row j}$ every pair of rows in $Q$ is linearly independent as it is linearly independent in $A$. Hence, by the Cauchy-Schwarz inequality
$$
\scalp{Q\row i, Q\row j} < \sqrt{\scalp{Q\row i, Q\row i}\scalp{Q\row j, Q\row j}}
$$
which implies that the corresponding $2\times 2$ submatrix of $ADA$ defined by $i$ and $j$ has non-zero determinant.
The existence of $\epsilon$ follows.
\end{proof}

\begin{lemma}\label{lem:make_non-sing}
Let $A\in \Real^{m \times m}$ be a symmetric non-negative matrix in which every pair of distinct rows is linearly independent. Let $D\in \Real^{m \times m}$ be a diagonal matrix of positive vertex weights. Then there is a $p \in \Nat$ such that the matrix
$$
(ADA)^{(p)}
$$
is non-singular.
\end{lemma}
\begin{proof}
Let $A' = ADA$ and consider the determinant
$$
\det(A') = \sum_{\pi \in S_m} \textup{sgn}(\pi) \prod_{i=1}^m A'_{i\pi(i)}
$$
where $S_m$ is the set of permutations of $[m]$. For some $\pi \in S_m$ define $t(\pi) = |\{i \mid \pi(i) \neq i\}|$. Let $\epsilon$ be as in Lemma~\ref{lem:a_to_ada}. Then
\begin{equation}\label{eq:2606091840}
\prod_{i=1}^m| A'_{i\pi(i)}| \le \epsilon^{t(\pi)} \prod_{i=1}^m \sqrt{A'_{ii}}  \prod_{i=1}^m \sqrt{ A'_{\pi(i)\pi(i)}} = \epsilon^{t(\pi)} \prod_{i=1}^m A'_{ii}. 
\end{equation}
Let $\idfunc$ denote the trivial permutation. Then
\begin{eqnarray*}
\det((A')^{(p)}) &\ge& \left(\prod_{i=1}^m A'_{ii}\right)^p - \sum_{\pi \in S_m\setminus\{\idfunc\}} \left(\prod_{i=1}^m A'_{i\pi(i)}\right)^p.
\end{eqnarray*}
By equation \eqref{eq:2606091840}, we have
\begin{eqnarray*}
  m!\epsilon^p\left(\prod_{i=1}^m A'_{ii}\right)^p &\ge& \sum_{\pi \in S_m\setminus\{\idfunc\}} \left(\prod_{i=1}^m A'_{i\pi(i)}\right)^p
\end{eqnarray*}
and hence, as $0 < \epsilon < 1$, the matrix $(ADA)^{(p)}$ is non-singular for large enough $p$.
\end{proof}

\subsubsection{The proof of Lemma \ref{lem:dg_omit_vertexweights}.}\label{sec:proof_omit_vv}
By Lemma~\ref{lem:make_non-sing} there is a $p \in \Nat$ such that $(ADA)^{(p)}$ is non-singular. We will fix such a $p$ for the rest of the proof.

Let $G,\vpin$ be an instance of $\evalk(A)$ with $G = (V,E)$ a graph. Construct from $G$ a graph $G'$ as follows.
Define
\begin{eqnarray*}
V'  &=& \{ v_0,\ldots,v_{d-1} \mid v\in V, d= \vdeg_G(v)\}
\end{eqnarray*}
that is, for each vertex $v \in V$ we introduce $d(v)$ new vertices. 
Let $\vpin'$ be the pinning which for every $v \in \df(\vpin)$ satisfies $\vpin'(v_0) = \vpin(v)$.
Let $E'''$ be the set which contains, for each $v\in V$ the cycle on $\{v_0,\ldots, v_{\vdeg_G(v)-1}\}$, i.e. $\{v_0v_1, \ldots, v_{\vdeg_G(v)-1}v_0\} \subseteq
E'''$. Let $E''$ be the set of edges, such that each edge from $E$ incident with $v$ is connected to exactly one of the $v_i$. Define $E' = E'' \cup E'''$ and denote by $G'_{p,r}$ the graph obtained from $G'$ by replacing each edge in $E'''$ with a distinct copy of the graph $T_{p,r}$ to be defined next. 

\begin{figure}
\begin{center}
\begin{minipage}{0.75\textwidth}
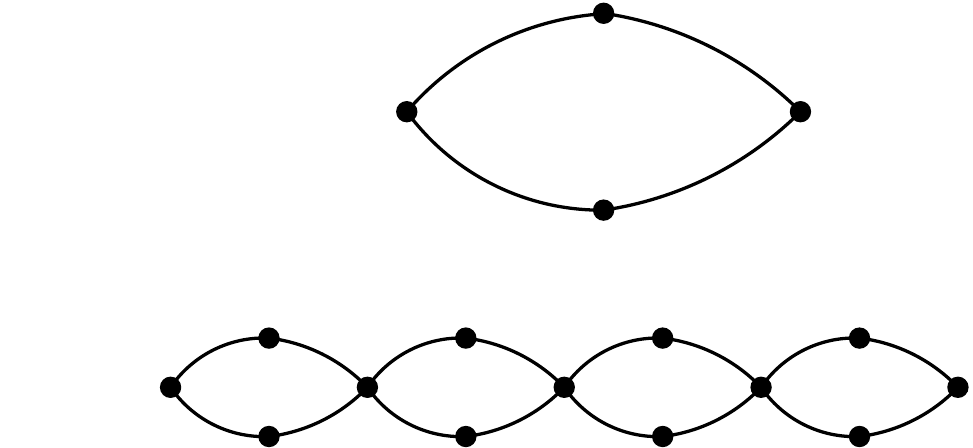
\end{minipage}
\end{center}
\caption{The graphs $T_2$ and $T_{2,4}$.
}
\label{fig:bounded_deg}
\end{figure}
To conveniently define the graph $T_{p,r}$ we will first describe another graph $T_p$. This graph consists of two distinguished vertices $a$ and $b$ connected by $p$ many length $2$ paths from $a$ to $b$. Then $T_{p,r}$ is the series composition of $r$ many copies of $T_p$. The construction is illustrated in Figure~\ref{fig:bounded_deg}.
We call $a$ and $b$ the "start" and "end" vertex of $T_p$ and $T_{p,r}$ has start and end vertices induced by the start and end vertices of the series composition of $T_p$. 

For a graph $H$ with designated ``start'' and ``end'' vertex we denote by
$Z_{A,D}(i,j; H)$ the partition function $Z_{A,D}(\phi, H)$ for $\phi$ such that it pins the start vertex of $H$ to $i$ and the end vertex to $j$. Note that, by definition, the vertex weights of the pinned vertices do not occur in this term.

\begin{claim}
Let $C = (ADA)^{(p)}$, then with $X = D^{(1/2)}$, we have for all $i,j \in [m]$ and $r \in \Nat$,
\begin{equation}\label{eq:2303091755}
Z_{A,D}(i,j; T_{p,r}) = (X_{ii}X_{jj})^{-1}((XCX)^r)_{ij}.
\end{equation}
\end{claim}
\begin{clproof}
Straightforwardly,
\begin{eqnarray*}
Z_{A,D}(i,j; T_p) &=& \left(\sum_{k=1}^m A_{ik}A_{kj}D_{kk}\right)^p = (ADA)^{(p)}_{ij} = C_{ij}
\end{eqnarray*}
and therefore
\begin{eqnarray*}
Z_{A,D}(i,j; T_{p,r}) &=& \sum_{\substack{\vcfg:[r+1] \rightarrow [m] \\ \vcfg(1)=i,\,\vcfg(r+1)=j}}
\prod_{k=1}^r Z_{A,D}(\vcfg(k),\vcfg(k+1);T_p) \prod_{k=2}^r D_{\vcfg(k),\vcfg(k)} \\
&=& \sum_{\substack{\vcfg:[r+1] \rightarrow [m] \\ \vcfg(1)=i,\,\vcfg(r+1)=j}}
\prod_{k=1}^r C_{\vcfg(k),\vcfg(k+1)} \prod_{k=2}^r (X_{\vcfg(k),\vcfg(k)})^2\\
&=& (X_{ii}X_{jj})^{-1} \sum_{\substack{\vcfg:[r+1] \rightarrow [m] \\ \vcfg(1)=i,\,\vcfg(r+1)=j}} \prod_{k=1}^r X_{\vcfg(k),\vcfg(k)} C_{\vcfg(k),\vcfg(k+1)} X_{\vcfg(k+1),\vcfg(k+1)} 
\end{eqnarray*}
By inspection, the last line equals the right hand side of equation \eqref{eq:2303091755} --- as claimed.
\end{clproof}
Using the expression just obtained for $Z_{A,D}(i,j; T_{p,r})$, we will now rewrite $Z_{A,D}(\vpin',G'_{p,r})$. 
To do so, we will, for all $v \in V$, count the indices of the vertices $v_{0}, \ldots, v_{\vdeg_G(v) - 1}$ modulo $\vdeg_G(v)$. In particular $v_{\vdeg_G(v)} = v_0$. 
Define $\gamma= \prod_{v\in \df(\vpin')} D_{\vpin'(v),\vpin'(v)}$. We have for all $r \in \Nat$,
\begin{eqnarray*}
Z_{A,D}(\vpin',G'_{p,r}) = \gamma^{-1} \sum_{\vpin' \subseteq \vcfg:V' \rightarrow [m]} \prod_{uv \in E''} A_{\vcfg(u)\vcfg(v)} \prod_{v\in V'} D_{\vcfg(v),\vcfg(v)} \prod_{v\in V} \prod_{i=0}^{\vdeg_G(v) - 1} Z_{A,D}(\sigma(v_i),\sigma(v_{i+1});T_{p,r}).
\end{eqnarray*}
As the vertices in $V'$ are grouped according to the vertices in $V$ we have
$$
\prod_{v\in V'} D_{\vcfg(v),\vcfg(v)} = \prod_{v\in V}\prod_{i=0}^{\vdeg_G(v)-1} D_{\vcfg(v_i),\vcfg(v_i)}.
$$
Using equation \eqref{eq:2303091755}, we further see that for each $\vcfg : V \rightarrow [m]$,
the expression
$$
\prod_{v\in V'} D_{\vcfg(v),\vcfg(v)} \prod_{v\in V} \prod_{i=0}^{\vdeg_G(v) - 1} Z_{A,D}(\vcfg(v_i),\vcfg(v_{i+1});T_{p,r})
$$
turns into
\begin{eqnarray*}
\prod_{v\in V} \prod_{i=0}^{\vdeg_G(v) - 1} D_{\vcfg(v_i),\vcfg(v_i)} \dfrac{((XCX)^r)_{\vcfg(v_i),\vcfg(v_{i+1})}}{X_{\vcfg(v_i),\vcfg(v_i)}X_{\vcfg(v_{i+1}),\vcfg(v_{i+1})}}
&=&
\prod_{v\in V} \prod_{i=0}^{\vdeg_G(v) - 1} ((XCX)^r)_{\vcfg(v_i),\vcfg(v_{i+1})}
\end{eqnarray*}
Thus
\begin{eqnarray*}
Z_{A,D}(\vpin',G'_{p,r}) &=& \gamma^{-1}\sum_{\vpin'\subseteq\vcfg:V' \rightarrow [m]} \prod_{uv \in E''} A_{\vcfg(u)\vcfg(v)} \prod_{v\in V} \prod_{i=0}^{\vdeg_G(v) - 1} 
((XCX)^r)_{\sigma(v_{i})\sigma(v_{i+1})}\\
&=&  \gamma^{-1}\sum_{\vpin'\subseteq\vcfg:V' \rightarrow [m]} \prod_{uv \in E''} A_{\vcfg(u)\vcfg(v)} \prod_{uv\in E'''} ((XCX)^r)_{\sigma(u)\sigma(v)}
\end{eqnarray*}
Given the $\eval(A,D)$ oracle, we can, in polynomial time, evaluate this expression for every $r$ which is polynomial in the size of $G$. Therefore Lemma~\ref{lem:DG__eig_interpolate} implies that we can compute the value
\begin{eqnarray*}
Z &=&  \gamma^{-1}\sum_{\vpin'\subseteq\vcfg:V' \rightarrow [m]} \prod_{uv \in E''} A_{\vcfg(u)\vcfg(v)} \prod_{uv\in E'''} I_{\sigma(u)\sigma(v)} 
\end{eqnarray*}
in polynomial time.
The proof follows, if we can show that $\gamma \cdot Z = Z_{A}(\vpin,G)$. To see this, note that, for every configuration $\vcfg:V' \rightarrow [m]$ the corresponding weight $\prod_{uv \in E''} A_{\vcfg(u)\vcfg(v)} \prod_{uv\in E'''} I_{\sigma(u)\sigma(v)}$ in the above expression is zero unless the following holds:
For all $v \in V$ we have $\vcfg(v_0) = \ldots = \vcfg(v_{d-1})$ for $d = d_G(v)$.
for such a configuration, define a configuration $\vcfg' : V \rightarrow [m]$ such that $\vcfg(v) = \vcfg(v_0)$ for every $v \in V$.
It follows from the construction of $G'$ that then 
$$
\prod_{uv \in E''} A_{\vcfg(u),\vcfg(v)} = \prod_{uv \in E} A_{\vcfg'(u),\vcfg'(v)}.
$$
Since every configuration $\vcfg':V \rightarrow [m]$ arises this way, we have $\gamma \cdot Z = Z_{A}(\vpin,G)$.
This finishes the proof of $\evalk(A) \Tle \evalk(A,D)$. Since the proof is correct also for empty input pinnings $\vpin$, this also proves $\eval(A) \Tle \eval(A,D)$.

\section{The Pinning Lemma}\label{sec:pin}
In this section we shall give the proof of the Pinning Lemma~\ref{lem:pinning}. Before we do this, however, we will introduce a technical reduction which will also be used later.

\subsection{The Twin Reduction Lemma}\label{sec:twin_resolution}
For a symmetric $m\times m$ matrix $A$, we say that two rows $A \row i$ and $A \row j$ are \sdefi{twins}{twin}, if $A \row i = A \row j$. A matrix $A$ is \sdef{twin-free} if $A\row i \neq A\row j$ for all row indices $i \neq j$.
The concept of twins induces an equivalence relation on the rows of $A$. Let $I_1,\ldots ,I_k$ be the equivalence classes of this relation. The \sdef{twin resolvent} of $A$ is the $k \times k$ matrix $\twres{A}$, defined by
$$
\twres{A}_{i,j} = A_{\mu,\nu} \text{ for some } \mu \in I_{i} \text{ and } \nu \in I_{j}.
$$
We say that $\twres A$ is obtained from $A$ by \sdef{twin reduction}.
Further, we define the \sdef{twin resolution mapping} $\tau: [m] \rightarrow
[k]$ of $A$ in such a way that for all $\mu \in [m]$ we have $\mu \in I_{\tau(\mu)}$. That is, $\tau$ maps every $\mu \in [m]$ to the class $I_j$ it is contained in. Hence
\begin{equation}
 \twres A_{\tau(i),\tau(j)} = A_{i,j} \text{ for all } i,j \in [m].
\end{equation}
To use twin reductions in the context of partition functions we need to clarify the effect which twin reduction on a matrix $A$ induces on the corresponding diagonal matrix of vertex weights $D$. We will see that the diagonal $k \times k$ matrix $D^{\twres A}$ defined in the following captures this effect
 \begin{equation}
     D^{\twres A}_{i,i} = \sum_{\nu \in I_i} D_{\nu,\nu} \text{ for all } i \in [k].
 \end{equation}

\begin{lemma}[Twin Reduction Lemma]\label{lem:twin_red}
Let $A\in \Ralg^{m\times m}$ be non-negative and symmetric and $D \in \Ralg^{m\times m}$ a diagonal matrix of positive vertex weights. Let $\twres{A}$ be the twin resolvent of $A$. Then 
$$ \eval(\twres{A},D^{\twres{A}}) \Tequiv \eval(A,D) \;\text{ and }\; \evalk(\twres{A},D^{\twres{A}}) \Tequiv \evalk(A,D).$$
\end{lemma}
\begin{proof}
Let $\tau$ be the twin-resolution mapping of $A$. Let $G = (V,E)$ be a graph and $\vpin$ a pinning. Define $V' = V \setminus \df(\vpin)$. By the definition of $\tau$ we have
\begin{eqnarray*}
Z_{A,D}(\phi,G) &=& \sum_{\phi \subseteq \vcfg:V \rightarrow [m]} \prod_{uv \in E} A_{\vcfg(u),\vcfg(v)} \prod_{v \in V'} D_{\vcfg(v),\vcfg(v)} \\
&=& \sum_{\phi \subseteq \vcfg:V \rightarrow [m]} \prod_{uv \in E} \twres A_{\tau\circ\vcfg(u),\tau \circ \vcfg(v)} \prod_{v \in V'} D_{\vcfg(v),\vcfg(v)}.
\end{eqnarray*}
For all configurations $\vcfg: V \rightarrow [m]$ we have $\tau \circ \vcfg : V \rightarrow [k]$. Hence, we can partition the configurations $\vcfg$ into classes according to their images under concatenation with $\tau$ and obtain
\begin{eqnarray*}
Z_{A,D}(\vpin,G) &=& \sum_{\vcfg':V \rightarrow [k]}\sum_{\substack{\vpin \subseteq \vcfg:V \rightarrow [m] \\ \tau \circ \vcfg = \vcfg'}} \prod_{uv \in E} \twres A_{\vcfg'(u), \vcfg'(v)} \prod_{v \in V'} D_{\vcfg(v), \vcfg(v)} \\
           &=& \sum_{\tau\circ \vpin \subseteq \vcfg':V \rightarrow [k]} \prod_{uv \in E} \twres A_{\vcfg'(u), \vcfg'(v)} \cdot \Delta(\vcfg').
\end{eqnarray*}
Here, $\Delta(\vcfg')$ is defined by
$$
\Delta(\vcfg') = \sum_{\substack{\phi \subseteq\vcfg:V \rightarrow [m] \\ \tau \circ \vcfg = \vcfg'}} \prod_{v \in V'} D_{\vcfg(v), \vcfg(v)}.
$$
Fix some $\vcfg': V \rightarrow [k]$, we will argue that
\begin{equation}\label{eq:0205091539}
 \Delta(\vcfg') = \prod_{v \in V'} D^{\twres A}_{\vcfg'(v), \vcfg'(v)}
\end{equation}
For every configuration $\vcfg: V \rightarrow [m]$ we have $\tau
\circ \vcfg = \vcfg'$ if, and only if, ${\vcfg'}^{-1}(\{i\}) =
\vcfg^{-1}(I_i)$ for all $i \in [k]$. Define, for each $i \in [k]$, the set $V_i := {\vcfg'}^{-1}(\{i\})$ and the mapping $\vpin_i := \vpin\upharpoonright_{\df(\vpin) \cap V_i}$. Then
\begin{eqnarray*}
\Delta(\vcfg')
&=&\sum_{\substack{ \vpin \subseteq  \vcfg: V \rightarrow [m] \\ \forall \; i \in [k]:\;\vcfg(V_i) \subseteq I_i }} \prod_{v \in V'} D_{\vcfg(v),\vcfg(v)} 
\end{eqnarray*}
Define $V'_i := V_i \setminus \df(\vpin_i)$, then
\begin{eqnarray*}
\Delta(\vcfg')
&=& \prod_{i=1}^k \sum_{\phi_i \subseteq \vcfg_i:V_i \rightarrow I_i} \prod_{v \in V'_i}  D_{\vcfg_i(v), \vcfg_i(v)}\\
 &=&\prod_{i=1}^k \prod_{v \in V'_i} \sum_{\nu \in I_i}   D_{\nu, \nu} \\
&=&\prod_{v \in V'} D^{\twres A}_{\vcfg'(v), \vcfg'(v)}\\
\end{eqnarray*}
This proves equation \eqref{eq:0205091539}. Therefore,
\begin{equation*}
Z_{A,D}(\phi, G) = \sum_{\tau\circ \phi \subseteq \vcfg':V \rightarrow [k]} \prod_{uv \in E} \twres A_{\vcfg'(u), \vcfg'(v)} \prod_{v \in V'} D^{\twres A}_{\vcfg'(v), \vcfg'(v)}
= Z_{\twres A,D^{\twres A}}(\tau \circ \phi, G).
\end{equation*}
This witnesses the claimed reducibilities.
\end{proof}

\subsection{Proof of the Pinning Lemma}

We shall prove the Pinning Lemma \ref{lem:pinning} now.
The proof of this Lemma relies on a result of Lov\'asz \cite{lov06}. To state this result we need some
further preparation. Let $A,D$ be $m \times m$ matrices and $A',D'$ be $n \times n$ matrices such that 
$D$ and $D'$ are diagonal. The pairs $(A,D)$ and $(A',D')$ are \sdef{isomorphic}, if there is a 
bijection $\alpha: [m] \to [n]$ such that $A_{ij} = A'_{\alpha(i)\alpha(j)}$ for all $i,j \in [m]$ and
$D_{ii} = D'_{\alpha(i),\alpha(i)}$ for all $i \in [m]$. An \sdef{automorphism} is an isomorphism of 
$(A,D)$ with itself.

We will moreover need to consider pinnings a bit differently than before. Rather than defining a pinning 
$\vpin$ for some given graph $G$, it will be convenient in the following to fix pinnings and consider 
graphs which are compatible with these. To define this more formally, we fix $\vpin:[k] \rightarrow [m]$ 
to denote our pinning. A \sdef{$k$-labeled graph} $G=(V,E)$ is then a graph whose vertex set satisfies $V \supseteq [k]$.
In this way $\vpin$ is compatible with every $k$-labeled graph.

\begin{lemma}[Lemma 2.4 in \cite{lov06}]\label{lem:lov06}
Let $A \in \Real^{m \times m}$ be a non-negative symmetric and twin-free matrix and $D \in \Real^{m \times m}$ a diagonal matrix of positive vertex weights. Let $\phi,\psi$ be pinnings. If $Z_{A,D}(\phi,G) = Z_{A,D}(\psi,G)$ for all $k$-labeled graphs $G$, then there is an automorphism $\alpha$ of $(A,D)$ such that $\phi = \alpha(\psi)$.
\end{lemma}
Utilizing this result, we can now prove the pinning result first for twin-free matrices.

\begin{lemma}\label{lem:pinning_twin_free}
Let $A\in \Ralg^{m\times m}$ be non-negative, symmetric, and twin-free and $D \in \Ralg^{m\times m}$ a diagonal matrix of positive vertex weights.  Then
$$
  \evalk(A,D) \Tequiv \eval(A,D).
$$
\end{lemma}
\begin{proof}
As $\eval(A,D) \Tle \evalk(A,D)$ holds trivially we only need to prove $\evalk(A,D) \Tle \eval(A,D)$.

Let $G=(V,E)$ and a pinning $\vpin$ be an instance of $\evalk(A,D)$. 
By appropriate permutation of the rows/columns of $A$ and $D$ (cf. Lemma~\ref{lem:principle_permute}) we may assume that $[k] = \textup{img } \vpin \subseteq [m]$ for some $k \le m$.
Let $\hat{G} = (\hat V, \hat E)$ be the graph obtained from $G$ by collapsing the the sets $\vpin^{-1}(i)$ for all $i \in [k]$. Formally, define a map
$$
\gamma(v) = \left\lbrace \begin{array}{l l}
                          i &, v \in \vpin^{-1}(i) \text{ for some } i \in [k] \\
                          v &, \text{ otherwise}
                         \end{array}\right.
$$
Then $\hat G$ is a $k$-labeled  multigraph (with possibly some self-loops) defined by 
\begin{eqnarray*}
\hat{V} & = & [k] \;\dot\cup\; (V \setminus \df(\vpin))\\
\hat{E} & = & \{ \gamma(u)\gamma(v) \mid uv \in E\}.
\end{eqnarray*}
Recall that in partition functions of the form $Z_{A,D}(\vpin,G)$ vertices pinned by $\vpin$
do not contribute any vertex weights.
Hence, $Z_{A,D}(\vpin,G) = Z_{A,D}(\idfunc_{[k]},\hat G)$ where $\idfunc_{[k]}$ denotes the identity map on $[k]$.
Call two mappings $\chi,\psi: [k] \rightarrow [m]$ \emph{equivalent} if there is an automorphism $\alpha$ of $(A,D)$ such that $\chi = \alpha\circ \psi$. Partition the mappings $\psi:[k]\rightarrow [m]$ into equivalence classes $I_1, \ldots, I_c$ according to this definition and for all $i \in [c]$ fix some $\psi_i \in I_i$.
Assume furthermore, that $\psi_1 = \idfunc_{[k]}$.
Clearly for any two $\chi, \psi$ from the same equivalence class, we have $Z_{A,D}(\chi,F) = Z_{A,D}(\psi,F)$ for every graph $F$. Therefore, for every graph $G'$,
\begin{equation}\label{eq:std_sum_fixation}
Z_{A,D}(G') = \sum_{i = 1}^c Z_{A,D}(\psi_i,G')\cdot \left(\sum_{\psi \in I_i} \prod_{v \in \df(\psi)} D_{\psi(v)\psi(v)}\right) 
\end{equation}
Define, for each $i \in [c]$ the value $c_i = \left(\sum_{\psi \in I_i} \prod_{v \in \df(\psi)} D_{\psi(v)\psi(v)} \right)$.
We claim the following
\begin{claim}\label{cl:2605091505}
Let $I \subseteq [c]$ be a set of cardinality at least $2$ such that $1 \in I$. Assume that we can, for every $k$-labeled graph $G'$, compute the value
\begin{equation}\label{eq:260320091543}
\sum_{i \in I} c_i\cdot Z_{A,D}(\psi_i,G').
\end{equation}
Then there is a proper subset $I' \subset I$ which contains $1$ such that we can compute, for every $k$-labeled graph $G''$, the value
\begin{equation}\label{eq:2605090815}
\sum_{i \in I'} c_i\cdot Z_{A,D}(\psi_i,G'').
\end{equation}
\end{claim}
\bigskip
This claim will allow us to finish the proof. To see this, note first that by equation \eqref{eq:std_sum_fixation} we can compute the value \eqref{eq:260320091543} for $I = [c]$ and $G' = \hat G$. Thus after at most $c$ iterations of Claim~\ref{cl:2605091505} we arrive at $c_1\cdot Z_{A,D}(\psi_1,\hat G)$. Further, $c_1$ is effectively computable in time depending only on $D$ and therefore we can compute $Z_{A,D}(\idfunc_{[k]},\hat G) = Z_{A,D}(\vpin,G)$. This proves the reducibility $\evalk(A,D) \Tle \eval(A,D)$. 

\paragraph*{Proof Of Claim~\ref{cl:2605091505}.} Assume that we can compute the value given in \eqref{eq:260320091543}.
Lemma~\ref{lem:lov06} implies that for every pair $i\neq j \in I$ there is a $k$-labeled graph $\Gamma$ such that 
\begin{equation}\label{eq:distinguish_pf}
Z_{A,D}(\psi_i,\Gamma) \neq Z_{A,D}(\psi_j,\Gamma).
\end{equation}
Fix such a pair $i\neq j \in I$ and a graph $\Gamma$ satisfying this equation. Note that this graph can be computed effectively in time depending only on $A,D$ and $\psi_i,\psi_j$. Let $G^s$ denote the graph obtained from $G$ by iterating $s$ times the $k$-labeled product of $G$ with itself. 
We can thus compute
\begin{equation}\label{eq:2603091557}
\sum_{i \in I} c_i\cdot Z_{A,D}(\psi_i,G'\Gamma^s) = \sum_{i \in I}^c c_i Z_{A,D}(\psi_i,G')\cdot Z_{A,D}(\psi_i,\Gamma)^s. 
\end{equation}
Partition $I$ into classes $J_0,\ldots ,J_z$ such that for every $\nu \in [0,z]$ we have $i',j' \in J_\nu$ if, and only if, $Z_{A,D}(\psi_{i'},\Gamma) = Z_{A,D}(\psi_{j'},\Gamma)$.
Since one of these sets $J_\nu$ contains $1$ and all of these are proper subsets of $I$, it remains to show that we can compute, for each $\nu \in [0,z]$, the value 
$$
\sum_{i' \in J_\nu} c_i Z_{A,D}(\psi_{i'},G').
$$
To prove this, define $x_\nu : = Z_{A,D}(\psi_{i'},\Gamma)$ for each $\nu \in [z]$ and an $i' \in J_{\nu}$. Equation \eqref{eq:2603091557} implies that we can compute
$$
\sum_{\nu = 0}^z x_\nu^s \left(\sum_{i' \in J_\nu} c_{i'} Z_{A,D}(\psi_{i'},G')\right). 
$$
One of the values $x_\nu$ might be zero. Assume therefore w.l.o.g. that $x_0 = 0$, then evaluating the above for $s= 1, \ldots, z$ yields a system of linear equations, which by Lemma~\ref{lem:interpolate} can be solved in polynomial time such that we can recover the values $\sum_{i' \in J_\nu} c_{i'} Z_{A,D}(\psi_{i'},G')$ for each $\nu \ge 1$. 
Using equation \eqref{eq:260320091543} we can thus also compute the value
\[
\left(\sum_{i \in I} c_i\cdot Z_{A,D}(\psi_i,G')\right) - \sum_{\nu = 1}^z \left(\sum_{i' \in J_\nu} c_{i'} Z_{A,D}(\psi_{i'},G')\right) = \sum_{i' \in J_0} c_{i'} Z_{A,D}(\psi_{i'},G'). 
\]
\end{proof}
The proof of the Pinning Lemma now follows easily from Lemmas~\ref{lem:twin_red} and \ref{lem:pinning_twin_free}.

\begin{proof}[of the Pinning Lemma \ref{lem:pinning}]
Fix the twin resolvent $A' = \twres A$ of $A$ and let $D' = D^{\twres A}$. 
It suffices to show
\begin{equation} \label{eq:intro_pin}
  \evalk(A',D') \Tequiv \eval(A',D').
\end{equation}
To see this note that by the Twin Reduction Lemma~\ref{lem:twin_red} we then have the chain of reductions
$$\evalk(A,D) \Tequiv \evalk(A',D') \Tequiv \eval(A',D') \Tequiv \eval(A,D).$$
The proof of equation \eqref{eq:intro_pin} follows from Lemma~\ref{lem:pinning_twin_free}.
\end{proof}

\section{The General Conditioning Lemma}
\label{sec:gc}

\subsection{Dealing with $\{0,1\}$ Matrices}

As a technical prerequisite, we need a part of the result of \cite{dyegre00} (Theorem 1.1 in there).
We call a block \emph{non-trivial} if it contains a non-zero entry.

\begin{lemma}[\#H-Coloring Lemma]\label{lem:0-1_hardness}
Let $A$ be a symmetric connected and bipartite $\{0,1\}$-matrix with underlying non-trivial block $B$. If $B$ contains a zero entry then the problem $\evalk(A)$ is $\#\PP$-hard.
\end{lemma}
\begin{proof}
We will start with the following claim which captures the main reduction. Call a matrix $A$ \emph{powerful} if it is a symmetric connected and bipartite $\{0,1\}$-matrix whose underlying block $B$ contains a zero entry.

\begin{claim}\label{cl:2912091747}
Let $A$ be a  powerful matrix with underlying $m\times n$ block $B$. If either $n > 2$ or $m > 2$ then there is a powerful matrix $A'$ with underlying $m'\times n'$ block $B'$ such that $2 \le m' \le m$ and $2 \le n' \le n$, at least one of these inequalities is strict and
$$
\evalk(A') \Tle \evalk(A).
$$
\end{claim}
Before we prove the claim, let us see how it helps in proving the lemma. Let $A$ be as in the statement of the lemma. Iterating Claim~\ref{cl:2912091747} for a finite number of steps we arrive at $\evalk(A') \Tle \evalk(A)$ such that the block $B'$ underlying $A'$ is a $2\times 2$ matrix of the form (up to permutation of rows/columns)
$$
\left(\begin{array}{c c}
          1 & 1 \\
          1 & 0
      \end{array}\right).
$$
For every empty pinning $\vpin$ and bipartite graph $G$ with $c$ connected components $Z_{A'}(\vpin,G)$ equals $2^c$ times the number of independent sets of $G$. This straightforwardly gives rise to a reduction from the problem of counting independent sets in bipartite graphs which is well-known to be $\#\PP$-hard (see \cite{probal83}).

\paragraph*{Proof of Claim~\ref{cl:2912091747}.}
As $B$ is a block with a zero entry, there are indices $i,j,k,l$ such that $B_{ik} = B_{il} = B_{jk} = 1$ and $B_{jl} = 0$. Fix these indices and let $I = \{\nu \mid B_{i\nu} = 1\}$ and $J = \{\mu \mid B_{\mu k} = 1\}$ be the sets of indices of $1$-entries in row $i$ and column $k$. Let $A^*$ be the connected bipartite matrix with underlying block $B_{IJ}$. We will show that
\begin{equation}\label{eq:2912091918}
\evalk(A^*) \Tle \evalk(A).
\end{equation}
To see this, let $G,\vpin$ be an instance of $\evalk(A^*)$. We shall consider
only the case here that $G =(U \cup W,E)$ is connected and bipartite with
bipartition $U,W$ and that $\vpin$ pins a vertex $a \in U$ to a row of
$B_{IJ}$. All other cases follow similarly.

Define $G'$ as the graph obtained from $G$ by adding two new vertices $u'$ and $w'$ and connect every vertex of $U$ by an edge to $w'$ and every vertex of $W$ by a edge to $u'$. Let $\vpin'$ be the pinning obtained from $\vpin$ by adding $u' \mapsto i$ and $w' \mapsto k$. We have $Z_{A^*}(\phi, G) = Z_{A}(\phi', G')$ which yields a reduction witnessing equation \eqref{eq:2912091918}.

Before we proceed, we need another

\begin{claim}\label{cl:2912091802}
Let $A^+$ be a  powerful matrix with underlying $m^+\times n^+$ block $B^+$. There is a twin-free powerful matrix $A''$ with underlying $m''\times n''$ block $B''$ such that $2 \le m'' \le m^+$ and $2 \le n'' \le n^+$ and
$$
\evalk(A'') \Tle \evalk(A^+).
$$
\end{claim}
\begin{clproof}
This is a straightforward combination of the Twin Reduction Lemma \ref{lem:twin_red} and Lemma \ref{lem:dg_omit_vertexweights}.
\end{clproof}

Combining equation \eqref{eq:2912091918} and Claim~\ref{cl:2912091802} we arrive at $\evalk(A'') \Tle \evalk(A)$ for a powerful twin-free matrix $A''$. The block $B''$ underlying $A''$ has some dimension $m''\times n''$ such that $2\le m'' \le m$ and $2 \le n'' \le n$. Further, up to permutation of rows/columns, this block has the following form
$$
B'' = \left( \begin{array}{c c c c}
1 & 1 & \ldots & 1\\
1 & 0 & \ldots & * \\
\vdots &  & \ddots & \\
1 & * & \ldots & 
\end{array}\right)
$$
To prove Claim~\ref{cl:2912091747} it suffices to devise a reduction witnessing
\begin{equation}\label{eq:3012091327}
\evalk(A') \Tle \evalk(A'') 
\end{equation}
for a powerful matrix $A'$ such that the block $B'$ underlying $A'$ has either fewer rows or columns than $B''$.

To devise such a reduction, assume that $B''$ has at least three rows (the case that $B''$ has at least three 
columns is symmetric). As $B''$ is twin-free, there must be entries $B''_{2a} \neq B''_{3a}$. Assume w.l.o.g. 
that $B''_{2a} = 0$ and $B''_{3a} = 1$. As $B''$ is twin-free we further have that $B''_{3b} = 0$ for some $b$.
Let $K = \{\kappa \mid B''_{3\kappa} = 1\}$ be the indices of non-zero entries of $B'' \row 3$.
Define $B' = B''_{*,K}$ and let $A'$ be the connected bipartite matrix with underlying block $B'$. 
It remains to devise the desired reduction.

Let $G,\vpin$ be an instance of $\evalk(A')$. As before, we shall consider
only the case that $G =(U \cup W,E)$ is connected and bipartite with
bipartition $U,W$ and that $\vpin$ pins a vertex $a \in U$ to a row of
$B'$. All other cases follow similarly.

Define $G'$ as the graph obtained from $G$ by adding one new vertex $u'$ and connect every vertex of $W$ by a edge to $u'$. Let $\vpin'$ be the pinning obtained from $\vpin$ by adding $u' \mapsto 3$. This yields a reduction witnessing equation \eqref{eq:3012091327} and hence finishes the proof of Claim~\ref{cl:2912091747}.
\end{proof}

\subsection{From General Matrices to Positive Matrices}
In the first step of the proof of the General Conditioning Lemma \ref{lem:new_gen_cond} we will see how to restrict attention to matrices with positive entries.
\begin{lemma}[The Lemma of the Positive Witness]\label{lem:red_to_conn_zero-free}
Let $A \in \Ring^{m \times m}$ be a symmetric non-negative matrix containing a block of row rank at least $2$. Then there is an  $\Ring$-matrix $A'$ satisfying condition \cond A, that is $A'$ is positive symmetric of row rank at least $2$, such that
$$
  \evalk(A') \Tle \evalk(A).
$$
\end{lemma}
The proof of this lemma relies on the elimination of zero entries. We do this in two steps, first (in the next Lemma) we pin to components. Afterwards, we will see how to eliminate zero entries within a component.

\begin{lemma}[Component Pinning]\label{lem:comp_pin}
Let $A \in \Ring^{m\times m}$ be a symmetric matrix.
Then for each component $C$ of $A$ we have
$$ \evalk(C) \Tle \evalk(A).$$
\end{lemma}
\begin{proof}
Let $G,\vpin$ be the input to $\evalk(C)$ for some component $C$ of some order $m' \times m'$. Let $G_1, \ldots, G_\ell$ be the components of $G$ and $\vpin_1,\ldots,\vpin_\ell$ the corresponding restrictions of $\vpin$. Since
$$
Z_C(\vpin, G) = \prod_{i=1}^\ell Z_C(\vpin_i, G_i)
$$
we may assume w.l.o.g that $G$ is connected which we shall do in the following.
Let $v$ be a vertex in $G$ which is not pinned by $\vpin$. Define, for each $i \in [m]$, $\vpin_i$ as the pinning obtained from $\vpin$ by adding $v \mapsto i$. We have
$$
Z_C(\vpin, G) = \sum_{i = 1}^m Z_C(\vpin_i,G).
$$
Thus it suffices to compute $Z_C(\vpin_i,G)$ which can be obtained straightforwardly as these values equal $Z_A(\vpin_i,G)$.
\end{proof}

%

\begin{lemma}[Zero-Free Block Lemma]\label{lem:zero_free_blocks}
Let $A \in \Ring^{m\times m}$ be a symmetric matrix.
Then either $\evalk(A)$ is $\#\PP$-hard or no block of $A$ contains zero entries.
\end{lemma}
\begin{proof} 
Let $A''$ be the matrix obtained from $A$ by replacing each non-zero entry by $1$. 
Let $C$ be a component of $A''$ whose underlying block $B$ contains a zero entry.
For every given graph $G$ and pinning $\vpin$ each configuration $\vcfg$ which contributes a non-zero weight to $Z_A(\vpin,G)$ contributes a weight $1$ to $Z_{A''}(\vpin,G)$. 
Thus by $\evalk(A) \Tequiv \cntk(A)$ 
(cf. Lemma~\ref{lem:cntk_eq_evalk}) and component pinning 
(cf. Lemma \ref{lem:comp_pin}) it is easy to see that
$$
\evalk(C) \Tle \evalk(A'') \Tle \evalk(A).
$$
If $C$ is bipartite, the result follows from the \#H-Coloring Lemma~\ref{lem:0-1_hardness}.

If $C$ is not bipartite, we need a bit of additional work. 
Note that in this case $C = B$. Let $A'$ be the connected bipartite $2r \times 2r$ matrix with 
underlying block $B$. If we can show that
$$
\evalk(A') \Tle \evalk(B)
$$
then the result follows from the \#H-Coloring Lemma~\ref{lem:0-1_hardness}.

Let $G, \vpin$ be an instance of $\evalk(A')$. For simplicity, assume that $G = (U \cup W,E)$ is connected and bipartite. If $\vpin$ is empty, then we have $Z_{A'}(G) = 2\cdot Z_{B}(G)$.
Otherwise, $Z_{A'}(\vpin,G) = 0$ unless $\vpin$ maps all elements of $\df(\vpin) \cap U$ to $[r]$ and all entries of $\df(\vpin) \cap W$ to $[r+1,2r]$ or vice versa. Since both cases are symmetric, 
we consider only the first one. Let $\vpin': \df(\vpin) \to [2r]$ be the pinning which agrees with $\vpin$ on $\df(\vpin) \cap U$ and for all $w \in \df(\vpin)\cap W$ satisfies $\vpin'(w) = \vpin(w) - r$. By inspection we have
\[
Z_{A'}(\vpin, G) = Z_B(\vpin',G'). 
\]
\end{proof}

\begin{proof}[of Lemma \ref{lem:red_to_conn_zero-free}]
Let $A \in \Ring^{m \times m}$ contain a block $B$ of rank at least $2$. By $2$-stretching (cf. Lemma~\ref{lem:basic_reductions}) we have $\evalk(A^2)\Tle \evalk(A)$ and Lemma~\ref{lem:a_square_prop} guarantees that $A^2$ contains a component $BB^T$ which has rank at least $2$. By component pinning (Lemma~\ref{lem:comp_pin}) we have $\evalk(BB^T) \Tle \evalk(A^2)$. If $BB^T$ contains no zero entries, we let $A' = BB^T$. 

Otherwise let $A' \in \Nat^{2 \times 2}$ be a matrix satisfying the conditions of the lemma. As $BB^T$ contains a zero entry, Lemma~\ref{lem:zero_free_blocks} implies that $\evalk(BB^T)$ is $\#\PP$-hard and thus $\evalk(A')\Tle  \evalk(BB^T)$ finishing the proof.
\end{proof}

\subsection{From Positive Matrices to X-matrices}\label{subsec:pos_to_X}
Let $A$ be a matrix which satisfies condition \cond A. We will now see, how to obtain a matrix $A'$ which additionally satisfies \cond B. That is, we will prove the following lemma.
\begin{lemma}[$X$-Lemma] \label{lem:X-Lemma}
Let $A \in \Ring^{m \times m}$ be a matrix satisfying condition \cond A.
Then there is an $\Ring$-matrix $A'$ satisfying conditions \cond A and \cond B such that
$$
\evalk(A') \Tle \evalk(A).
$$
\end{lemma}
To prepare the proof, we present some smaller lemmas which will also be useful in later sections.
\begin{lemma}[Prime Elimination Lemma] \label{lem:prime_elim}
Let $A \in \Int^{m\times m}$ $(A \in \Int[X]^{m\times m})$ and $p$ be a prime number (an irreducible polynomial). Let $A'$ be the matrix obtained from $A$ by replacing all entries divisible by $p$ with $0$. Then 
$$\evalk(A') \Tle \evalk(A).$$ 
\end{lemma}
\begin{proof}
By Lemma~\ref{lem:cntk_eq_evalk} it suffices to give a reduction witnessing $\evalk(A') \Tle \cntk(A)$, 
which we will construct in the following. Let $G = (V,E)$ be a given graph $\vpin$ a pinning. 
As $\wset_{A'}(G) = \wset_A(G) \setminus \mset{w \,\vert \, w \textrm{ divisible by } p}$, 
we have $\wset_{A'}(G) \subseteq \wset_A(G)$. Moreover $N_{A'}(G,\vpin,w) = N_A(G,\vpin,w)$ for all 
$w \in \wset_{A'}(G)$ which implies
$$ Z_{A'}(\vpin,G) = \sum_{w \in \wset_{A'}(G)} w\cdot N_A(G,\vpin,w).$$
The values $N_A(G,\vpin,w)$ can be obtained directly using the $\cntk(A)$ oracle.
\end{proof}
Let $p$ be a prime number (an irreducible polynomial, respectively) and $a \in \Int$ ($a \in \Int[X]$, resp.). Define 
$$
\idxsymbm{a\vert_p} = \left\lbrace \begin{array}{l l}
                             p^{\max\{k \ge 0 \mid p^k \text{ divides } a\}}&, \text{ if } a \neq 0 \\
                             0 &, \text{ otherwise.}
                        \end{array}\right.
$$ 
For a matrix $A$ the matrix \idxsymb{$A\vert_p$} is then defined by replacing each entry $A_{ij}$ with $A_{ij}\vert_p$.

\begin{lemma}[Prime Filter Lemma]\label{lem:prime_filter}
Let $A \in \Int^{m\times m}$ $(A \in \Int[X]^{m\times m})$ and $p$ be a prime number (an irreducible polynomial). Then $$\evalk(A\vert_p) \Tle \evalk(A).$$ 
\end{lemma}
\begin{proof}
By Lemma~\ref{lem:cntk_eq_evalk} it suffices to give a reduction witnessing $\evalk(A\vert_p) \Tle \cntk(A)$, which we will construct in the following. For a given graph $G = (V,E)$ and a pinning $\vpin$ we have
$$ Z_{A\vert_p}(\vpin,G) = \sum_{w \in \wset_{A}(G)} w\vert_p \cdot N_A(G,\vpin,w).$$
The values $N_A(G,\vpin,w)$ can be obtained directly using a $\cntk(A)$ oracle.
\end{proof}

\begin{lemma}[Renaming Lemma] \label{lem:rename}
Let $p \in \Int[X]\setminus \mset{-1,0,1}$ and $A \in \Int[X]^{m\times m}$ a $p$-matrix. Let $q \in \Int[X]$ and define $A' \in \Int[X]^{m\times m}$ by 
$$ A'_{ij} = \left\lbrace\begin{array}{l l}
                           q^l &, \textrm{ there is an } l\ge 0 \textrm{ s.t. } A_{ij} = p^l \\
                           0 &, \textrm{ otherwise }
                        \end{array}\right.$$
That is, $A'$ is the matrix obtained from $A$ by substituting powers of $p$ with the corresponding powers of $q$. Then $$\evalk(A') \Tle \evalk(A).$$
\end{lemma}
\begin{proof}
Consider $A'$ as a function $A'(Y)$ in some indeterminate $Y$. We have $A = A'(p)$. Let $\ell_{\textup{max}}$ be the maximum power of $Y$ occurring in $A'(Y)$. For every graph $G$ and pinning $\vpin$, the value $f(Y) := Z_{A'(Y)}(\vpin, G)$ is a polynomial in $Y$ of maximum degree $\vert E \vert \cdot \ell_{\textup{max}}$.
By $t$-thickening (cf. Lemma~\ref{lem:basic_reductions}), using an $\evalk(A)$ oracle, we can compute the values
$ f(p^t) $ for $t = 1, \ldots, \vert E \vert \cdot \ell_{\textup{max}}$. Thus, by interpolation (cf. Lemma~\ref{lem:interpolate}) we can compute the value $f(q)$ for a $q$ as given in the statement of the lemma. By $f(q) = Z_{A'(q)}(\vpin, G)$ this proves the claimed reducibility.
\end{proof}

\begin{lemma}[Prime Rank Lemma] \label{lem:prime_rank}
Let $A \in \Int^{m\times m}$ $(A \in \Int[X]^{m\times m})$ contain a block of row rank at least $2$. There is a prime number (an irreducible polynomial) $p$ such that $A\vert_p$ contains a block of row rank at least $2$.
\end{lemma}
\begin{proof}
Let $A \row i$ and $A\row {i'}$ be linearly independent rows from a block of $A$.
Assume, for contradiction, that for all primes (irreducible polynomials, resp.) $p$ every block in $A\vert_p$ has rank at most $1$.  We have, for all $j \in [m]$, 
$$
A_{ij} = \prod_{p} A_{ij}\vert_p \quad  \text{ and } \quad
A_{i'j} = \prod_{p} A_{i'j}\vert_p
$$
where the products are over all primes (irreducible polynomials, resp.) dividing an entry of $A$. By assumption, there are $\alpha_p, \beta_p \in \Int$ (in $\Int[X]$, resp.) such that $\alpha_p  \cdot A\row{i}\vert_p = \beta_p \cdot A\row{i'}\vert_p$ for all primes (irreducible polynomials). Therefore,
$$
A_{ij} \prod_{p} \alpha_p = \prod_{p} \alpha_p  A_{ij}\vert_p = \prod_{p} \beta_p  A_{i'j}\vert_p = A_{i'j} \cdot \prod_{p} \beta_p \text{ for all } j \in [m].
$$
And hence, $A \row{i'}$ and $A\row i$ are linearly dependent --- a contradiction.
\end{proof}

\paragraph*{Dealing with algebraic numbers.}\label{sec:algstrucpf}
The following lemma tells us that the structure of the numbers involved in computations of partition functions on matrices with algebraic entries is already captured by matrices with natural numbers as entries.

\begin{lemma}[The Arithmetical Structure Lemma]\label{lem:alg_to_w_alg_red}
Let $A \in \Ralg^{m \times m}$ be symmetric and non-negative. There is a matrix $A'$ whose entries are natural numbers such that
$$
\evalk(A') \Tequiv \evalk(A).
$$
If further $A$ contains a block of rank at least $2$ then this is also true for $A'$.
\end{lemma}
We have to introduce some terminology. Let $B = \{b_1, \ldots, b_n\} \subseteq \Ralg$ be a set of positive numbers. The set $B$ is called \sdef{multiplicatively independent}, if for all $\lambda_1, \ldots, \lambda_n \in \Int$ the following holds: if $b_1^{\lambda_1} \cdots b_n^{\lambda_n}$ is a root of unity then $\lambda_1 = \ldots = \lambda_n = 0$.
In all other cases we say that $B$ is \sdef{multiplicatively dependent}. 
We say that a set $S$ of positive numbers is \sdef{effectively representable} in terms of $B$, if for given $x \in S$ we can compute $\lambda_1, \ldots ,\lambda_n \in \Int$ such that
$x \cdot b_1^{\lambda_1} \cdots b_n^{\lambda_n} = 1$.
A set $B$ is an \sdef{effective representation system} for a set $S$, if $S$ is effectively representable in terms of $B$ and $B$ is multiplicatively independent.

We need a result from \cite{ric01} which we rephrase a bit for our purposes.

\begin{lemma}[Theorem 2 in \cite{ric01}]\label{lem:Richardson}
Let $a_1,\dots, a_n \in \Q(\theta)$ be positive real numbers given in standard representation, each of description length at most $s$. There is a matrix $A \in \Int^{n \times n}$ such that, for vectors $ \vec \lambda \in \Int^n$ we have
\begin{equation}\label{eq:0506091450}
 \prod_{i=1}^n a_i^{\lambda_i}  = 1 \text{ if, and only if, } A\cdot \vec \lambda = 0.
\end{equation}
The description length of $A$ is bounded by a computable function in $n$ and $s$.
\end{lemma}
This result straightforwardly extends to an algorithm solving the multiplicative independence problem for algebraic numbers.

\begin{cor}\label{cor:05061625}
Let $a_1,\dots, a_n \in \Q(\theta)$ be positive reals given in standard representation. There is an algorithm which decides if there is a non-zero vector $\vec \lambda = (\lambda_1,\ldots, \lambda_n) \in \Int^n$ such that
\begin{equation}\label{eq:0506091456}
\prod_{i=1}^n a_i^{\lambda_i} = 1. 
\end{equation}
Furthermore, if it exists, the algorithm computes such a vector $\vec \lambda$.
\end{cor}

\begin{lemma}\label{lem:04061632}
Let $S \subseteq \Ralg$ be a set of positive numbers. There is an effective representation system $B \subseteq \Ralg$ of positive numbers for $S$ which can be computed effectively from $S$. 
\end{lemma}
\begin{proof}
We shall start with the following 
\begin{claim}\label{cl:04061619}
If $S$ is multiplicatively dependent then there is a set $B'\subseteq \Ralg$ of non-negative numbers such that $|B'| < |S|$ and $S$ is effectively representable by $B'$.
\end{claim}
\begin{clproof}
Let $S =\{b_1,\ldots ,b_n\}$ then Corollary~\ref{cor:05061625} implies that we can compute a non-zero vector $\vec \lambda \in \Int^n$ such that 
$b_1^{\lambda_1}\cdots b_n^{\lambda_n} = 1$.
We can easily make sure that at least one of the $\lambda_i$ is larger than zero. Assume therefore w.l.o.g. that $\lambda_1 > 0$. Fix a set $B'=\{b'_2,\ldots,b'_{n}\}$ where each $b'_i$ is the positive real $\lambda_1$-th root of $b_i$, that is $(b'_i)^{\lambda_1} = b_i$. Then
$$
b_1^{\lambda_1} \cdot \left(\prod_{i = 2}^{n} (b'_i)^{\lambda_i} \right)^{\lambda_1} = 1 \quad \text{ and hence }\quad 
b_1 \cdot \prod_{i = 2}^{n} (b'_i)^{\lambda_i} = 1.
$$
All operations are computable and effective representation of $S$ by $B'$ follows.
\end{clproof}
Apply Claim~\ref{cl:04061619} recursively on $S$. Since the empty set is multiplicatively independent, after at most finitely many steps, we find an effective representation system $B$ for $S$.
\end{proof}

\begin{proof}[of Lemma \ref{lem:alg_to_w_alg_red}]
Let $S$ be the set of non-zero entries of $A$. By Lemma~\ref{lem:04061632} we can compute an effective representation system $B$ for $S$. However, with respect to our model of computation we need to be a bit careful, here: assume that $S \subseteq \Q(\theta)$ for some primitive element $\theta$. The application of Lemma~\ref{lem:04061632} does not allow us to stipulate that $B  \subseteq \Q(\theta)$. But in another step of pre-computation, we can compute another primitive element $\theta'$ for the elements of $B$ such that $B \subseteq \Q(\theta')$ (c.f \cite{coh93}). Then we may consider all computations as taking place in $\Q(\theta')$.

Assume that $B=\{b_1,\ldots, b_n\}$, then every non-zero entry of $A$ has a unique computable representation 
$$
A_{ij} = \prod_{\nu = 1}^n b_{\nu}^{\lambda_{ij\nu}}.
$$
Let $p_1,\ldots,p_\beta$ be $\beta = |B|$ distinct prime numbers and define $A'$ as the matrix obtained from $A$ by replacing in each non-zero entry $A_{ij}$ the powers of $b \in B$ by the corresponding powers of primes, that is,
$$
A'_{ij} =  \prod_{\nu = 1}^n p_\nu^{\lambda_{ij \nu}}.
$$
Recall the definition of $\wset_A(G)$ in equation \eqref{eq:define_wset}. For each $w \in \wset_A(G)$ we can, in polynomial time compute a representation $w = \prod_{i,j} A_{ij}^{m_{ij}}$ as powers of elements in $S$. 
The effective representation of $S$ in terms of $B$ extends to $\wset_A(G)$ being effectively representable by $B$. Moreover, as $S$ depends only on $A$, the representation of each $w \in \wset_A(G)$ is even polynomial time computable.
We have
$$
Z_{A}(\vpin, G) = \sum_{w \in \wset_{A}(G)} w \cdot N_{A}(G,\vpin ,w)
$$
In particular, for each $w \in \wset_A(G)$, we can compute unique $\lambda_{w,1},\ldots, \lambda_{w,n} \in \Int$ such that $w  \cdot  b_1^{\lambda_{w,1}} \cdots b_n^{\lambda_{w,n}}  = 1$.
Define functions $f$ and $g$ such that for every $w \in \wset_A(G)$ we have 
$$
f(w) =  \prod_{\nu = 1}^n p_\nu^{\lambda_{w,\nu}} \quad \text { and } \quad
g(w) =  \prod_{\nu = 1}^n b_\nu^{\lambda_{w,\nu}}.
$$
Thus we obtain
$$
Z_{A'}(\vpin, G) = \sum_{w \in \wset_{A}(G)} w \cdot \dfrac{f(w)}{g(w)} \cdot N_{A}(G,\vpin ,w).
$$
This yields a reduction for $\evalk(A') \Tle \evalk(A)$. The other direction follows by
$$
Z_{A}(\vpin, G) = \sum_{w' \in \wset_{A'}(G)} w' \cdot \dfrac{g(w)}{f(w)} \cdot N_{A'}(G,\vpin ,w').
$$
This also proves the reducibilities $\eval(A') \Tequiv \eval(A)$ since the input pinnings remain unaffected.
To finish the proof it remains to consider the case that $A$ contains a block of rank at least $2$. We have to show that $A'$ has this property as well. Let us now argue that $A'$ contains a block of rank at least $2$. Let $A \row i$ and $A\row {i'}$ be linearly independent rows from a block of $A$. Assume, for contradiction, that $A'\row{i'} = \alpha \cdot A'\row{i}$ for some $\alpha$.
Let $J$ be the set of indices $j$ such that $A'_{ij} \neq 0$. For each $j \in J$ we have
$$
\alpha = A'_{i'j} \cdot (A'_{ij})^{-1} = \prod_{\nu = 1}^n p_\nu^{\lambda_{i'j \nu} - \lambda_{ij \nu}}.
$$
Hence, for
$\beta = b_1^{\lambda_{i'j 1} - \lambda_{ij 1}} \cdots b_n^{\lambda_{i'j n} - \lambda_{ij n}}$ we obtain
$A\row{i'} = \beta \cdot A\row{i}$ --- a contradiction.
\end{proof}

\subsubsection{Proof of the $X$-Lemma \ref{lem:X-Lemma}.} 
Let $A \in \Ring^{m \times m}$ be a matrix satisfying condition \cond A.
Recall that $\Ring$ is one of $\Ralg,\Qu,\Int,\Int[X]$ or $\Qu[X]$.

If $\Ring = \Ralg$ then the entries of $A$ are all positive real values. Thus Lemma~\ref{lem:alg_to_w_alg_red} implies that there is a positive matrix $A' \in \Int^{m \times m}$ of rank at least $2$ such that $\evalk(A') \Tequiv \evalk(A)$.

If $A$ is a matrix of entries in $\Qu$ ($\Qu[X]$, respectively) then let $\lambda$ be the lowest common denominator of (coefficients of) entries in $A$. For a given graph $G = (V,E)$ and pinning $\vpin$ we have
\[ Z_{\lambda A}(\vpin,G) = \lambda^{\vert E \vert}Z_{A}(\vpin,G).\]
The matrix $\lambda \cdot A$ is a matrix with entries in $\Int$ ($\Int[X]$, respectively).

It remains to prove the Lemma for the case that $\Ring$ is either $\Int$ or $\Int[X]$. By the Prime Rank Lemma~\ref{lem:prime_rank} there is a prime (irreducible polynomial) $p$ such that 
$A\vert_p$ contains a block of rank at least $2$. Fix such a $p$ and define $A' = A\vert_p$.
By the Prime Filter Lemma~\ref{lem:prime_filter} we have
$$
\evalk(A') \Tle \evalk(A).
$$
Furthermore, by the Renaming Lemma~\ref{lem:rename} we may assume that $A'$ is an $X$-matrix. This completes the proof.

\subsection{From $X$-matrices to the General Conditioning Lemma}

\begin{lemma}\label{lem:all_in_1_cells}
Let $A \in \Int[X]^{m\times m}$ be symmetric and positive such that not all $1$-entries of $A$ are contained in $1$-cells. Then $\evalk(A)$ is $\#\PP$-hard.
\end{lemma}
\begin{proof}
 If not all $1$-entries are contained in $1$-cells, then there are $i,j,k,l$ such that $A_{ik} = A_{il} =A_{jk} = 1$ and $A_{jl} \neq 1$. Let $A'$ be obtained from $A$ by replacing each entry not equal to $1$ by $0$. We have $\evalk(A') \Tle \evalk(A)$ by the Prime Elimination Lemma~\ref{lem:prime_elim}.

By construction $A_{ik} = A_{il} =A_{jk} = 1$ and $A_{jl} = 0$. Thus $A'$ contains a block with zero entries and $\evalk(A')$ is $\#\PP$-hard by Lemma~\ref{lem:zero_free_blocks}.
\end{proof}

\begin{proof}[Of the General Conditioning Lemma \ref{lem:new_gen_cond}]
By Lemma~\ref{lem:red_to_conn_zero-free}, there is a matrix $C'$  which satisfies \cond A such that
$\evalk(C') \Tle \evalk(A).$
Then by the X-Lemma~\ref{lem:X-Lemma}, there is a matrix $C$ which satisfies \cond A and \cond B such that
$\evalk(C) \Tle \evalk(C').$
Lemma~\ref{lem:all_in_1_cells} implies that $\evalk(C)$ is $\#\PP$-hard if not all $1$-entries are contained in $1$-cells. In this case we let $A'$ be some $2 \times 2$ matrix satisfying conditions \cond{A}--\cond{C}.

Assume therefore that all $1$-entries of $C$ are contained in $1$-cells.
If $C$ contains exactly one $1$-cell then the proof follows by the symmetry of $C'$. We only have to make sure that we can permute the entries of $C$ such that condition \cond C is satisfied. This is guaranteed by the Permutability Principle~\ref{lem:principle_permute}.

Assume therefore that $C$ contains more than one $1$-cell. Define an $X$-matrix $C^* = C^2\vert_{X}$.
\begin{claim}
For every $1$-cell $C_{KL}$ of $C$ the principal submatrix $C^*_{KK}$ is a $1$-cell of $C^*$.
\end{claim}
\begin{clproof}
Note that $C^*_{ij} = 1$ only if there is an $\ell$ such that $C_{i\ell} = C_{j\ell} = 1$. This proves the claim.
\end{clproof}
Now $C^*$ has all $1$-entries in principal $1$-cells and is of rank at least $2$ by the fact that there are at least two $1$-cells in $C$.
A reduction witnessing $\evalk(C^*) \Tle \evalk(C)$ is given by applying $2$-stretching (cf. Lemma~\ref{lem:basic_reductions}) and the Prime Filter Lemma~\ref{lem:prime_filter} in this order.
\end{proof}

\section{The Two $1$-Cell Lemma}
\label{sec:21c}

\paragraph*{The \cond{T1C} -- Conditions for Matrices with two $1$-cells.} We define two additional conditions.
\begin{condition}{(T1C -- A)} $A$ has at least two $1$-cells.\end{condition}
\begin{condition}{(T1C -- B)} All diagonal entries of $A$ are $1$.\end{condition}
A \sdef{$1$-row} (\sdef{$1$-column}) in a matrix $A$ is a row (column) which contains a at least one $1$ entry. We call all other rows (columns) \sdefi{non-$1$-rows}{non-$1$-row} (\sdefi{non-$1$-columns}{non-$1$-column}).

\begin{lemma}[$1$-Row-Column Lemma]\label{lem:1-row}
Let $A \in \Int[X]^{m \times m}$ be a positive and symmetric $X$-matrix. 
Let $A'$ be obtained from $A$ by removing all non-$1$-rows and non-$1$-columns. Then
$$\evalk(A') \Tle \evalk(A).$$
\end{lemma}
\begin{proof}
For every $i \in [m]$ let $c_i$ denote the number of $1$-entries in row $A\row i$. And let $A''$ be the matrix defined by $A''_{ij} = c_i c_j A'_{ij}$ for all $i,j \in [m]$.

We start with a reduction witnessing $\evalk(A'') \Tle \evalk(A)$.
Let $G=(V,E)$, $\vpin$ be an instance of $\evalk(A'')$.
Let $\Delta$ be the maximum degree of $X$ in $A$, let $k = \Delta \cdot \vert E \vert + 1$ and define a graph $G' =(V',E')$ by
\begin{eqnarray*}
V' & = & \{ v,v^1 \mid v \in V \} \\
E' & = & E \cup \{e_v^1, \ldots, e_v^k \mid v \in V, \; e_v^i = vv^1, \; \forall i\in [k]\}.
\end{eqnarray*}
We have
\begin{eqnarray*}
Z_A(\vpin,G') &=& \sum_{\vpin \subseteq \vcfg : V' \rightarrow [m]} \prod_{uv \in E'} A_{\vcfg(u),\vcfg(v)} \\
              &=& \sum_{\vpin \subseteq \vcfg : V' \rightarrow [m]} \prod_{uv \in E} A_{\vcfg(u),\vcfg(v)}\prod_{v \in V} (A_{\vcfg(v),\vcfg(v^1)})^k \\
\end{eqnarray*}
That is, for every $\vcfg$ the degree (as a polynomial in $X$) of the weight expression
$$
\prod_{uv \in E} A_{\vcfg(u),\vcfg(v)}\prod_{v \in V} (A_{\vcfg(v),\vcfg(v^1)})^k
$$
is smaller than $k$ if, and only if, $A_{\vcfg(v)\vcfg(v^1)} = 1$ holds for every $v \in V$. That is, if a configuration $\vcfg$ maps a vertex of $V$ to a non-1-row, then the degree of the weight expression of $\vcfg$ will always be at least $k$.
On the other hand, for each configuration $\vcfg$ mapping $G$ to $A$ there are exactly $c_1^{\vert \vcfg^{-1}(1) \vert} \cdot \ldots c_m^{\vert \vcfg^{-1}(m) \vert}$ many configurations $\vcfg' : V' \to [m]$ of $G'$ extending $\vcfg$ in such a way that their weight is of degree smaller than $k$.

By the polynomial time equivalence of $\cntk(A)$ and $\evalk(A)$ it suffices to reduce $\evalk(A'') \Tle \cntk(A)$. We can do this by computing the values $w \cdot N_{A}(G',\vpin,w)$ for all weights $w$ of degree smaller than $k$.
Then, with $A' = A''\vert_{X}$, the remaining step $\evalk(A') \Tle \evalk(A'')$ is given by the Prime Filter Lemma~\ref{lem:prime_filter}.
\end{proof}

\begin{lemma}\label{lem:two-1-cell-1}
Let $A \in \Int[X]^{m \times m}$ be a matrix which satisfies \cond A and contains at least two $1$-cells. Then there is a matrix $A'$ satisfying conditions \cond{A}--\cond{C} and both \cond{T1C} conditions such that
$$
 \evalk(A') \Tle \evalk(A).
$$
\end{lemma}
\begin{proof}
Let $i,j,i',j' \in [m]$ be witnesses for the existence of two $1$-cells in $A$ in the sense that $A_{ij} = A_{i'j'} = 1$ but $A_{i'j} \neq 1$ and let $p$ be an irreducible polynomial which divides  $A_{i'j}$. Let $B$ be the matrix obtained from $A\vert_{p}$ by replacing all powers of $p$ with the corresponding powers of $X$. Then by Prime Filter Lemma~\ref{lem:prime_filter} and Renaming Lemma~\ref{lem:rename}, we have
$\evalk(B) \Tle \evalk(A)$.
Note that $B$ satisfies condition \cond B and it satisfies \cond A as $A$ does.
We may assume that all $1$-entries of $B$ are contained in $1$-cells, because otherwise $\evalk(B)$ would be $\#$P-hard by Lemma~\ref{lem:all_in_1_cells}.

As it could possibly be the case that not all $1$-cells of $B$ are on the diagonal, we form 
$B' = B^2\vert_X$. 
\begin{claim}
For every $1$-cell $B_{KL}$ of $B$ the principal submatrix $B'_{KK}$ is a $1$-cell of $B'$.
\end{claim}
\begin{clproof}
Note that $B'_{ij} = 1$ only if there is an $\ell$ such that $B_{i\ell} = B_{j\ell} = 1$. This proves the claim.
\end{clproof}
By this claim, since $B$ contains at least two $1$-cells, the matrix $B'$ does so as well and therefore it satisfies \cond A. Condition \cond B is satisfied by definition. We have $\evalk(B') \Tle \evalk(B)$ by application of $2$-stretching (cf. Lemma~\ref{lem:basic_reductions}) and the Prime Filter Lemma~\ref{lem:prime_filter} in this order. 
Applying the $1$-Row-Column Lemma~\ref{lem:1-row} on $B'$ then yields $\evalk(A')\Tle \evalk(B')$
for a matrix $A'$ which (up to permuting rows and columns) has the desired properties.
\end{proof}
The \sdef{cells} of a matrix $A$ satisfying conditions \cond{A}--\cond{C} are the submatrices $A_{I_iI_j}$ for $I_i,I_j$ as defined in condition \cond{C}. We call such a matrix $A$ a \sdef{cell matrix}, if in each of its cells $A_{I_iI_j}$ all entries are equal.

\begin{lemma}\label{lem:two-1-cell-2}
Let $A \in \Int[X]^{m \times m}$ satisfy conditions \cond{A} -- \cond{C} and both \cond{T1C} conditions. Then there is a cell matrix $C \in \Int[X]^{m \times m}$ which also satisfies \cond{A} -- \cond{C} and both \cond{T1C} conditions such that
$$
\evalk(C) \Tle \evalk(A).
$$
\end{lemma}
\begin{proof}
We define a sequence of matrices with $A_0 = A$ and for all $\nu \in \Nat$ we let $A_{\nu+1} = A^2_\nu\vert_X$. As $A$ satisfies conditions \cond A -- \cond C and both \cond{T1C} conditions, this is true for all matrices in the sequence. Further we have $\evalk(A_{\nu+1}) \Tle \evalk(A_{\nu})$ for all $\nu$, by applying $2$-stretching (cf. Lemma~\ref{lem:basic_reductions}) and the Prime Filter Lemma~\ref{lem:prime_filter} in this order.

By definition, $\deg ((A_{\nu+1})_{ij}) = \min \left\{\deg ((A_\nu)_{ik}) + \deg ((A_\nu)_{jk}) \mid k \in [m] \right\}$  for all $i,j \in [m]$. As condition \cond {T1C -- B} implies $\deg(A_{jj}) = 0$ for all $j \in [m]$, we obtain
$$
\deg((A_{\nu+1})_{ij}) \le\deg((A_{\nu})_{ij}) \text{ for all } \nu. 
$$
That is, the degrees of the entries of $A_{\nu}$ are non-increasing with $\nu$ and thus there is a $\mu$ such that $A_{\mu+1} = A_\mu$. To finish the proof we will show that $C = A_\mu$ is a cell matrix. Let $C_{IJ}$ be a cell of $C$ which is not a $1$-cell. Let $i\in I$ and $j \in J$ be such that $\deg(C_{ij})$ is minimal. By definition we have $C = C^2\vert_X$ and therefore, for all $j' \in J$,
$$
\deg(C_{ij'}) = \deg((C^2\vert_X)_{ij'}) = \min\{ \deg(C_{ik}) + \deg(C_{jk}) \mid k \in [m]\} \le \deg(C_{ij}).
$$
The inequality follows from the fact that $C_{jj} = 1$. Then by the minimality of $\deg(C_{ij})$ we have
$ \deg(C_{ij'}) = \deg(C_{ij})$. Analogous reasoning on $i' \in I$ yields $\deg(C_{ij}) = \deg(C_{i'j'})$ for all $i'\in I, j'\in J$. Thus $C$ is a cell matrix.
\end{proof}

\subsection{\#\PP-hardness}

\begin{lemma}\label{lem:hardness_base}
Let $\delta \in \Nat$ and $A'$ be a matrix of the form
$$
A' = \left(\begin{array}{c c c c c c}
     1 & \cdots & 1 & 2^{\delta} & \cdots & 2^{\delta} \\
     \vdots & & \vdots & \vdots & & \vdots \\
     1 & \cdots & 1 & 2^{\delta} & \cdots & 2^{\delta} \\
     2^{\delta} & \cdots & 2^{\delta} & 1 & \cdots & 1 \\
     \vdots & & \vdots & \vdots & & \vdots \\
     2^{\delta} & \cdots & 2^{\delta} & 1 & \cdots & 1 
    \end{array}\right)
$$
Then $\evalk(A)$ is $\#\PP$-hard.
\end{lemma}
\begin{proof}
Let $\twres{A'}$ be the twin resolvent of $A'$. By the Twin Reduction Lemma~\ref{lem:twin_red} we obtain
$\eval(\twres {A'},D) \Tequiv \eval(A')$ where, for $a,b,\delta \ge 1$, the matrices $\twres{A'}$ and $D$ satisfy
$$
\twres{A'} = \left(
     \begin{array}{c c}
      1&2^{\delta}\\
      2^{\delta}& 1
     \end{array}
 \right)
\text{ and } 
D  = \left(
     \begin{array}{c c}
      a & 0\\
      0 & b
     \end{array}
 \right)
$$ 
We have $\eval(\twres{A'}) \Tle \eval(\twres{A'},D)$ by Lemma \ref{lem:dg_omit_vertexweights}. Therefore it remains to show that $\eval(\twres{A'})$ is $\#\PP$-hard. To see this, let $G = (V,E)$ be a graph and for two sets of vertices $U,W \subseteq V$ let $e(U,W)$ denote the number of edges in $G$ between $U$ and $W$. We have
\begin{align*}
Z_{A}(G) &= \sum_{\vcfg : V \rightarrow [m]} \prod_{uv \in E} A_{\vcfg(u),\vcfg(v)}
         &= \sum_{\vcfg : V \rightarrow [m]} 2^{\delta \cdot e(\vcfg^{-1}(1),\vcfg^{-1}(2))}
\end{align*}
Let $c_i$ be the number of $\vcfg: V \rightarrow [m]$ with weight $2^{\delta i}$ then we have.
\begin{eqnarray*}
Z_{A}(G) &=& \sum_{i=0}^{|E|} c_i 2^{\delta i}
\end{eqnarray*}
By $\eval(A) \Tequiv \cnt(A)$ we can determine the coefficients $c_i$. Let $\nu$ be maximum such that $c_\nu \neq 0$. Then $c_\nu$ is the number of maximum cardinality cuts in $G$. Therefore, this yields a reduction from the problem \#MAXCUT of counting maximum cardinality cuts --- a problem well known to be $\#\PP$-hard (this follows, for example, from the work of Simon \cite{sim77}).
\end{proof}

\begin{lemma}\label{lem:two-1-cell-3}
Let $A \in \Int[X]^{m \times m}$ be a cell matrix satisfying conditions \cond{A}--\cond{C} and both \cond{T1C} conditions. Then $\evalk(A)$ is $\#\PP$-hard.
\end{lemma}
\begin{proof}
Let $\delta = \min \{\deg (A_{ij}) \mid A_{ij} \neq 1,\; i,j \in [m]\}$ and
    $\Delta = \max \{\deg (A_{ij}) \mid A_{ij} \neq 1,\; i,j \in [m]\}$.
Let $A_{IJ}$ be a cell of $A$ with entries $X^\delta$ and define $A' = A_{(I \cup J)(I \cup J)}$ which, by symmetry, the definition of the cells and \cond{C} has the form
$$
A' = A_{(I \cup J)(I \cup J)} = \left(
    \begin{array}{c c c c c c}
     1 & \cdots & 1 & X^{\delta} & \cdots & X^{\delta} \\
     \vdots & & \vdots & \vdots & & \vdots \\
     1 & \cdots & 1 & X^{\delta} & \cdots & X^{\delta} \\
     X^{\delta} & \cdots & X^{\delta} & 1 & \cdots & 1 \\
     \vdots & & \vdots & \vdots & & \vdots \\
     X^{\delta} & \cdots & X^{\delta} & 1 & \cdots & 1 
    \end{array}\right)
$$
We will show first that $\evalk(A') \Tle \evalk(A)$. For a graph $G= (V,E)$ and a pinning $\vpin$ of $\evalk(A')$, define a graph $G' = (V',E')$ as follows. Let $k = \vert E \vert \cdot \Delta + 1$, we obtain $G'$ from $G$ by adding two apices which are connected to each vertex of $V$ by edges with multiplicity $k$, that is,
\begin{eqnarray*}
 V' & = & V \dot\cup \{x,y\} \\
 E' & = & E \cup \{(xv)^k , (yv)^k \mid v \in V \}.
\end{eqnarray*}
Let furthermore $\vpin'$ be the extension of $\vpin$ to $G'$ by adding $x \mapsto i_0$, $y \mapsto j_0$ for some $i_0 \in I$ and $j_0 \in J$.
Consider a configuration $\vcfg: V' \rightarrow [m]$ with $\vpin' \subseteq \vcfg$. For some appropriate $K(\vcfg)$ we have
$$
\prod_{uv \in E'} A_{\vcfg(u),\vcfg(v)} = X^{K(\vcfg)} \prod_{uv \in E} A_{\vcfg(u),\vcfg(v)}.
$$
Further, if $\vcfg(V) \subseteq I \cup J$ then $K(\vcfg) = k\vert V \vert \delta$ 
and otherwise, by the definition of $\delta$ and the cell structure of $A$, we have $K(\vcfg) \ge k\vert V \vert \delta + k\delta$.
In particular, as $\deg\left( \prod_{uv \in E} A_{\vcfg(u),\vcfg(v)} \right) \le \Delta\cdot \vert E \vert < k$ we have
$$
   \text{ The degree of $\vcfg$ is strictly less than }  k\vert V \vert \delta + k \text{ iff } \vcfg(V) \subseteq I \cup J.
$$
We can thus compute $Z_{A'}(\vpin,G)$ using a $\cntk(A)$ oracle, by determining the values $w \cdot N_{A}(G',\vpin',w)$ for all $w$ of degree at most $k|V|\delta + k - 1$.
This yields a reduction $\evalk(A') \Tle \evalk(A)$ by the polynomial time equivalence of $\cntk(A)$ and $\evalk(A)$.

Let $A''$ be the matrix obtained from $A'$ by substituting $X$ with $2$. Trivially $\eval(A'') \Tle \evalk(A')$ and $\eval(A'')$ is $\#\PP$-hard by Lemma \ref{lem:hardness_base}.
\end{proof}

\begin{proof}[of the Two-$1$-Cell Lemma]
Let $A \in \Int[X]^{m \times m}$ be a positive symmetric matrix containing at least two $1$-cells. That is, $A$ satisfies condition \cond A. Application of Lemmas~\ref{lem:two-1-cell-1}, \ref{lem:two-1-cell-2} and \ref{lem:two-1-cell-3} in this order yields the result.
\end{proof}

\section{The Single $1$-Cell Lemma}
\label{sec:11c}

\begin{lemma}[Symmetrized $1$-Row Filter Lemma]\label{lem:symm_1-row_filter}
Let $A \in \Int[X]^{m \times m}$ be a matrix satisfying conditions \cond{A}--\cond{C} and containing exactly one $1$-cell. Let $C$ be obtained from $A$ by removing all non-$1$-rows. Then
$$
\evalk(CC^T) \Tle \evalk(A).
$$
\end{lemma}
\begin{proof}
As $A$ satisfies conditions \cond A -- \cond C and contains a single $1$-cell we see that there is an $r < m$ such that with $I = [r]$ the principal submatrix $A_{II}$ forms this single $1$-cell. 
Recall that $\evalk(A^2) \Tle \evalk(A)$ by $2$-stretching (cf. Lemma~\ref{lem:basic_reductions}), and we have
$$
(A^2)_{ij} = \sum_{k=1}^m A_{ik}A_{jk} = \sum_{k=1}^r A_{ik}A_{jk} + \sum_{k=r+1}^m A_{ik}A_{jk}.
$$ 
Therefore, for $i,j \in [r]$ we have $(A^2)_{ij} = r + \sum_{k=r+1}^m A_{ik}A_{jk}$, i.e. a polynomial with a constant term. 
On the other hand, if either $i\notin I$ or $j\notin I$ then $A^2_{ij}$ is divisible by $X$.
Further we have $C = A_{[r][m]}$ and thus $CC^T = (A^2)_{[r][r]}$. Hence $CC^T$ is the submatrix of $A^2$ which consists of exactly those entries not divisible by $X$.
Recall that for any graph $G$ and pinning $\vpin$, we have
$$
Z_{A^2}(\vpin, G) = \sum_{w \in \m{W}_{A^2}(G)} w \cdot N_{A^2}(G,\vpin, w)
$$
and by the above
$$
Z_{CC^T}(\vpin, G) = \sum_{\substack{w \in \m{W}_{A^2}(G) \\ X\textup{ does not divide } w}} w \cdot N_{A^2}(G,\vpin, w)
$$
Filtering the weights $w$ appropriately by means of the $\evalk(A) \Tequiv \cntk(A)$ correspondence, we can devise a reduction witnessing $\evalk(CC^T) \Tle \evalk(A)$.
\end{proof}

\paragraph*{Single $1$-Cell Conditions \cond{S1C}.} We define some further conditions for matrices $A \in \Int[X]^{m \times m}$ satisfying conditions \cond{A}--\cond{C}. Define $\delta_{ij} = \deg (A_{ij})$ for all $i,j \in [m]$ and let $r := \min\{ i\in [m] \mid A_{i1} > 1\}$. That is, $r$ is the smallest index such that $\delta_{1r} = \delta_{r1}$ is greater than zero and it exists because $A$ has rank at least $2$.
\begin{condition}{(S1C--A)} $A$ has exactly one $1$-cell.\end{condition}
\begin{condition}{(S1C--B)} $A\row 1 = \ldots = A\row{r-1}$, i.e. the first $r-1$ rows of $A$ are identical.\end{condition}
\begin{condition}{(S1C--C)} For $i \in [r,m]$ we have $ \delta_{1k} \le \delta_{ik}$ for all $k \in [m]$. \end{condition}

\begin{lemma} \label{lem:57}
 Let $A \in \Int[X]^{m \times m}$ be a matrix satisfying conditions \cond{A}--\cond{C} and \cond{S1C--A}. Then there is a matrix $A'$ satisfying \cond{A}--\cond{C} and all \cond{S1C} conditions such that
$$
\evalk(A') \Tle \evalk (A).
$$
\end{lemma}
\begin{proof}
The matrix $A$ already satisfies \cond{A}--\cond{C} and \cond{S1C--A}.
Observe first that if $A$ satisfies \cond{S1C--B} as well, a matrix $A'$ whose existence we want to prove, can be defined as follows. For $t \in \Nat$, let $A' = A'(t)$ be given by
$$
A'_{ij} = A_{ij} (A_{ii}A_{jj})^t \text{ for all } i,j \in [m].
$$
Clearly, for all $t \in \Nat$ the matrix $A'(t)$ satisfies conditions \cond{A}--\cond C, \cond {S1C--A} and \cond {S1C--B}. 

We claim that there is a $t \in \Nat$ such that $A'(t)$ also satisfies condition \cond{S1C--C}. To see this, note first that the degrees of $A$ satisfy $\min \{\delta_{rr}, \ldots, \delta_{mm}\} > 0$ as all $1$-entries of $A$ are contained in the single $1$-cell $A_{[r-1][r-1]}$.
Therefore, there is a $t$ such that $\delta_{1j} \le t \cdot \min \{\delta_{rr}, \ldots, \delta_{mm}\}$ for all $j \in [m]$. Fix such a $t$ and note that
$$
\deg(A'_{ij}) = \delta_{ij} + t\cdot \delta_{ii} + t \cdot \delta_{jj}.
$$
Thus, with $\delta_{11} = 0$ we see that for all $j \in [m]$ and $i \in [r,m]$
$$
\deg(A'_{1j}) = \delta_{1j} + t \cdot \delta_{jj} \le t\cdot \delta_{ii} + t \cdot \delta_{jj} \le \deg(A'_{ij}).
$$
This proves that $A'$ satisfies condition \cond{S1C--C}. Reducibility is given as follows
\begin{claim} For all $t \in \Nat$ we have $\evalk(A'(t)) \Tle \evalk(A).$
\end{claim}
\begin{clproof}
Let $G =(V,E)$, $\vpin$ be an instance of $\evalk(A')$, Let $G' =(V,E')$ be the graph obtained from $G$ by adding $t$ many self-loops to each vertex.
Then $Z_{A'}(\vpin, G) = Z_A(\vpin, G')$, witnessing the claimed reducibility.
\end{clproof}
It remains to show how to obtain condition \cond{S1C--B}. We give the proof by induction on $m$. For $m=2$ this is trivial. If $m\ge r > 2$ assume that $A$ does not satisfy \cond{S1C--B}. Define $C = A_{[r-1][m]}$, i.e. the matrix consisting of the first $r-1$ rows of $A$. By the Symmetrized $1$-Row Filter Lemma~\ref{lem:symm_1-row_filter} we have $\evalk(CC^T) \Tle \evalk(A)$. Further, $C$ has rank at least $2$ as $C_{i1} = 1$ for all $i \in [r-1]$ but the rows of $C$ are not identical and therefore $CC^T$ has rank at least $2$, as well (see Lemma~\ref{lem:a_square_prop}).

Application of the General Conditioning Lemma~\ref{lem:new_gen_cond} yields $\evalk(C') \Tle \evalk(CC^T)$ for a $k \times k$ matrix satisfying \cond{A}--\cond{C} such that $k \le r-1$. If $C'$ has at least two $1$-cells we apply the Two $1$-Cell Lemma~\ref{lem:two-1-cell} to obtain $\#\PP$-hardness of $\evalk(C')$. Therefore we can chose some $A'$ satisfying the conditions of the Lemma such that $\evalk(A') \Tle \evalk(C')$.
If $C'$ has only one $1$-cell then the proof follows by the induction hypothesis, as $C'$ has order $k \le r-1$.
\end{proof}

\bigskip

\paragraph{Definition of $A^{[k]}$ and $C^{[k]}$.} For the remainder of this section we fix $A\in \Int[X]^{m \times m}$ satisfying conditions \cond{A}--\cond{C} and all \cond{S1C} conditions. Furthermore, $\delta_{ij}$ for all $i,j \in [m]$ and $r$ are defined for $A$ as in the Single $1$-Cell Conditions.
For $k \in \Nat$ define the matrix $A^{[k]}$ by 
\begin{equation}\label{eq:def_Ak}
A^{[k]}_{ij} = A_{ij}(A_{1j})^{k-1} \text{ for all  } i,j \in [m]. 
\end{equation}
Let further, $C^{[k]}$ be defined by
\begin{equation}\label{eq:define_Ck}
C^{[k]} = A^{[k]}(A^{[k]})^T. 
\end{equation}

\begin{lemma}\label{lem:ck_reduction}
Let $A$ and $C^{[k]}$ be defined as above. Then
$$ \evalk(C^{[k]}) \Tle \evalk(A).$$
\end{lemma}
\begin{proof}
Let $G = (V,E)$ and a pinning $\vpin$ be an instance of $\evalk(C^{[k]})$.
Define a graph $G' = (V',E')$ as follows:
\begin{eqnarray*}
V' & := & V \;\dot \cup \; \{z\} \;\dot \cup \;\{v_e \mid e \in E\} \\
E' & := & \{uv_e, v_ev \mid e = uv \in E \} \;\dot \cup \;\{(v_ez)^{2k-2} \mid e \in E \}.
\end{eqnarray*}
The edges $v_ez$ have multiplicity $2k-2$. Let $\vpin'$ be the extension of $\vpin$ defined by $\vpin \cup \{ z \mapsto 1 \}$.
Let $\vcfg \supseteq \vpin'$ be a configuration on $G'$ and $A$. Then its weight 
equals 
\begin{eqnarray*}
\prod_{uv \in E'} A_{\vcfg(u),\vcfg(v)} &=& \prod_{uv \in E} \sum_{\nu = 1}^m A_{\vcfg(u),\nu} A_{\vcfg(v),\nu} (A_{1,\nu})^{2k-2} \\
&=& \prod_{uv \in E} C^{[k]}_{\vcfg(u),\vcfg(v)}.
\end{eqnarray*}
The last equality follows directly from the definition of $C^{[k]}$. This yields $\evalk(C^{[k]}) \Tle  \evalk(A)$ as required.
\end{proof}
\newcommand{\mult}{\textup{ \textsf{mult}}}
\newcommand{\lrk}{\lambda^{1/k}}
\newcommand{\Ck}{C^{[k]}}
\newcommand{\Coo}{C^{[1]}_{11}}
We need some further notation to deal with polynomials. For $f \in \Qu[X]$ and $\lambda \in \C$ let $\mult(\lambda,f)$ denote the \emph{multiplicity} of $\lambda$ in $f$ if $\lambda$ is a root of $f$, and $\mult(\lambda,f) = 0$, otherwise.
The \emph{$k$-th root} of $\lambda \in \C$ is the $k$-element set $\lambda^{1/k} = \{\mu \in \C \mid \mu^k = \lambda \}$. Slightly abusing notation, $\lambda^{1/k}$ will denote any element from this set.
For every root $\lambda$ of $C^{[1]}_{11}$, every $r \le j \le m$ and $k \in \Nat$ we define
\begin{equation}\label{eq:def_m_lam_j}
m(\lambda,j) = \min\{\mult(\lrk,\Ck_{1j}) \mid k \ge 1\}.
\end{equation}
The following Lemma is the technical core of this section.

\begin{lemma}[Single-$1$-Cell Technical Core]\label{lem:technical_core}
Let $C^{[k]}$ be defined as above. The following is true for all $j \in [r,m]$.
\begin{itemize}
 \item[(1)] For any root $\lambda$ of $\Coo$ and all $k \in \Nat$ we have 
   \begin{itemize}
    \item[(1a)] $\mult(\lrk,\Ck_{11}) = \mult(\lambda, \Coo)  \ge m(\lambda,j)$.
    \item[(1b)] $\mult(\lrk, \Ck_{jj}) \ge m(\lambda,j)$.
   \end{itemize}
 \item[(2)] If row $A \row 1$ and $A \row j$ are linearly dependent then for any root $\lambda$ of $\Coo$ and all $k \in \Nat$ we have 
          $$\mult(\lrk,\Ck_{11}) = \mult(\lrk,\Ck_{1j}) = \mult(\lrk, \Ck_{jj}).$$
 \item[(3)] If row $A \row 1$ and $A \row j$ are linearly independent there is a root $\lambda$ of $\Coo$ such that $$\mult(\lambda, \Coo) > m (\lambda,j).$$
\end{itemize} 
\end{lemma}
The proof of this lemma is technically involved and quite long. We therefore show first, how to prove the Single-$1$-Cell Lemma, provided that the above holds.

The following basic facts\footnote{cf., for example, Chapter IV in \cite{lan02}} about polynomials will be used frequently in the following.

\begin{lemma}\label{lem:59}
Let $f \in \Qu[X]$ and $\lambda \in \C$.
\begin{itemize}
 \item[(1)] There exists a unique (up to a scalar factor) irreducible polynomial $p_{\lambda} \in \Qu[X]$ such that $\lambda$ is a root of $p_{\lambda}$. If $f(\lambda) = 0$ then $p_{\lambda}$ divides $f$.
 \item[(2)] If $\mult(\lambda,f)=s$ then $f = p^s_{\lambda} \bar{f}$ for some $\bar{f} \in \Qu[X]$ which satisfies $\bar{f} (\lambda) \neq 0$.
 \item[(3)] If $f(\lambda) = 0$ then $\lrk$ is a root of $f(X^k)$ and $\mult(\lrk,f(X^k)) = \mult(\lambda, f(X))$.
\end{itemize} 
\end{lemma}

\begin{proof}[of the Single $1$-Cell Lemma \ref{lem:sing_1_cell}]
Let $A \in \Int[X]^{m \times m}$ satisfy conditions \cond{A}--\cond{C} and \cond{S1C--A}. By Lemma~\ref{lem:57} we may assume w.l.o.g. that $A$ indeed satisfies all \cond{S1C} conditions.
By the Two--$1$--Cell Lemma~\ref{lem:two-1-cell}, it suffices to prove the existence of a positive symmetric matrix $C$ which contains at least two $1$-cells such that 
$$
\evalk(C) \Tle \evalk(A).
$$
Recall the definition of $\Ck$ in equation \eqref{eq:define_Ck} and the definition of $r$ in the \cond{S1C} conditions. We prove a technical tool.

\begin{claim}\label{cl:lem_61}
There is a root $\lambda$ of $\Coo$, an index $j \in [r,m]$ and $k \in \Nat$ such that
\begin{itemize}
 \item[(1)] $\mult(\lrk, \Ck_{1j}) < \mult(\lrk, \Ck_{11})$. 
 \item[(2)] For all $i \in [r,m]$ we have $\mult(\lrk, \Ck_{1j}) \le \mult(\lrk, \Ck_{ii}).$
\end{itemize}
\end{claim}
\begin{clproof}
Choose $\lambda$ and $j \in [r,m]$ such that $m(\lambda,j)$ is minimal for all choices of $\lambda$ and $j$ which satisfy
\begin{equation}
 A\row 1 \text{ and } A\row j \text{ are linearly independent and } \mult(\lambda,\Coo) > m(\lambda,j).
\end{equation}
Note that these $\lambda$ and $j$ exist by Lemma~\ref{lem:technical_core}(3).
Choose $k \in \Nat$ such that, by equation \eqref{eq:def_m_lam_j}, $m(\lambda,j)= \mult(\lrk,\Ck_{1j})$.
This proves (1) by Lemma~\ref{lem:technical_core}(1a).

To prove (2) fix $i \in [r,m]$. If $A\row 1$ and $A \row i$ are linearly independent our choice of $\lambda$ and $j$ implies that $m(\lambda,i) \ge m(\lambda,j)$ and therefore
$$
 \mult(\lrk,\Ck_{ii}) \ge m(\lambda,i) \ge m(\lambda,j) = \mult(\lrk,\Ck_{1j})
$$
where the first inequality holds by Lemma~\ref{lem:technical_core}(1b).

If otherwise $A \row 1$ and $A\row i$ are linearly dependent then 
$$
\mult(\lrk,\Ck_{ii}) = \mult(\lrk,\Ck_{11}) \ge m(\lambda,j) = \mult(\lrk,\Ck_{1j}).
$$
The first equality holds by Lemma~\ref{lem:technical_core}(2) and the inequality is true by Lemma~\ref{lem:technical_core}(1a).
\end{clproof}
Choose $j,\lambda,k$ as in Claim~\ref{cl:lem_61}. Define $t = \mult(\lrk,\Ck_{1j})$.
Let $p_{\lambda}$ be an irreducible polynomial such that $p_{\lambda}(\lambda)=0$. Let $s : = \min\{ \mult(\lrk,\Ck_{ab}) \mid a,b \in [m] \}$ and define the positive symmetric matrix
$$ 
C := \dfrac{1}{p^s_{\lambda}}  \Ck\vert_{p_{\lambda}}.
$$
Note that $s \le t$ by definition. We consider two cases. If $s = t$ then $C_{1j} = C_{j1} = 1$ but by Claim~\ref{cl:lem_61}(1) we have $C_{11} \neq 1$. Therefore, $C$ contains at least two $1$-cells.

If otherwise $s < t$, fix witnesses $a,b \in [m]$ with $\mult(\lrk,\Ck_{ab}) = s$. 
\begin{claim}
Either $a \ge r$ or $b \ge r$. 
\end{claim}
\begin{clproof}
Recall that by condition \cond{S1C--B} the first $r-1$ rows of $A$ --- and hence those of $A^{[k]}$ --- are identical. By the definition $C^{[k]} = A^{[k]}(A^{[k]})^T$ (cf. equation \eqref{eq:define_Ck}) we have
$$
\Ck_{ij} = \sum_{\nu=1}^m A^{[k]}_{i\nu}A^{[k]}_{j\nu} = \sum_{\nu=1}^m A_{i\nu}A_{j\nu}(A_{1\nu})^{2k-2}
$$
In particular, for all $a',b' \in [r-1]$ we have $\Ck_{a'b'}  = \Ck_{11}$ and hence $\mult(\lrk, \Ck_{a'b'}) = \mult(\lrk, \Ck_{11}) > \mult(\lrk, \Ck_{ab})$ which proves the claim.
\end{clproof}
Combining this claim with Claim~\ref{cl:lem_61}(2) and the fact that $0<t=\mult(\lrk,\Ck_{1j})$, we have either $C_{aa} \neq 1$ or $C_{bb} \neq 1$. As $C_{ab} = C_{ba} = 1$ by the definition of $a,b$ we see that $C$ contains at least two $1$-cells.
We have $\evalk(C^{[k]}) \Tle \evalk(A)$ by Lemma~\ref{lem:ck_reduction}. Further
$\evalk(C^{[k]}\vert_{p_{\lambda}}) \Tle \evalk(C^{[k]})$ by the Prime Filter Lemma~\ref{lem:prime_filter}.
Then $\evalk(C) \Tle \evalk(A)$ follows from the fact that
\[
Z_{C^{[k]}\vert_{p_{\lambda}}}(\vpin, G) = p^{s\cdot |E|}_{\lambda}Z_{C}(\vpin, G) \text{ for all graphs } G=(V,E) \text{ and pinnings } \vpin.
\]
\end{proof}

\subsection{Proof of the Single-$1$-Cell Technical Core Lemma \ref{lem:technical_core}}

\setcounter{claim}{0}
\paragraph*{Proof of the Single-$1$-Cell Technical Core Lemma \ref{lem:technical_core}}
Fix $j\in [r,m]$ and define the value $b = \min \{\delta_{ji} - \delta_{1i} \mid i \in [m]\}$. By condition \cond{S1C--C} we have $b \ge 0$. Further, if rows $A \row 1$ and $A \row j$ are linearly dependent, then $A\row j  = X^b A\row 1$.

For simplicity of notation, define $a_i := \delta_{1i}$, $b_i := \delta_{ji} - b$ and $c_{i} : = b_{i} - a_{i} = \delta_{ji} - \delta_{1i} - b $ for all $i \in [m]$. Observe that all $c_i$ are non-negative and moreover $c_i = 0$ for all $i\in [m]$ iff rows $A \row 1$ and $A \row j$ are linearly dependent. If these rows are linearly independent then not all $c_i$ are equal.
We have
$$
\begin{array}{c c r l l l l}
A\row 1 & = & (& X^{a_1} & \ldots & X^{a_m}  & ) \\
A\row j & = & ( & X^{b + b_1} & \ldots & X^{b + b_m} &).
\end{array}.
$$
By the definition of $A^{[k]}$ in equation \eqref{eq:def_Ak}, we have $A^{[k]}_{i\nu} = A_{i\nu}(A_{1\nu})^{k-1} = A_{i\nu}X^{(k-1) a_{\nu}}$
and by $b+ b_\nu + (k-1)a_\nu = b+ c_\nu + ka_\nu$ we have
$$
\begin{array}{c c r l l l l}
A^{[k]} \row 1 & = & (& X^{ka_1} & \ldots & X^{ka_m}  & ) \\
A^{[k]} \row j & = & ( & X^{b + c_1 + ka_1} & \ldots & X^{b + c_m + ka_m} &).
\end{array}
$$
Therefore, by $C^{[k]} = A^{[k]} (A^{[k]})^T$,
\begin{equation}\label{eq:c_shape}
\begin{array}{c c l l l l l l}
\Ck_{11} & = &  \phantom{X^{2b}(}X^{2ka_1} & + & \ldots & + &X^{2ka_m}\\
\Ck_{1j} & = & X^{b\phantom{2}}(X^{c_1 + 2ka_1} & + & \ldots & + &X^{c_m + 2ka_m})\\
\Ck_{jj} & = & X^{2b} (X^{2c_1 + 2ka_1} & + & \ldots & + &X^{2c_m + 2ka_m}).
\end{array} 
\end{equation}
Let $\lambda$ be a root of $\Coo$ and recall that $a_1 = \ldots = a_{r-1} = 0$ as $r$ was defined to be minimal such that $\delta_{1r} = a_r > 0$. Therefore $\lambda \neq 0$ and by equation \eqref{eq:c_shape} and Lemma~\ref{lem:59}(3) we see that $\lrk$ is a root of $\Ck_{11}$ with 
\begin{equation}\label{eq:2003091712}
\mult(\lrk, \Ck_{11}) = \mult(\lambda, \Coo). 
\end{equation}
Assume first that $A\row 1$ and $A\row j$ are linearly dependent. Then $c_1 = \ldots = c_m = 0$ which, by \eqref{eq:c_shape} implies 
$ \Ck_{1j} = X^b \cdot \Ck_{11}$ and $\Ck_{jj} = X^{2b} \cdot \Ck_{11}$. As $\lambda \neq 0$ we have 
$$
\mult(\lrk,\Ck_{11}) = \mult(\lrk,\Ck_{1j}) = \mult(\lrk, \Ck_{jj}).
$$
This proves statement (2) of Lemma~\ref{lem:technical_core} and the case of Lemma~\ref{lem:technical_core}(1) where $A \row 1$ and $A \row j$ are linearly dependent.

It remains to prove statement (1) for linearly independent $A \row 1$ and $A \row j$ and statement (3). Assume in the following that $A \row 1$ and $A \row j$ are linearly independent. If furthermore there is a $k$ such that $\lrk$ is not a root of $\Ck_{1j}$ then $m(\lambda,j) = 0$ and the result follows.
Assume therefore that $\lrk$ is a root of $\Ck_{1j}$ for all $k \in \Nat$. Then, for all $k \ge 1$ we have
$$
0 = \Ck_{1j}(\lrk) =  \lambda^{b/k}(\lambda^{2a_1 + c_1/k}  +  \ldots  + \lambda^{2a_m + c_m/k})
$$
Defining $f_{\lambda}(z) = \lambda^{2a_1 + c_1 z}  +  \ldots  + \lambda^{2a_m + c_m z}$ we get
\begin{equation}\label{eq:2003091748}
0 = \Ck_{1j}(\lrk) = \lambda^{b/k} f_\lambda(1/k) \text{ for all } k\ge 1. 
\end{equation}
Let $f^{(l)}_{\lambda}$ denote the $l$-th derivative of $f_{\lambda}$ and let $\alpha \in \C$ be such that $\lambda = e^\alpha$. Then we have 
$f_{\lambda}(z) = e^{\alpha(2a_1 + c_1 z)}  +  \ldots  + e^{\alpha(2a_m + c_m z)}$
and thus 
\begin{equation}\label{eq:derivative_f}
f^{(l)}_{\lambda}(0) = (\alpha c_1)^l e^{2 \alpha a_1}  +  \ldots  + (\alpha c_m)^l e^{2 \alpha a_m}
                     = (\alpha c_1)^l \lambda^{2a_1}  +  \ldots  + (\alpha c_m)^l \lambda^{2 a_m}.
\end{equation}
We will prove the following property of the derivatives of $f_\lambda$.
\begin{equation}\label{eq:flambda_deriv_zero}
\textup{For all $l \ge 1$ we have $f^{(l)}_{\lambda}(0) = 0$.}
\end{equation}
The proof relies on the following claim.
\begin{claim} \label{cl:2203091500}
Suppose that $g(z) = u(z) + iv(z)$ is a function that is analytic in the real segment $[0,1]$ and $\{r_n\}_{n \in \Nat}$, $\{s_n\}_{n \in \Nat}$ are sequences from the real segment $[0,1]$ such that $\lim_{n \rightarrow \infty} r_n = \lim_{n \rightarrow \infty} s_n = 0$ and $u(r_n) = v(s_n) = 0$ for all $n \in \Nat$.
Then
\begin{itemize}
 \item[(a)] $g(0) =0$.
 \item[(b)] There are sequences $\{r'_n\}_{n \in \Nat}$ and $\{s'_n\}_{n \in \Nat}$ in the real segment $[0,1]$ such that
$$
 \lim_{n \rightarrow \infty} r'_n = \lim_{n \rightarrow \infty} s'_n = 0
$$
and $u'(r'_n) = v'(s'_n) = 0$ for all $n \in \Nat$ with $u',v'$ being the derivatives of $u$ and $v$.
\end{itemize}
\end{claim}
\begin{clproof}
W.l.o.g. we may assume that $\{r_n\}_{n \in \Nat}$ and $\{s_n\}_{n \in \Nat}$ are monotone. Then, since $g$ is continuous,
$$
g(0) = \lim_{z \rightarrow 0} g(z) = \lim_{z \rightarrow 0} u(z) + i\lim_{z \rightarrow 0} v(z) 
= \lim_{n \rightarrow \infty} u(r_n) + i \lim_{n \rightarrow \infty} v(s_n).
$$
Denote by $u_0,v_0$ the restrictions of $u$ and $v$ to the real segment $[0,1]$. Then $u_0,v_0$ are continuous and differentiable, and as $u(r_n) = v(s_n) = 0$ for all $n \in \Nat$ we see that\footnote{Recall that this is the Mean Value Theorem} there are $r'_n \in [r_{n+1},r_n]$ and $s'_n \in [s_{n+1},s_n]$ such that $u'_0(r'_n) = v'_0(s'_n) = 0$ for all $n\in \Nat$. 
\end{clproof}
To prove \eqref{eq:flambda_deriv_zero}, recall that by equation \eqref{eq:2003091748} and the assumption that $\lambda \neq 0$ we have $f_\lambda(1/k) = 0$ for all $k \in \Nat$. The definition of $f_\lambda$ implies $0 = f_\lambda(1/k) = u(1/k) + i v(1/k)$ so that the conditions of Claim~\ref{cl:2203091500} are satisfied. Applying Claim~\ref{cl:2203091500} inductively on the derivatives of $f_{\lambda}$ then yields  \eqref{eq:flambda_deriv_zero}.

\medskip

\noindent In the following, it will be convenient to partition $[m]$ into equivalence classes $N_0, \ldots, N_t$ such that $i,j \in N_\nu$ for some $\nu \in [0,t]$ iff $c_i = c_j$. For each $\nu$ let $\hat c_\nu$ be a representative of the $c_i$ values pertaining to $N_\nu$.
Since not all values $c_i$ are equal, we have $t \ge 1$. Further, condition \cond{S1C--B} implies that $c_1 = \ldots = c_{r-1}$ and we assume w.l.o.g. that $N_1 \supseteq [r-1]$. By definition, there is a $c_i = 0$ and we assume that $N_0$ is its corresponding equivalence class. 
Define, for each $\nu \in [0,t]$, a polynomial
\begin{equation}\label{eq:the_g_def}
g_\nu(X) = \sum_{i \in N_\nu} X^{2a_{i}}.
\end{equation}
By equations \eqref{eq:derivative_f} and \eqref{eq:flambda_deriv_zero} we obtain the following system of linear equations, for $l = 1, \ldots, t$
\begin{eqnarray*}
0 = f^{(l)}_\lambda(0) &=&  (\alpha \hat c_0)^l g_0(\lambda)  +  \ldots  + (\alpha \hat c_t)^l g_{t}(\lambda) \\
&=&  (\alpha \hat c_1)^l g_1(\lambda)  +  \ldots  + (\alpha \hat c_t)^l g_{t}(\lambda)
\end{eqnarray*}
The second equality follows from $\hat c_0 = 0$. The values $\hat c_1, \ldots, \hat c_t$ are pairwise different and non-zero by definition. Therefore $0 = f^{(l)}_\lambda(0)$ for $l= 1,\ldots, t$ forms a homogeneous system of linear equations with a Vandermonde determinant. Hence this system is non-singular, implying
$g_1(\lambda) = \ldots = g_t(\lambda) = 0$. As $\lambda$ is a root of $\Coo$ and $\Coo(X) = g_0(X) +  \ldots + g_t(X)$ we further infer $g_0(\lambda) =0$.

Note that our considerations so far are independent of the specific root $\lambda$  of $\Coo$. Therefore, we see that every irreducible polynomial $g$ which divides $\Coo$ divides all the $g_i$ as well (cf. Lemma~\ref{lem:59}(1)). Let $h_1,\ldots,h_{z}$ be the different irreducible divisors of $\Coo$ such that w.l.o.g. the leading coefficients of these $h_i$ are positive. 
For each $i \in [z]$ let $m_i$ be maximal such that $h^{m_i}_i(X)$ divides $g_\nu(X)$ for all $\nu \in [0,t]$. Defining $h(X) = h_1^{m_1}(X) \cdots h_z^{m_{z}}(X)$ we see that $h(X)$ divides each $g_{\nu}(X)$. Thus, for every $\nu \in [0,t]$, there are polynomials $f_\nu(X) = h(X)^{-1} \cdot g_\nu(X)$.
Hence
\begin{equation}\label{eq:0105091640}
\Coo  =  h(X) (f_0(X) + \ldots + f_t(X)) \quad \text{ and } \quad g_\nu(X) = h(X)\cdot f_\nu(X) \text{ for all } \nu \in [0,t].
\end{equation}
Observe that the degree of at least one polynomial $f_\nu$ is positive. To see this, recall that by $N_1 \supseteq [r-1]$ and $a_1 = \ldots = a_{r-1} = 0$ the polynomial $g_1$ has a non-zero constant term. Furthermore, by the fact that all $1$-entries of $A$ are contained in $1$-cells and $A$ contains exactly one such $1$-cell (by condition \cond{S1C--A}), for all $\nu \neq 1$ the polynomial $g_\nu$ does not have a constant term. Therefore $g_1$ and  $g_\nu$ differ by more than a constant factor, for all $\nu \neq 1$. With $h(X)$ being the greatest common divisor of the $g_\nu$ the existence of such an $f_\nu$ now follows.
The leading coefficients of all $f_\nu$ are positive, since that of $h(X)$ is and all coefficients of the $g_\nu$ are positive. Therefore,
\begin{equation}\label{eq:deg_of_f_sum}
\deg( f_0(X) + \ldots + f_t(X)) > 0.
\end{equation}

\begin{claim}\label{cl:first}
 $m(\lambda,j) = \mult(\lambda,h)$. 
\end{claim}
\begin{clproof}
It follows from the definition of $m(\lambda,j)$ in equation \eqref{eq:def_m_lam_j} that we have to prove that $\mult(\lrk,\Ck_{1j}) \ge \mult(\lambda,h)$ for all $k$ and that there is a $k$ such that equality holds.
Clearly, $\mult(\lrk,\Ck_{1j}) \ge \mult(\lambda,h) $ holds for all $k\ge 1$ because, by equation \eqref{eq:c_shape} and the definition of the $g_\nu$ in \eqref{eq:the_g_def}, we have
\begin{eqnarray*}
\Ck_{1j} &=& \phantom{h(X^k)}X^b\left( X^{\hat c_{0}} g_0(X^k) + \ldots + X^{\hat c_{t}} g_t(X^k)\right)\\
         &=& X^bh(X^k)\left( X^{\hat c_{0}} f_0(X^k) + \ldots + X^{\hat c_{t}} f_t(X^k)\right). 
\end{eqnarray*}
It remains to find a $k$ such that $\mult(\lrk,\Ck_{1j}) = \mult(\lambda,h)$.
Define a polynomial $\bar f (X,z) = X^{\hat c_{0}z} f_0(X) + \ldots +  X^{\hat c_{t}z} f_t(X)$, then
\begin{equation}\label{eq:ck_def_for_claim1}
\Ck_{1j} = X^b h(X^k) \cdot \bar f(X^k, k^{-1})
\end{equation}
Note that there is an $l \le t$ such that $\bar f(\lambda,\frac{l}{t!}) \neq 0$. To see this, assume the contrary and note that
$$
\bar f\left(\lambda,\frac{l}{t!}\right) = (\lambda^{\hat c_{0}/{t!}})^l f_0(\lambda) + \ldots +  (\lambda^{\hat c_{t}/{t!}})^l f_t(\lambda)
$$
Since the values $\hat c_0, \ldots, \hat c_t$ are pairwise different, this gives rise to an invertible system of linear equations with a Vandermonde determinant:
\begin{equation*}
\begin{array}{c c r r r r r r}
0 & = &  (\lambda^{\hat c_{0}/{t!}}) f_0(\lambda) &+& \ldots  & +  &(\lambda^{\hat c_{t}/{t!}}) f_t(\lambda)\\
0 & = &  (\lambda^{\hat c_{0}/{t!}})^2 f_0(\lambda) &+& \ldots  & +  &(\lambda^{\hat c_{t}/{t!}})^2 f_t(\lambda)\\
\vdots & & & & & & \\
0 & = &  (\lambda^{\hat c_{0}/{t!}})^t f_0(\lambda) &+& \ldots  & +  &(\lambda^{\hat c_{t}/{t!}})^t f_t(\lambda)\\
\end{array}  
\end{equation*}
Therefore, $f_1(\lambda) = \ldots = f_t(\lambda) =0$ which implies that all $f_{\nu}$ have a common root. By Lemma~\ref{lem:59}(1) they thus have a common non-trivial factor, in contradiction to the definition of $h$.

Hence we may fix an $l \le t$ such that $\bar f(\lambda,\frac{l}{t!}) \neq 0$. Define $k = t!/l$, then
$$
0 \neq \bar f(\lambda,\dfrac{l}{t!}) = \bar f(\lambda, k^{-1}) = \bar f((\lrk)^k, k^{-1}).
$$
That is, $\lrk$ is not a root of $\bar f (X^k,k^{-1})$ which implies $\mult(\lrk,\Ck_{1j}) = \mult(\lambda,h)$ by equation \eqref{eq:ck_def_for_claim1}.
\end{clproof}
Let us now prove statement \ref{lem:technical_core}(1b).
By equation \eqref{eq:c_shape}, we have
\begin{eqnarray*}
\Ck_{jj} & = & \phantom{h(X^k)}X^{2b}\left( X^{2\hat c_{0}} g_0(X^k) + \ldots + X^{2\hat c_{t}} g_t(X^k)\right) \\
         & = & X^{2b}h(X^k)\left( X^{2\hat c_{0}} f_0(X^k) + \ldots + X^{2\hat c_{t}} f_t(X^k)\right).
\end{eqnarray*}
Therefore, $\mult(\lrk,\Ck_{jj}) \ge \mult(\lrk,h(X^k)) = \mult(\lambda, h(X))$, which, by Claim~\ref{cl:first} implies $\mult(\lrk,\Ck_{jj}) \ge m(\lambda,j)$ --- just as claimed in (1b).

\bigskip 

\noindent To prove (1a) and (3) recall that, by equation \eqref{eq:0105091640}, 
$$
\Coo = h_1^{m_1}(X)\cdots h_z^{m_z}(X)\left(f_0(X) + \ldots + f_t(X)\right)
$$
which implies that every root $\lambda$ of $f_0(X) + \ldots + f_t(X)$ is a root of $\Coo$. And since the $h_i$ are the irreducible divisors of $\Coo$ such a $\lambda$ is therefore a root of one of the $h_i$ as well. By equation \eqref{eq:deg_of_f_sum} and the fact that all roots of $\Coo$ are non-zero, the polynomial $f_0(X) + \ldots + f_t(X)$ has at least one non-zero root, and hence 
$$\mult(\lambda,\Coo) > \mult(\lambda,h).$$
By equation \eqref{eq:2003091712} thus $\mult(\lrk, \Ck_{11}) = \mult(\lambda,\Coo) > \mult(\lambda,h)$ and (3) then follows by Claim~\ref{cl:first}. 
This also proves (1a) for the case that this root $\lambda$ of $\Coo$ is a root of $f_0(X) + \ldots + f_t(X)$. Otherwise $\mult(\lambda,\Coo) = \mult(\lambda,h)$ which proves (1a) by Claim~\ref{cl:first}.
This finishes the proof of the lemma.



\bibliographystyle{amsalpha}
\bibliography{bib}

\end{document}